\documentclass[a4paper,10pt]{article}

\usepackage[utf8]{inputenc}
\usepackage[margin=1.5cm]{geometry}
\usepackage{mathtools, amssymb, amsfonts, amsthm, mathrsfs, braket, mathbbol}
\usepackage{upgreek, stmaryrd, xparse, wasysym}
\usepackage{hyperref}
\usepackage[style=numeric, sortcites]{biblatex}
\usepackage{graphicx, tikz-cd, subcaption}
\usepackage[shortlabels, inline]{enumitem}
\usepackage{parskip}
\usepackage{comment}

\graphicspath{{./figs/}}
\title{Emergence of the Circle in a Statistical Model of Random Cubic Graphs}
\author{Christy Kelly$^1$\footnote{\url{ckk1@hw.ac.uk}}, Carlo Trugenberger$^2$, Fabio Biancalana$^1$\\
$^1$ School of Engineering and Physical Sciences, Heriot-Watt University, Edinburgh, UK.\\
$^2$ SwissScientific, Geneva, Switzerland.}
\date{October 2020}

\setlist[enumerate,1]{label={(\roman *)}}

\theoremstyle{plain}
\newtheorem{theorem}{Theorem}
\newtheorem{proposition}[theorem]{Proposition}
\newtheorem{lemma}[theorem]{Lemma}
\newtheorem{corollary}[theorem]{Corollary}
\newtheorem{example}[theorem]{Example}

\theoremstyle{definition}
\newtheorem{definition}[theorem]{Definition}

\theoremstyle{remark}
\newtheorem{remark}[theorem]{Remark}

\begin{document}
	\maketitle
	\begin{abstract}
		We consider a formal discretisation of Euclidean quantum gravity defined by a statistical model of random $3$-regular graphs and making using of the Ollivier curvature, a coarse analogue of the Ricci curvature. Numerical analysis shows that the Hausdorff and spectral dimensions of the model approach $1$ in the joint classical-thermodynamic limit and we argue that the scaling limit of the model is the circle of radius $r$, $S^1_r$. Given mild kinematic constraints, these claims can be proven with full mathematical rigour: speaking precisely, it may be shown that for $3$-regular graphs of girth at least $4$, any sequence of action minimising configurations converges in the sense of Gromov-Hausdorff to $S^1_r$. We also present strong evidence for the existence of a second-order phase transition through an analysis of finite size effects. This---essentially solvable---toy model of emergent one-dimensional geometry is meant as a controllable paradigm for the nonperturbative definition of random flat surfaces.
	\end{abstract}
	\section{Introduction}
	Discrete models of Euclidean quantum gravity based on dynamical triangulations or random tensors typically converge to either a crumpled phase of infinite Hausdorff dimension or to a phase of branched polymers in the large $N$ limit \cite{ADJ, Gurau_Invitation, Gurau_RandomTensors, GurauRyan_MelonsBranchedPolymers}. Indeed if we ignore the crumpled phase which is manifestly pathological, the branched polymer appears to be the principal universal scaling limit of random regularised geometries when the dimension $D\neq 2$; presumably, this makes it the main fixed point in some associated renormalisation group flow. The branched polymer itself may be characterised as the continuum limit of random discrete one-dimensional objects and for this reason is known as the continuum random tree in the mathematics literature \cite{Aldous_CRTI,Aldous_CRTII,Aldous_CRTIII}; it also has Hausdorff dimension 2 \cite{ADJ_SummingGenera} and spectral dimension $4/3$ \cite{JonssonWheater_SDBP} making it a highly fractal object. The situation is a little more involved for $D=2$; pure quantum gravity is equivalent to the quantum Liouville theory \cite{Polyakov_QuantumGeometryBosonicStrings,KPZ_Fractal2dQG,David_CFT,DistlerKawai_CFT2DQG} while the spectral and Hausdorff dimensions of these theories are $4$ and $2$ respectively \cite{AmbjornWatabiki_ScalingQG,AmbjornEtAl_SD2dQG}. When coupled to conformal matter of central charge $c$, however, branched polymers appear as a possible phase of $2D$-quantum gravity above the so-called \textit{$c=1$ barrier} and are interpreted as a kind of discrete manifestation of a tachyonic instability that leads to the breakdown of the continuum Liouville theory in this regime \cite{Cates_BranchedPolymer, KPZ_Fractal2dQG, David_CFT, ADJ}. As such, branched polymers are also generally regarded as a highly pathological model of $D$-dimensional Euclidean quantum spacetimes. 
	
	On the other hand, the formalism of causal dynamical triangulations---see \cite{AGJL,Loll_CDTReview} for a review---has shown that alternative scaling limits exist: two-dimensional causal dynamical triangulations, for instance, is equivalent to Horava-Lifschitz gravity in two dimensions and the spectral and Hausdorff dimensions both take on their natural values \cite{AmbjornEtAl_2DCDTHL,DurhuusJonssonWheater_SDCDT}. Again building causal structure into the dynamical triangulations model \textit{a priori} leads to substantial improvements in the geometric character of one of the phases in three dimensions \cite{AJL_3dCDT}, while in $4$-dimensions there is both a richer phase structure and at least one phase of reasonably geometric configurations \cite{AGJL,AmbjornEtAl_NewPhase, AmbjornEtAl_PhaseStructure}. From a practical point of view, the causal dynamical triangulations formalism represents a restriction of configuration space to a particular (non-spherical) topology in order to escape the universality classes of the Brownian map and the branched polymer; in particular, in the causal framework, there is a privileged foliation of spacetime which makes the quantum spacetimes obtained in (for instance) $2D$ causal dynamical triangulations homotopy cylinders with vanishing Euler characteristic. For compact surfaces, this is an essential condition for the existence of a Lorentzian metric \cite{ONeill_SRGeom}, so it is perhaps no surprise that standard Euclidean models which are dominated by contributions of genus zero give distinct non-Lorentzable results.
	
	Rather than restricting configuration space to a particular topology by fixing a foliation, an alternative strategy for escaping the universality class of branched polymers involves the \textit{expansion} of configuration space to include structures which have even less \textit{a priori} geometry than piecewise linear ones. One class of models that corresponds to such an expansion of configuration space is the class of \textit{network models} of gravitation which have attracted growing attention from researchers as of late \cite{BianconiRahmede_Flavour,Gibbs_Review,FarrFink_GeometryContinuum,Baird_EmergenceGeom,ChenPlotkin_GraphModels,Conrady_SpaceLowTemp,Wall_DiscretQG,Lombard_NetworkGravity, KonopkaMarkopoulouSmolin_QGraphity,KonopkaMarkopoulouSeverini_QGraphity, KellyEtAl,Trugenberger_CombQG, Trugenberger_QGasNetSO}. It seems desirable to study such expansions for two reasons: firstly, in order to fully clarify the role of causal structure in suppressing Euclidean pathologies and secondly, because models of emergent and latent network geometry are of independent intrinsic interest in network theory \cite{Bianconi_Challenges,Wu_EmergentComplex,BianconiRahmede_EmergentHyperbolic,Krioukov_Clustering, KrioukovEtAl_Hyperbolic, FarrFink_GeometryContinuum}. Stating the first point more ambitiously the aim is to present a model in which causal structure itself emerges dynamically at macroscopic scales \cite{Seiberg_EmergentSpacetime}. For sceptics, it is worth stressing that even if this more ambitious aim is not achieved, insights into the role of causal structure can be obtained from an analysis of Euclidean models. This strategy of generalisation of the structures under consideration by way of network models is the one adopted here.
	
	Of course any such expansion on phase space comes with new associated difficulties, such as reduced prospects for the recovery of an average approximate geometric structure as well as worsened control of the entropy. Both of these problems are manifest in discrete models with superexponential growth of the size of configuration space in terms of the number of spacetime points; for instance in causal set theory this superexponential growth is explicitly used to justify the nonlocality of the causal set action \cite{BenincasaDowker_ScalCurvCausSet}. It was also used as evidence against an early regularisation of random surfaces in terms of lattices imbedded in $\mathbb{Z}^D$ \cite{ADJ}. Nonetheless, from our perspective the issue of the superexponential growth of configuration space is perhaps somewhat moot; in standard random graph models \cite{Bollobas_RandomGraph,Janson_RandomGraph,AlbertBarabasi_StatMechCompNet,ParkNewman_StatMech} it has long been recognised that phase transitions occur in a manner that depends on threshold functions $\beta=\beta(N)$ that depend on the system size. From a physical perspective, the bare parameter under variation---which will in general have some relation to the relevant physical couplings of the theory---is thus to be regarded as scale dependent and the phenomenal coupling is a kind of renormalised parameter with the scale dependence factored out. Note that pursuing this line of thought in \cite{Trugenberger_CombQG,KellyEtAl} led to area-law scaling for the entropy of the model. This largely resolves the problem of excess entropy in principle; the problem of emergent geometry, however, remains to be addressed.
	
	In \cite{Trugenberger_CombQG, KellyEtAl} we considered a series of related network models regarded as combinatorial quantum gravity. There we showed that, subject to certain constraints, random $4$ and $6$-regular graphs spontaneously organised themselves into Ricci-flat graphs that were, on average, locally isomorphic to $\mathbb{Z}^2$ and $\mathbb{Z}^3$ respectively. {(Recall that a $k$-regular graph is one in which every vertex has $k$ neighbours.) This is in line with naive expectations; taking $\mathbb{Z}^D$ as the paradigm for a flat discrete $D$-manifold, we expect to fix the dimension by considering $2D$-regular graphs.} We also provided some evidence for the existence of a continuous phase transition in the form of a divergent correlation length plot. Nonetheless the numerical analysis presented in that work left many basic questions open including, crucially, the \textit{global} geometric consistency of observed classical configurations as evidenced, for instance, by spectral and Hausdorff dimension results. Further evidence for the existence and continuous nature of the phase transition is also desirable.
	
	Improvements in the code---partly following insights presented in \cite{Kelly_Exact}---have allowed some of these issues to be addressed, and in a future work we aim to present the case of flat surfaces; here, however, we address a more controllable model in which the configuration space consists of cubic ($3$-regular) graphs. In this model we have a classification of possible classical configurations and strong indications that in the $N\rightarrow \infty$ limit, the graphs in question converge to the circle $S^1_r$ of radius $r$ for some fixed $r>0$. This follows from spectral and Hausdorff dimension results, both of which come out at approximately $1$ for large graphs in the classical limit, as well as the observed nature of the classical configurations. We have further circumstantial evidence relating to the orientability of the observed classical configurations. However, we may do rather better in that an additional kinematic constraint---which appears to arise dynamically in the model we consider---allows us to show rigorously that a sequence of classical configurations converges to $S^1_r$ in the sense of Gromov-Hausdorff. We also present strong evidence for the continuous nature of the phase transition by an analysis of finite-size effects. We thus present a model with a continuous phase transition between a phase of random discrete spaces and emergent regular one-dimensional geometry.  {This is perhaps evidence for the existence of a nontrivial UV fixed point for gravity, which in the Euclidean context is equivalent to the random structure of spacetime near Planckian scales.}
	
	The basis for the combinatorial approach to quantum gravity introduced in \cite{Trugenberger_CombQG,KellyEtAl} is a rough analogue of the Ricci curvature introduced by Yann Ollivier \cite{Ollivier_RCMCMS, Ollivier_RCMS}. We discuss the Ollivier curvature more fully in appendix \ref{appendix: OllivierCurvature}. It is valid for arbitrary metric measure spaces \cite{Gromov_MetricStructures} and utilises synthetic notions of curvature that have arisen in optimal transport theory \cite{Villani_Topics, Villani_OptimalTransport}. It has also seen wide application in network theory, particularly as a measure of clustering and robustness \cite{Ni_RicIntTop, Sandhu_Cancer, Sandhu_Market, TannenbaumEtAl_Cancer, WangEtAl_DiffGeom, WangEtAl_Interference, WangEtAl_QUBO, WangEtAl_WirelessNetwork, WhiddenMatsen_SubtreeGraph, FarooqEtAl_Brain, SiaEtAl_CommunityDetection, JostLiu_RicciCurv}. Intuitively speaking the Ollivier curvature captures the notion that the average distance between two unit-mass balls in a positively curved space is less than the distance between their centres. The first hint that Ollivier curvature may play a role in the context of quantum gravity appears to be in \cite{Trugenberger_RHLW}, but other than our own work, we have recently seen Ollivier curvature appearing in the context of some much publicised network models of gravity by Wolfram and collaborators \cite{Gorard, Wolfram_Project}. From a slightly different perspective, Klitgaard and Loll have generalised the basic intuition of Ollivier curvature to define a notion of \textit{quantum Ricci curvature} as a quasi-local quantum observable for dynamically triangulated models that is valid at near the Planck scale \cite{KlitgaardLoll_QuantumRicciCurvature, KlitgaardLoll_ImplementingQuantumRicciCurvature, KlitgaardLoll_HowRound}. It is also worth noting that recently, it has been shown that optimal transport ideas have a role to play in classical gravity \cite{MondinoSuhr, McCann}. More generally, optimal transport ideas have seen fruitful application in noncommutative geometry \cite{AndreaMartinetti, MartinettiWallet} and in the renormalisation group flow of the nonlinear $\sigma$-model \cite{Carfora_WassGeomRicciFlow, CarforaFamiliari}.
	
	Note that the Ollivier curvature is not the only valid notion of discrete curvature. A closely related coarse notion of curvature has been introduced in the context of optimal transport theory by Sturm \cite{Sturm_GeomI,Sturm_GeomII} and Lott and Villani \cite{LottVillani}; see also \cite{Ohta_MeasureContraction}. Unfortunately the Sturm-Lott-Villani curvature does not generalise easily to discrete spaces because the $L^2$-Wasserstein space over such spaces lacks geodesics; Erbar and Maas have generalised the Sturm-Lott-Villani curvature to discrete spaces using an alternative metric on the space of probability distributions but its concrete properties are as yet unclear \cite{Maas_EntropyGradientFlow, ErbarMaas_Curvature}. A variety of alternative discrete notions of curvature also exist \cite{Najman_DiscreteCurvature, Forman, SreejithEtAl_Forman1, SreejithEtAl_Forman2, SreejithEtAl_Forman3, BakryEmery_DiffHyp, Bakry_AnalGeom, LinLuYau, Kamtue_Review}, and it would be interesting to see how far the results obtained here depend on the precise notion of discrete curvature adopted. Indeed, the \textit{Forman curvature} \cite{Forman,SreejithEtAl_Forman1,SreejithEtAl_Forman2,SreejithEtAl_Forman3}, an alternative notion of curvature defined for networks (strictly speaking regular $CW$-complexes), is the basis for a recent proposal of \textit{network gravity} \cite{Lombard_NetworkGravity} where quantum spacetimes are grown according to a stochastic model governed by the discrete Forman curvature.
	
	In this paper we adopt a now common perspective in mathematics viz. we essentially interpret emergent geometry in terms of Gromov-Hausdorff convergence \cite{Schramm_ScalingLimits}. (Of course, physically speaking we also require a continuous phase transition---or at least a divergent correlation length---in order to justify the construction of the scaling limit in the first place.) As such, it is worth briefly reflecting on the nature of Gromov-Hausdorff convergence and some of its physical ramifications. Given any metric space $(X,\mathcal{D})$ one can define a finite metric $\mathcal{D}_H$ on the space of compact subsets of $X$ called the \textit{Hausdorff metric}. The \textit{Gromov-Hausdorff distance} between two (isometry classes of) compact metric spaces $X$ and $Y$ is the infimum of the Hausdorff distance between the two spaces over all pairs of isometric imbeddings of the spaces $X$ and $Y$ into an arbitrary ambient space $Z$. The key point to note is that Gromov-Hausdorff convergence characterises the scaling limit invariantly precisely because one minimises over all possible backgrounds; this ensures that discrete manifolds, regarded as sequences of graphs which converge to a given manifold in the sense of Gromov-Hausdorff, are a generally covariant regularisation of their respective scaling limits. There are several technical caveats worth considering at greater length which we discuss in appendix \ref{appendix: GromovHausdorffDistance}.
	
	In this sense we have a reasonably rigorous interpretation of (Euclidean quantum) spacetimes as (sequences of) graphs. Gromov-Hausdorff convergence however is not sufficient; as stressed in \cite{Trugenberger_CombQG,KellyEtAl} we also need at least the convergence of some normalisation of the discrete Einstein-Hilbert action to its continuum counterpart in the Gromov-Hausdorff limit.  {This precise convergence question is somewhat subtle and has yet to be answered fully, essentially because the Ollivier curvature depends both on the metric and the measure theoretic structures of the spaces in question, where typically the graph measures are discrete while the requisite manifold measures are continuous.} For the models we have studied thus far this has been no obstacle since the curvature of classical configurations vanishes trivially and we wish to characterise Ricci-flat scaling limits by Ollivier-Ricci flat graphs.  {That is to say, \textit{here we make no claims about quantum gravity in general} considering only the generation of flat geometry from random degrees of freedom. We consider the precise relation of the model considered here to quantum gravity below.} 
	\section{Combinatorial Quantum Gravity}
	\subsection{The Model}
    We shall consider a discrete statistical model $(\Omega,\mathcal{A})$ defined schematically by the partition function
    \begin{align}\label{equation: StatisticalModel_PartitionFunction}
        \mathcal{Z}[\beta,\mathcal{A}]= \sum_{\omega\in \Omega}\exp(-\beta \mathcal{A}(\omega))
    \end{align}
    where $\Omega$ is some \textit{configuration space} consisting of graphs, $\mathcal{A}:\Omega\rightarrow \mathbb{R}$ is an \textit{action functional} and  {$\beta=\beta(\tilde{\beta},N)$ an \textit{a priori} $N$-dependent function which acts as a scale-dependent parameter for the model. We call $\beta$ the \textit{inverse temperature} of the system. Typically we will consider $\Omega$ consisting of random regular graphs at fixed $N$ subject to an additional constraint discussed subsequently, and investigate different values of $N$. Note that in principle when we say \textit{graph} we mean \textit{abstract graph}, but in practice abstract graphs are somewhat difficult to work with and simulations will use labelled graphs. The number of labellings of an abstract graph $\omega$ with $N$ vertices is $N!/|\text{Aut}(\omega)|$ where $\text{Aut}(\omega)$ is the automorphism group of $\omega$ and so at fixed $N$ we over-count each configuration in the partition function a fixed number of times unless global symmetries are present. Since a typical graph has no nontrivial automorphisms we have neglected this latter consideration. Note that since the Gromov-Hausdorff limit is characterised invariantly, there is in fact no strict need to consider abstract graphs.} The action $\mathcal{A}$ is a \textit{discrete Einstein-Hilbert action}, defined:
    \begin{equation}\label{equation: Action}
        \mathcal{A}(\omega)=-\sum_{e\in E(\omega)}\kappa_\omega(e)=-\frac{1}{2}\sum_{u\in V(\omega)}\sum_{v\in N_\omega(u)}\kappa_\omega(uv),
    \end{equation}
    where $N_\omega(u)$ denotes the set of neighbours of $u$ in $\omega$ and $\kappa_\omega(e)$ is the \textit{Ollivier curvature} of the edge $e\in E(\omega)$. A more complete presentation of the Ollivier curvature is given in appendix \ref{appendix: OllivierCurvature}, but for present purposes it is sufficient to recognise that the Ollivier curvature is a coarse version of the Ricci curvature in the following sense:  {consider two points $x$ and $y$ in a manifold $\mathcal{M}$ that are separated by a sufficiently small distance $\ell$, as well as the (unique) vector field $X$ parallel to the geodesic connecting $x$ and $y$.} Then we have:
    \begin{equation}
        \kappa_{\mathcal{M}}(e)\sim \ell^2\text{Ric}(X,X)+\mathcal{O}(\ell^3).
    \end{equation}
    In this way the Ollivier curvature represents a discretisation of the manifold Ricci curvature. The action \ref{equation: Action} thus corresponds to a discretisation of the (Euclidean) Einstein-Hilbert action as long as the edges incident to a vertex span the tangent space.
    
    The Ollivier curvature is discrete in that it takes values in the rational numbers $\mathbb{Q}$ and is also local in the following sense: let $\omega$ be a graph; for each edge $uv\in E(\omega)$, we have
    \begin{align}
        \kappa_\omega(uv)=\kappa_{C(uv)}(uv)
    \end{align}
    where $C(uv)\subseteq \omega$ is a subgraph of $\omega$ called a \textit{core neighbourhood} of $uv$. For our purposes it is sufficient to assume that $C(uv)$ is the induced subgraph of $\omega$ with the vertex set
    \begin{align}
        V(C(uv))=N_\omega(u)\cup N_\omega(v)\cup \pentagon(uv)
    \end{align}
    where $\pentagon(uv)$ is the set of non-neighbours of $u$ and $v$ that lie on a pentagon supported by the edge $uv$. Discreteness and locality are of course naively desirable properties for a quantised gravitational coupling.
    
    It turns out that for certain classes of graph the Ollivier curvature may be evaluated exactly. We shall need a little notation. 
    \begin{itemize}
        \item $\triangle_{uv}$ denotes the number of triangles supported on the edge $uv$.
        \item Consider the induced subgraph on the set $V(C(uv))/\set{uv}$. This will have $K$ connected components---each roughly corresponding to a cycle---labelled with the lower case letter $k$. We shall call these connnected components the \textit{components of the core neighbourhood}.
        \item $\square^k_w$ denotes the number of vertices in the $k$th component of the core neighbourhood that neighbour $w\in \set{u,v}$ such that the shortest cycle support by $uv$ that they lie on is a square.
        \item $\pentagon^k_w$ denotes the same as $\square^k_w$ for $w\in \set{u,v}$ except the shortest cycle is a pentagon instead of a square.
    \end{itemize}
    Using this notation we have the following expression for the Ollivier curvature in cubic graphs \cite{Kelly_Exact}:
    \begin{align}\label{equation: OllivCurvCubic}
          \kappa_\omega(uv)=\frac{1}{3}\triangle_{uv}-\frac{1}{3}\left[1-\triangle_{uv}-\sum_k\square^k_u\land \square^k_v\right]_+-\frac{1}{3}\left[1-\triangle_{uv}-\sum_k(\square^k_u+\pentagon^k_u)\land (\square^k_v+\pentagon^k_v)\right]_+
    \end{align}
    for each edge $uv\in E(\omega)$ where the sum over $k$ runs over the components of the core neighbourhood, $a\land b\coloneqq \inf\set{a,b}$ and $[a]_+\coloneqq \max(a,0)$ for any $a\in \mathbb{R}$. Note that the main property of $3$-regular graphs that allows this expression to be derived is the severe restriction on possible core neighbourhoods imposed by regularity of low degree. 
    
    We will not, in general, study the full configuration space of $3$-regular graphs, instead imposing additional kinematic constraints on $\Omega$. One particularly attractive constraint, derived in \cite{KellyEtAl} is the so-called \textit{independent short cycle condition}: we say that an edge $uv$ has \textit{independent short cycles} iff any two short (length less than $5$) cycles supported on the edge share no other edges. Graphs satisfying this condition admit an exact expression, independently of any additional constraints on the core neighbourhoods due to regularity. Furthermore the quantities appearing in the exact expression have an unambiguous interpretation in terms of numbers of short cycles. In particular the constraint ensures that
    \begin{align}
        \square^k_u=\square^k_v && \pentagon^k_u=\pentagon^k_v
    \end{align}
    for all $k$. Thus
    \begin{align}
        \sum_k\square^k_u\land \square^k_v=\square_{uv} && \sum_k(\square^k_u+\pentagon^k_u)\land (\square^k_v+\pentagon^k_v)=\square_{uv}+\pentagon_{uv}
    \end{align}
    where $\square_{uv}$ and $\pentagon_{uv}$ are the number of squares and pentagons supported on the edge $uv$ respectively. For $3$-regular graphs, we may thus express the curvature of an edge as
    \begin{align}\label{equation: OllivCurv_ISC}
        \kappa_\omega(uv)=\frac{1}{3}\triangle_{uv}-\frac{1}{3}\left[1-\triangle_{uv}-\square_{uv}\right]_+-\frac{1}{3}\left[1-\triangle_{uv}-\square_{uv}-\pentagon_{uv}\right]_+.
    \end{align}
    This expression for the curvature allows us to rewrite the action. Defining
    \begin{align}
        P=\set{uv\in E(\omega):\triangle_{uv}+\square_{uv}>1} && Q=\set{uv\in E(\omega):\triangle_{uv}+\square_{uv}+\pentagon_{uv}>1},
    \end{align}
    we have:
    \begin{subequations}
    \begin{align}
        \mathcal{A}(\omega)&=\mathcal{A}_{MF}+\mathcal{A}_{P}+\mathcal{A}_{Q}\\
        \mathcal{A}_{MF}&=N-3\triangle_\omega -\frac{8}{3}\square_\omega -\frac{5}{3}\pentagon_\omega\\
        \mathcal{A}_P&=-\frac{1}{3}\sum_{uv\in P}(1-\triangle_{uv}-\square_{uv})\\
        \mathcal{A}_Q&=-\frac{1}{3}\sum_{uv\in Q}(1-\triangle_{uv}-\square_{uv}-\pentagon_{uv}),
    \end{align}
    \end{subequations}
    where $\triangle_{\omega}$, $\square_\omega$ and $\pentagon_\omega$ denote the total numbers of triangles, squares and pentagons in the graph $\omega$ respectively. $\mathcal{A}_{MF}$ is determined by global quantities and thus represents a kind of mean field contribution to the action.
    
    The precise significance of the independent short cycle condition is not entirely clear. In \cite{KellyEtAl} it was a kind of `integrability' constraint insofar as it allows one to write the action rather explicitly. It also functions similarly to standard hard core constraints in statistical mechanics and prevents short cycles from `condensing' on an edge, though to a certain extent this is also guaranteed by regularity. For $3$-regular graphs, the main utility of the hard core condition appears to be the dynamical suppression of triangles it entails which we will see is an important requirement for the emergence of geometric structure in the model, though it appears that a somewhat weaker constraint is sufficient for this purpose. 
    
    Having expressions of the form \ref{equation: OllivCurvCubic} and \ref{equation: OllivCurv_ISC} for the Ollivier curvature permits the efficient running of simulations studying statistical models $(\Omega(N),\mathcal{A})$, where $\mathcal{A}$ is the Einstein-Hilbert action above and $\Omega(N)$ is the class of cubic graphs on $N$ vertices. (Note that since each $\omega\in \Omega(N)$ is regular with odd degree, $N$ must be even.) We use elementary Monte Carlo techniques \cite{NewmanBarkema}, evolving the graphs at each step via edge switches: given the random edges $uv$ and $xy$ in a graph $\omega$ such that $u\nsim x$ and $v\nsim y$, we construct a new graph $\tilde{\omega}$ by breaking the edges $uv$ and $xy$ and forming the edges $ux$ and $vy$, subject to any additional constraints on the configuration space that we may wish to impose.
    \subsection{The Relation to Euclidean Quantum Gravity}
     {As we have tried to make clear in the introduction, the main purpose of this paper is not to study the problem of (Euclidean) quantum gravity proper but the generation of flat geometries in a model of random graphs. We believe the analogy between the action \ref{equation: Action} and the Einstein-Hilbert action are sufficient to call this model combinatorial quantum gravity, as we have done previously. At the same time it would be false to claim that we have no pretensions to addressing Euclidean quantum gravity and for the purpose of clarity it seems desirable to try and explain the present position of our approach in relation to ordinary Euclidean quantum gravity.}
    
     {Clearly our hope is that there is some Ollivier curvature based action---which we expect to look rather like the action in equation \ref{equation: Action}---that will act as a discrete regularisation of the Einstein-Hilbert action on graphs $\omega$ which are sufficiently close to a manifold $\mathcal{M}$ in the sense of Gromov-Hausdorff. As mentioned in the introduction, precise convergence results for the Ollivier curvature are not yet known, and until they are we are somewhat restricted in the exact analysis of quantum gravity in general. The exception is the Ricci-flat sector where the convergence problem is absent---since the discrete and continuous Einstein-Hilbert actions vanish trivially. Thus on this sector Gromov-Hausdorff convergence alone guarantees agreement of our model with ordinary Euclidean quantum gravity.}
    
     {What are the prospects of a precise convergence result? There are essentially two required steps: first we need to show that the Ollivier curvature is respected by Gromov-Hausorff limits, and secondly we need to specify a discrete action in terms of the Ollivier curvature that converges to the Einstein-Hilbert action. Recently there has been major progress with regards to the first step.} In his original work, Ollivier \cite{Ollivier_RCMCMS} demonstrated the stability of the Ollivier curvature under Gromov-Hausdorff limits, in the sense that given a sequence of metric-measure spaces $X_n\rightarrow X$ and pairs of points $(x_n,y_n)\rightarrow (x,y)$ with $(x_n,y_n)\in X_n\times X_n$ we have $\kappa(x_n,y_n)\rightarrow \kappa(x,y)$. The convergence $X_n\rightarrow X$ is, however, Gromov-Hausdorff convergence \textit{augmented} by additional assumptions controlling the Wasserstein distance between the push-forward measures of the points under isometric imbeddings. The difficulty, of course, is that it is these auxiliiary assumptions which represent the foremost challenge in showing the convergence of Ollivier curvature in general.  {Much of this challenge has been addressed in the recent paper \cite{Hoorn_Convergence} which shows} that there is pointwise convergence of the Ollivier curvature---suitably rescaled---in random geometric graphs in arbitrary Riemannian manifolds. This is the first rigorous demonstration of the convergence of a notion of network curvature to its Riemannian counterpart known to the authors. One interesting feature is the need for two distinct length-scales that that are both sent to $0$ in the continuum limit: one is the microscopic `edge' scale defined by the threshold for connecting points obtained by a Poisson point process and the other is an effective curvature scale governing the rescaling of the Ollivier curvature and the radius of the unit balls used for comparison in the random geometric graph. From a gravitational perspective, the kind of point-wise convergence described in \cite{Hoorn_Convergence} must be augmented by showing that the limit is in fact generally covariant; on the other hand it is perhaps a stronger requirement than is physically necessary since one expects only physically meaningful quantities to converge in general. Nevertheless this result considerably expands the potential of discrete-curvature quantum gravity models.
    
     {Assuming that there is indeed a precise convergence result, we may turn to some general problems of discrete approaches to quantum gravity. One basic question is whether the partition function \ref{equation: StatisticalModel_PartitionFunction} gives well-defined dynamics in the $N\rightarrow \infty$ limit. This of course requires energy-entropy balance, i.e. that $\beta(N)\mathcal{A}$ and $S$ have the same $N$-dependence where $S$ is the entropy of the model. In practice we use the energy-entropy balance condition to fix the $N$-dependence of $\beta$; $\mathcal{A}$ grows as $N$ so we need to know the $N$-dependence of the entropy. This depends rather strongly on our choice of $\Omega$; in the present paper we consider $\Omega$ consisting of random regular graphs, and the number of such graphs on $N$-vertices is known. Specifically, we have \cite{Wormald_ModRRG} that the number of $d$-regular graphs on $N$-vertices is
    \begin{align}
        |\Omega_{N,d}|\sim \frac{(dN)!}{\left(\frac{1}{2}dN\right)!2^{\frac{1}{2}dN}(d!)^N}\exp\left(\frac{1-d^2}{4}+\mathcal{O}\left(\frac{1}{N}\right)\right).
    \end{align}
    Using the Stirling approximation, taking the logarithm and using its properties we thus get the following naive estimate for the entropy:
    \begin{align}
        S&\sim \frac{1}{2}dN\log N+\mathcal{O}(N)
    \end{align}
    We do not, however, simply consider random regular graphs, but instead random regular graphs satisfying an additional hard-core constraint. This constraint will have the effect of reducing the number of configurations and \textit{may} lead to corrections in the entropy, though these---if they exist---are hard to compute. Indeed, below we find numerically that for $N$ large, $\beta$ is in fact constant with $N$ suggesting that the independent short cycle condition leads to a logarithmic correction to the entropy, at least in the case of $3$-regular graphs.  We do not believe that this constant growth of $\beta$ is a generic feature of the model; instead we see it as a consequence of considering cubic graphs which, as we argue below, correspond to one-dimensional geometries and consequently a severely restricted set of action minimising configurations.}
    
    {Conceptually the issue of the $N$-dependence of $\beta$ is an expression of a well-known issue with local actions in discrete quantum gravity \cite{BenincasaDowker_ScalCurvCausSet} and can perhaps be interpreted as an expression of the nonlocality of the dimensionless action $\beta\mathcal{A}$. From a network theoretic perspective it is an expression of the infinite dimensionality of network phase transitions. To see how these perspectives relate, consider the standard square-lattice Ising model in $D$-dimensions. Such a model is finite-dimensional because the number of possible local interactions experienced by any bulk spin is fixed regardless of the system size. The average behaviour of these spins then only depends on the value of $\beta$, with the same average effect arising from the same value of $\beta$ for a bulk system whatever the system size. Since it is these local interactions which determine the presence or absence of long-range order in the system we may control the phase of the system with a $\beta$-independent of the system size. In a graph, where local interactions are modelled by edges, the number and type of possible interactions depends on the system size and so the parameter $\beta$ can only have the same average effect on a vertex as the system grows if $\beta$ also grows with the system size to compensate for the additional possible interactions.}
    
    {Of course, like any discrete approach to quantum gravity, we must decide whether discreteness is fundamental as in causal set theory, or simply a regularisation technique that may be removed as a cut-off is removed in line with the asymptotic safety scenario. Both points of view require the recovery of a more or less geometric scaling limit, which as argued in the introduction we interpret in terms of Gromov-Hausdorff convergence. We adhere to the---more conservative---latter attitude which further demands a continuous phase transition; the main result of this paper is that both of these aims can be achieved by our model in the Ricci flat sector, i.e. precisely where our model potentially agrees with quantum gravity.} 
    
    {Finally, let us briefly comment on some defects of the Euclidean approach to quantum gravity. One major well-known problem is that the continuum Euclidean action is not positive definite for $D>2$ since one may choose a metric with conformal mode undergoing arbitrarily fast variations. Such a problem is of course immediately removed upon discretisation but may reappear in the $N\rightarrow \infty$ limit. In the model discussed here---which recall is not quantum gravity---this problem does not arise firstly because the Ollivier curvature is bounded between $-2$ and $1$ for unweighted graphs and secondly because the independent short cycle condition effectively excludes positive curvature (negative action) geometries. More generally, as described in sections 1.8 and 1.9 of \cite{AGJL}, it is possible that the partition function is concentrated on configurations with bounded action near a non-Gaussian UV fixed point since the effective Euclidean action contains an entropic term coming from the number of configurations which share the same value of the action. Since our claims in this paper essentially amount to the existence of a UV fixed point, it seems quite possible that our approach may permit a similar escape from the problem of unboundedness.}
    
     {The biggest problem with our approach from the perspective of quantum gravity proper is the fundamentally Euclidean nature of the approach; in particular this refers to the absence of causal structure and difficulties related to giving sense to some notion of `Wick rotation'. A related issue is the unitarity of the resulting quantum theory. We have little concrete to say on these matters at present.}
    \section{The Classical Limit}\label{section: ClassicalLimit}
    Recall that the Gibbs distribution given a partition function \ref{equation: StatisticalModel_PartitionFunction} is
    \begin{align}
        p(\omega;\beta)=\frac{1}{\mathcal{Z}[\beta, \mathcal{A}]}\exp(-\beta \mathcal{A}(\omega)).
    \end{align}
    In the limit $\beta\rightarrow 0$, this becomes the uniform distribution on $\Omega$, leading to the standard model of random regular graphs; this model is characterised, in particular, by small world behaviour and sparse short cycles \cite{Wormald_ModRRG}. We shall call this limit the \textit{random phase}.
    
    The opposite limit $\beta\rightarrow\infty$ will be called the \textit{classical limit}. As is well known, in the classical limit the Gibbs distribution is concentrated about minima of the action, justifying the terminology. Heuristically, this conclusion is essentially an application of the Laplace method of approximation \cite{BenderOrszag}: suppose that $\omega_0\in \Omega$ is a minimum of the action, i.e. $\mathcal{A}(\omega_0)\leq \mathcal{A}(\omega)$ for all $\omega\in \Omega$. We shall denote the set of all such minima by $\Omega_0$. Then
    \begin{align}
        \frac{p(\omega;\beta)}{p(\omega_0;\beta)}=\exp(-\beta(\mathcal{A}(\omega)-\mathcal{A}(\omega_0))=\exp(-\beta|\mathcal{A}(\omega)-\mathcal{A}(\omega_0|)
    \end{align}
    for all $\omega\in \Omega$. That is to say, the ratio $p(\omega;\beta)/p(\omega_0;\beta)$ decays exponentially for $\omega\notin \Omega_0$ and contributions to the Gibbs distribution are (exponentially) concentrated about minima of the action as $\beta$ increases. In particular, taking the limit $\beta\rightarrow \infty$ sufficiently rapidly we see that:
    \begin{align}
        p(\omega;\infty)\coloneqq \lim_{\beta\rightarrow \infty}p(\omega;\beta)=\frac{1}{|\Omega_0|}\sum_{\omega_0\in \Omega_0}\delta_{\omega_0}(\omega)
    \end{align}
    where for any $\omega_0\in \Omega_0$, $\delta_{\omega_0}:\Omega\rightarrow [0,1]$ denotes the Dirac mass:
    \begin{align}
        \delta_{\omega_0}(\omega)=\left\{\begin{array}{rl}
            1, & \omega=\omega_0 \\
            0, & \omega\neq \omega_0
        \end{array}\right..
    \end{align}
    That is to say, the distribution $p(\omega;\infty)$ is supported on minima of the action as required; note that in the context of quantum gravity this suggests the identification $\beta \propto \hbar^{-1}$, since then the (semi)classical limit corresponds to the limit $\hbar\rightarrow 0$. We call $\Omega_0$ the \textit{classical phase} and call configurations $\omega_0\in \Omega_0$ \textit{classical} or \textit{tree-level} configurations.
    
    The purpose of this section is to study the classical limit of our model in the thermodynamic limit, i.e. as the number of points $N$ in the graphs goes to infinity. We find that as $N\rightarrow \infty$ the classical configurations converge to a limiting geometry described by a circle $S^1_r$ of some radius $r>0$. We begin with a classification of the possible classical configurations for given $N$ in section \ref{subsection: ClassicalConfigurations} before arguing that the limit of these configurations is $S^1_r$ in section \ref{subsection: LimitGeometry}. Together these conclusions constitute an argument that the classical limit of the present model is characterised by emergent one-dimensional geometric structure as long as the the continuum limit taken in section \ref{subsection: LimitGeometry} is justified.
    
    More precisely, in section \ref{subsection: ClassicalConfigurations} we show that we may easily define a model in which the classical configurations are either \textit{prism graphs} or \textit{M\"{o}bius ladders}---essentially discretisations of cylinders and M\"{o}bius strips respectively. In particular cubic graphs satisfying the independent short cycle condition have this property. In section \ref{subsection: LimitGeometry}, we then provide numerical evidence that the classical configurations are in fact one-dimensional as $N\rightarrow \infty$ by looking at the behaviour of the Hausdorff and spectral dimensions of the graphs observed. We obtain further incidental evidence for the one-dimensional nature of the limit by looking at the sequence $\mathfrak{o}(\omega_N)$ as $N\rightarrow \infty$, where $\omega_N$ is a classical configuration on $N$ vertices and
    \begin{align}
        \mathfrak{o}(\omega)=\square_\omega \mod 2
    \end{align}
    for any graph $\omega$. As shown below, this quantity appears to take on the value $1$ iff $\omega_N$ is not orientable (i.e. a M\"{o}bius ladder) and $0$ otherwise (i.e. for $\omega_N$ a prism graph). We see that $\mathfrak{o}(\omega_N)$ oscillates between $1$ and $0$ as we increase $N$ by $2$; recalling that a $D$-dimensional CW-complex is orientable iff its $D$-th homology group is $\mathbb{Z}$ and insofar as $\omega_N$ is simply a discretisation of either a M\"{o}bius strip or a cylinder, the divergence in $\mathfrak{o}(\omega_N)$ seems to indicate that the second (cellular) homology groups of the classical geometries corresponding to $\omega_N$ are unstable under the thermodynamic limit. This is of course to be expected if the limiting geometry has dimension one rather than two, and in this way the divergence provides circumstantial evidence for dimensional reduction in the thermodynamic limit. Finally we note that we may prove rigorously that a sequence of classical configurations in a configuration space with triangles excluded converges in the sense of Gromov-Hausdorff to $S^1_r$; a precise statement and proof of this result is given in appendix \ref{appendix: GromovHausdorffDistance}.
    \subsection{Classical Configurations and the Suppression of Triangles}\label{subsection: ClassicalConfigurations}
    
    \begin{figure}
        \centering
        \includegraphics[width=0.5\textwidth]{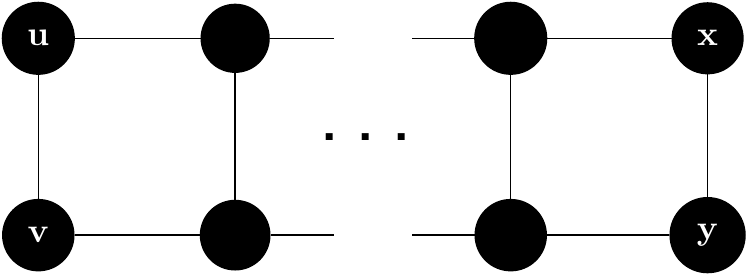}
        \caption{A chain of $n$ squares.}
        \label{figure: SquareChain}
    \end{figure}
    
    What are the classical configurations of the model? By the specification \ref{equation: Action}, action minimising configurations maximise the total curvature $\sum_{e\in E(\omega)}\kappa_\omega(e)$. In \cite{GraphCurvatureCalculator}, Cushing and collaborators have given a classification of positively curved cubic graphs which goes a long way towards a classification of $\Omega_0$. (Note that we say a graph is \textit{positively curved} iff all its edges have curvature at least zero.) Essentially we find that a positively curved graph is either a discrete cylinder or a discrete M\"{o}bius strip. More precisely:
    
    \begin{itemize}
        \item Consider a chain of $n$ squares, as in figure \ref{figure: SquareChain}. A \textit{prism graph} of length $n$, denoted $P_n$, is the graph obtained by identifying $u\cong x$ and $v\cong y$. 
        \item The \textit{M\"{o}bius ladder} of length $n$, denoted $M_n$, is obtained by gluing $u\cong y$ and $v\cong x$ in figure \ref{figure: SquareChain}.
    \end{itemize}
    
    With this terminology the classification of Cushing et al. \cite{GraphCurvatureCalculator} may be summarised as follows:
    \begin{itemize}
        \item If a $3$-regular graph $\omega$ has positive Ricci curvature for all vertices then it is either a prism graph $P_m$ for some $m\geq 3$ or a M\"{o}bius ladder $M_n$ for some $n\geq 2$.
        \item If $m=3$ the prism graph $P_m$ has edges $e\in E(P_m)$ with $\kappa_{P_m}(e)>0$. Otherwise $P_m$ is Ollivier-Ricci flat, i.e. $\kappa_{P_m}(e)=0$ for all $e\in E(P_m)$.
        \item Similarly if $n=2$ the M\"{o}bius ladder $M_n$ (which is also the complete graph on $4$-vertices $K_4$) has strictly positive curvature for each edge $e\in E(M_2)$. Otherwise $M_n$ is Ollivier-Ricci flat.
    \end{itemize}
    
    The key point to note is that for large $N$, a positively curved graph is Ollivier-Ricci flat and hence has total curvature zero. Moreover since the graphs in question have such an obvious geometric interpretation, it is tempting to see this as a model with a geometric classical phase. The issue, of course, is that one can imagine the situation where a graph has both positive and negative curvature edges with positive total curvature. Indeed such configurations may be constructed quite simply, and some examples are given in figure \ref{figure: TotalCurvatureTriangle}. Inspection of the expression \ref{equation: OllivCurvCubic} immediately shows that triangles necessarily appear in such configurations, and if we are to obtain a model with a classical phase $\Omega_0(N)=\Omega_0\cap \Omega_N=\set{P_N,M_N}$ where $\Omega_N$ denotes the class of $3$-regular graphs on $N$ vertices, it is sufficient to ensure that triangles are suppressed. This can of course be achieved by fiat---by restricting to configurations of girth greater than $3$ or bipartite graphs, for instance---but following \cite{KellyEtAl} we know that the suppression of triangles is a dynamical consequence of the model given the independent short cycle condition.
    
    To see this first note that $\mathcal{A}_{MF}$ is an extensive quantity, i.e. $\mathcal{A}_{MF}\sim N$. The scaling of $\mathcal{A}_P$ and $\mathcal{A}_Q$ depends on the behaviour of $|P|$ and $|Q|$ as $N\rightarrow \infty$. In \cite{KellyEtAl} it was argued that $\mathcal{A}_P$ and $\mathcal{A}_Q$ do play an important role as $\beta\rightarrow \infty$, but in the random phase $\beta\rightarrow 0$, short cycles are sparse and we may analyse the low $\beta$ dynamics by considering $\mathcal{A}_{MF}$ alone.
    
    \begin{figure}
        \centering
        \begin{subfigure}{0.45\textwidth}
			\centering
			\includegraphics[width=\textwidth]{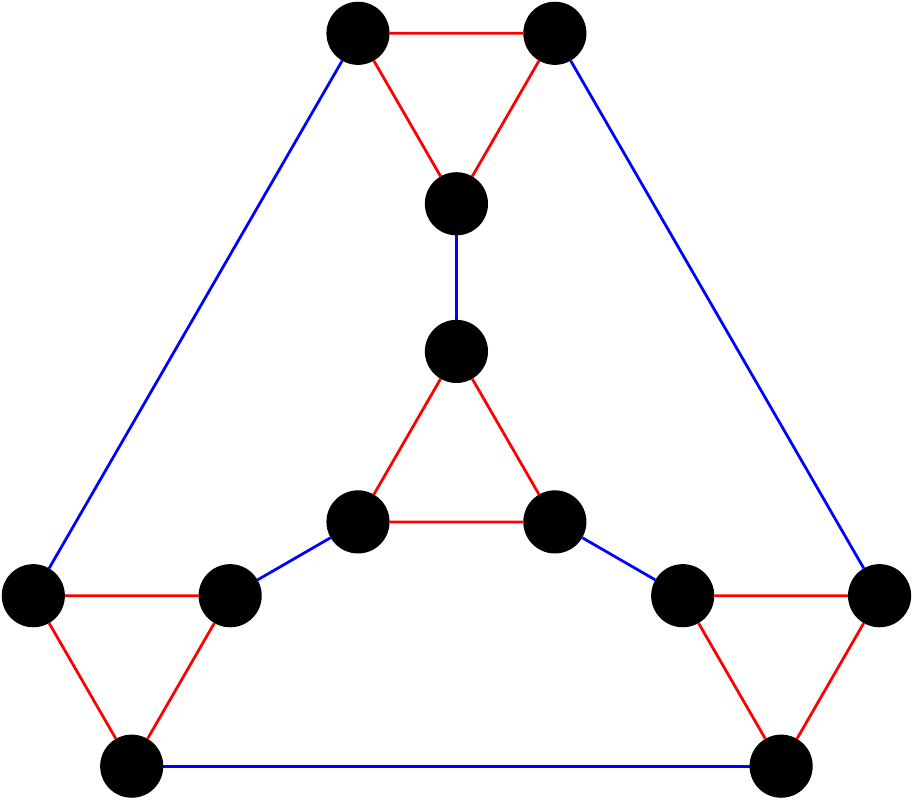}
			\subcaption{A cubic graph with vanishing total curvature and no squares.}
		\end{subfigure}
		\begin{subfigure}{0.45\textwidth}
			\centering
			\includegraphics[width=\textwidth]{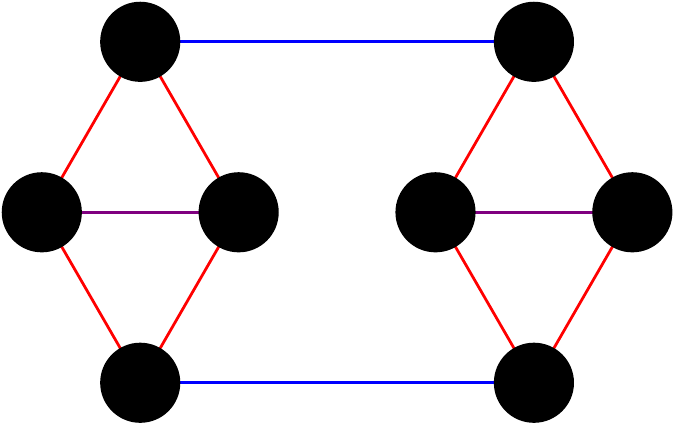}
			\subcaption{A cubic graph with total curvature $N/3>0$.}
		\end{subfigure}
        \caption{Cubic graphs with total curvature at least $0$ and with edges of strictly negative curvature. Both graphs are of a type that easily extends to larger numbers of vertices. Blue lines, red lines and purple lines have respectively curvatures of $-2/3$, $1/3$ and $2/3$.}
        \label{figure: TotalCurvatureTriangle}
    \end{figure}
    
    Naively we expect that $\mathcal{A}_{MF}$ is minimised by increasing the number of short cycles with the greatest effect coming from an increase in the number of triangles and the least effect coming from an increase in the number of pentagons. Indeed, an edge switch that converts a pentagon to a square leads to a reduction in the action of $1$, a pentagon to a triangle a reduction of $4/3$ and a switch from a square to a triangle a reduction of $1/3$. On these grounds we expect the number of triangles to \textit{increase} in the random phase. However there is an alternative way of looking at the independent short cycle constraint which leads to different conclusions. In particular the independent short cycle condition can be viewed as an excluded subgraph condition with the abstract graphs in figure \ref{figure: ISC} excluded. The key point to note is that the subgraphs \ref{figure: ISCa} and \ref{figure: ISCc} ensure that neither can an edge support two triangles nor can it support a triangle and a square. As such any triangle must share each of its edges with a pentagon and the gain in the action of $1/3$ that arises from the conversion of a triangle to a square is more than made up for by the potential loss of $3$ in the action arising from the conversion of each of the three pentagons to a square. Hence, due to `kinematic' constraints on the configuration space, triangles are strongly unfavoured in the local dynamics of the model in the random phase, and as such we expect them to be absent by the time the geometric phase is reached. On another note we expect pentagons to be similarly suppressed. Figure \ref{figure: Defect} indicates that these expectations are corroborated.
    
    \begin{figure}
	\centering
		\begin{subfigure}{0.3\textwidth}
			\centering
			\includegraphics[height=0.45\textwidth]{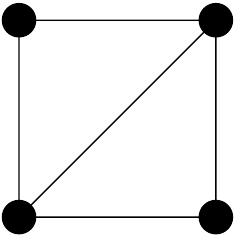}
			\subcaption{$(\triangle,\square)$}\label{figure: ISCa}
		\end{subfigure}
		\begin{subfigure}{0.3\textwidth}
			\centering
			\includegraphics[height=0.45\textwidth]{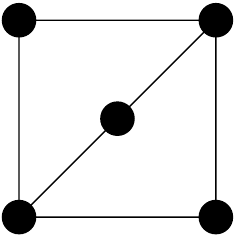}
			\subcaption{$(\square,\square)$}\label{figure: ISCb}
		\end{subfigure}
		\begin{subfigure}{0.3\textwidth}
			\centering
			\includegraphics[height=0.45\textwidth]{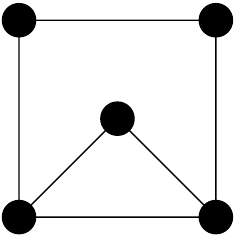}
			\subcaption{$(\triangle,\pentagon),\:(\square,\pentagon)$}\label{figure: ISCc}
		\end{subfigure}
		\begin{subfigure}{0.3\textwidth}
			\centering
			\includegraphics[height=0.45\textwidth]{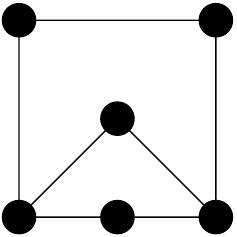}
			\subcaption{$(\square,\pentagon),\:(\pentagon,\pentagon)$}\label{figure: ISCd}
		\end{subfigure}
		\begin{subfigure}{0.3\textwidth}
			\centering
			\includegraphics[height=0.45\textwidth]{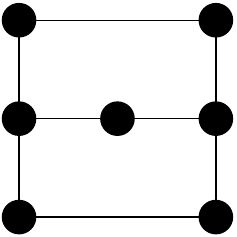}
			\subcaption{$(\pentagon,\pentagon)$}\label{figure: ISCe}
		\end{subfigure}
		\caption{Excluded subgraphs characterising graphs satisfying the independent short cycle condition. Pairs labelling the subgraphs indicate combinations of short cycles sharing more than a single edge excluded due to the subgraph in question.}\label{figure: ISC}
	\end{figure}
    
    Figure \ref{figure: ClassicalConfigs} shows examples of classical configurations observed following simulations run in a configuration space consisting of graphs satisfying the independent short cycle condition. In the simulations we looked at an exponential sweep of different values of $\beta$ from $-10\leq \log (\beta)\leq 10$ and allowed the graph to `thermalise' for $3N/2=|E(\omega)|$ sweeps at each value of $\beta$, where each sweep consists of $3N/2$ attempted Monte Carlo updates, i.e. one attempted update per edge on average. As desired one obtains both prism graphs and M\"{o}bius ladders indicating that we have an effective model in which the classical configurations are $\Omega_0(N)=\set{P_N,M_N}$ with high probability. In fact, for smaller configurations (each possible value of $N$, from $N=20$ to $N=50$) we have studied the appearance of prism graphs and M\"{o}bius ladders systematically and have not found a single exception to one of these configurations appearing in 100 runs for each graph size. Larger graph sizes have not been studied systematically, but again we have not observed a single configuration $\omega$ at $\beta=\exp(-10)$ that is neither a prism graph nor a M\"{o}bius ladder, over the course of several runs of each graph size from $N=100$ to $N=500$ at intervals of $50$ and then from $N=500$ to $N=1000$ at intervals of $100$.
    
    \begin{figure}
        \centering
        \includegraphics[width=0.7\textwidth]{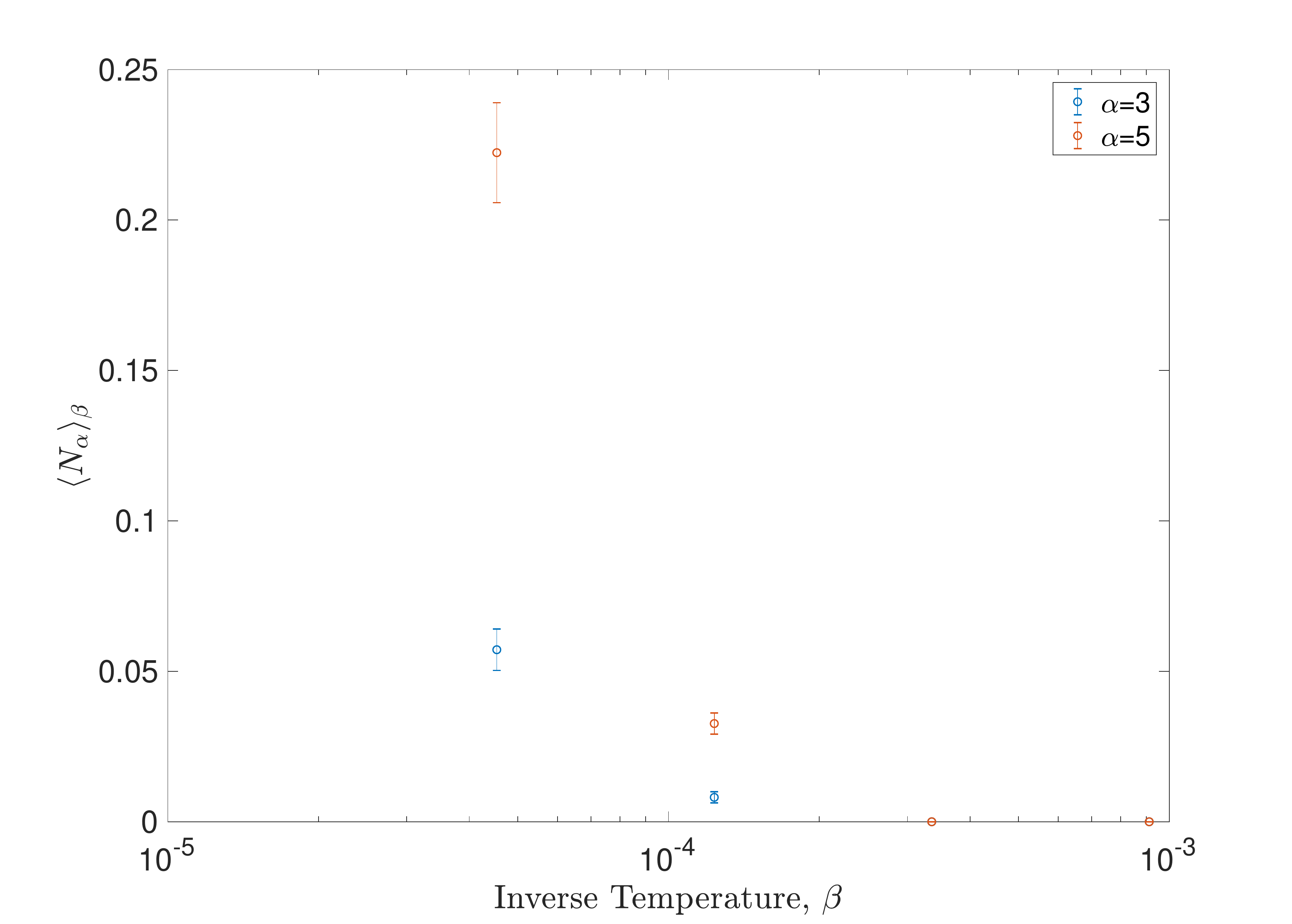}
        \caption{$\braket{N_\alpha}_\beta$ denotes the expected number of cycles of length $\alpha\in \set{3,5}$ at a given value of $\beta$ in a graph with $N=500$ vertices. As can be seen, odd short cycles are totally suppressed at small values of $\beta$.}
        \label{figure: Defect}
    \end{figure}
    
    \begin{figure}
        \centering
        \begin{subfigure}{0.45\textwidth}
			\centering
			\includegraphics[width=\textwidth]{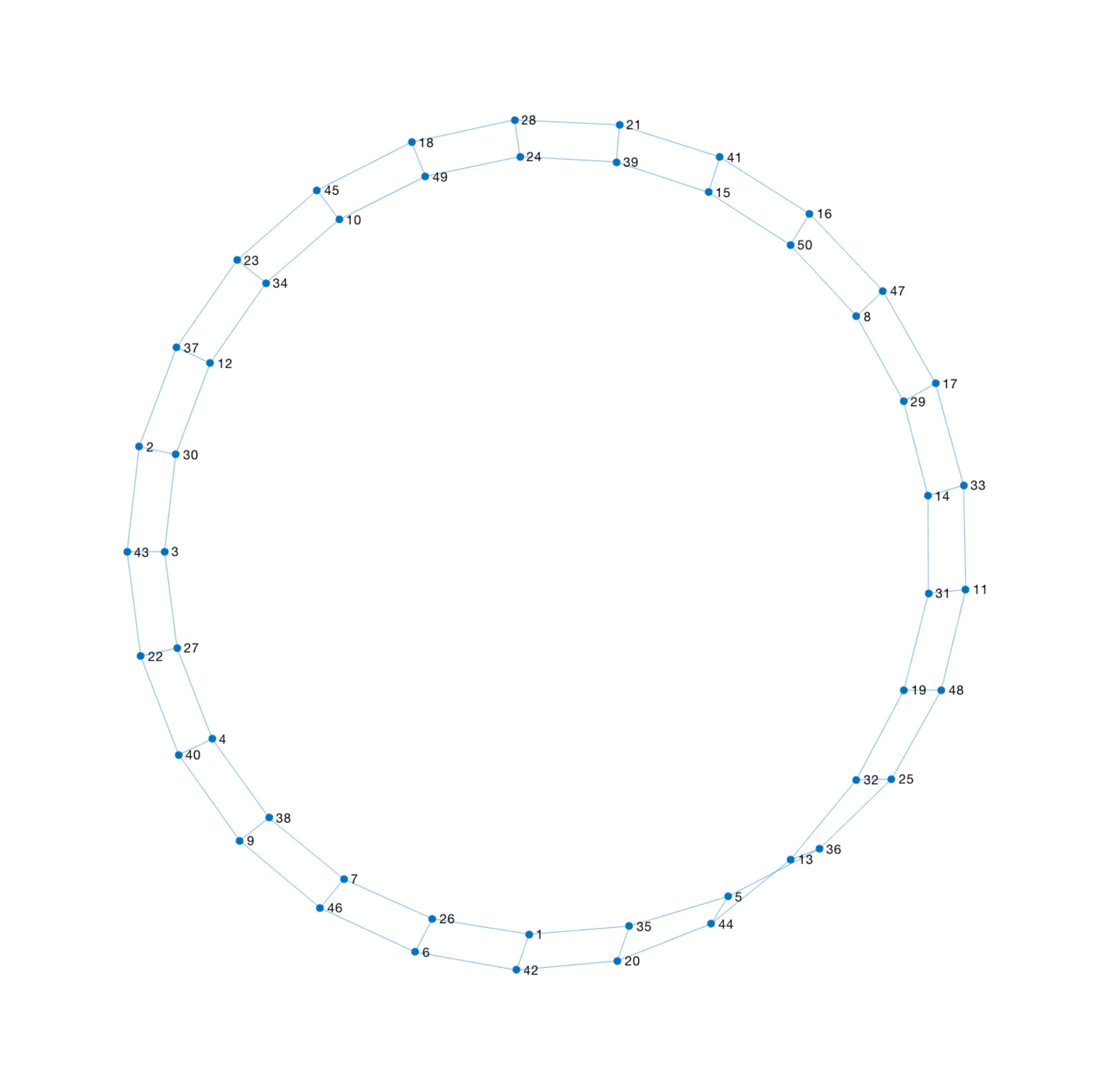}
			\subcaption{Mobi\"{u}s ladder on 50 vertices.}
		\end{subfigure}
		\begin{subfigure}{0.45\textwidth}
			\centering
			\includegraphics[width=\textwidth]{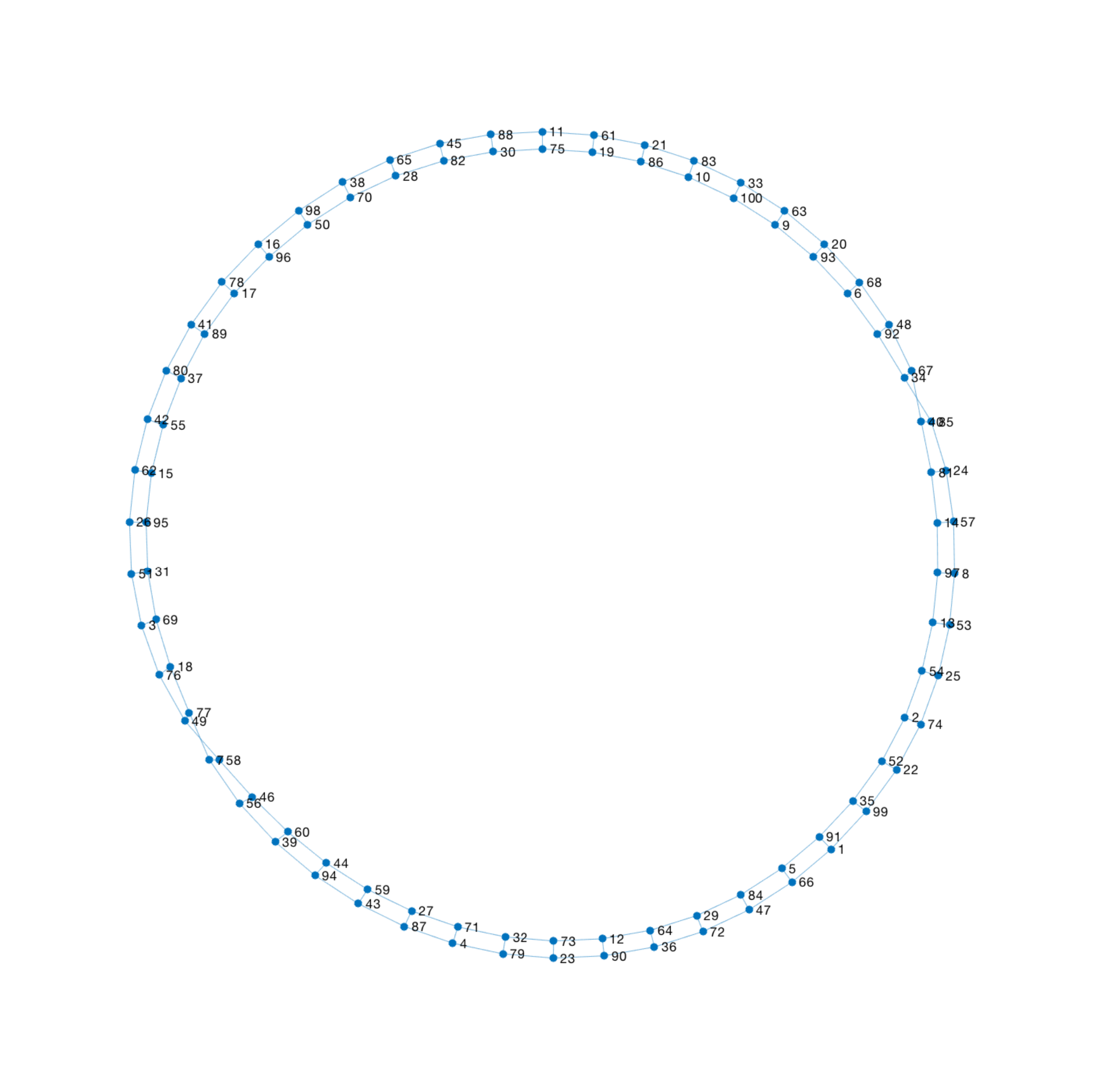}
			\subcaption{Prism graph on 100 vertices.}
		\end{subfigure}
        \caption{Observed classical configurations for $N=50$ and $N=100$ after running simulations in a configuration space consisting of $3$-regular graphs satisfying the hard core condition. As expected we have a prism graph and a M\"{o}bius ladder. Larger configurations also follow this pattern.}
        \label{figure: ClassicalConfigs}
    \end{figure}

    \subsection{Classical Configurations in the Thermodynamic Limit}\label{subsection: LimitGeometry}
    
    \begin{figure}
        \centering
        \includegraphics[width=0.7\textwidth]{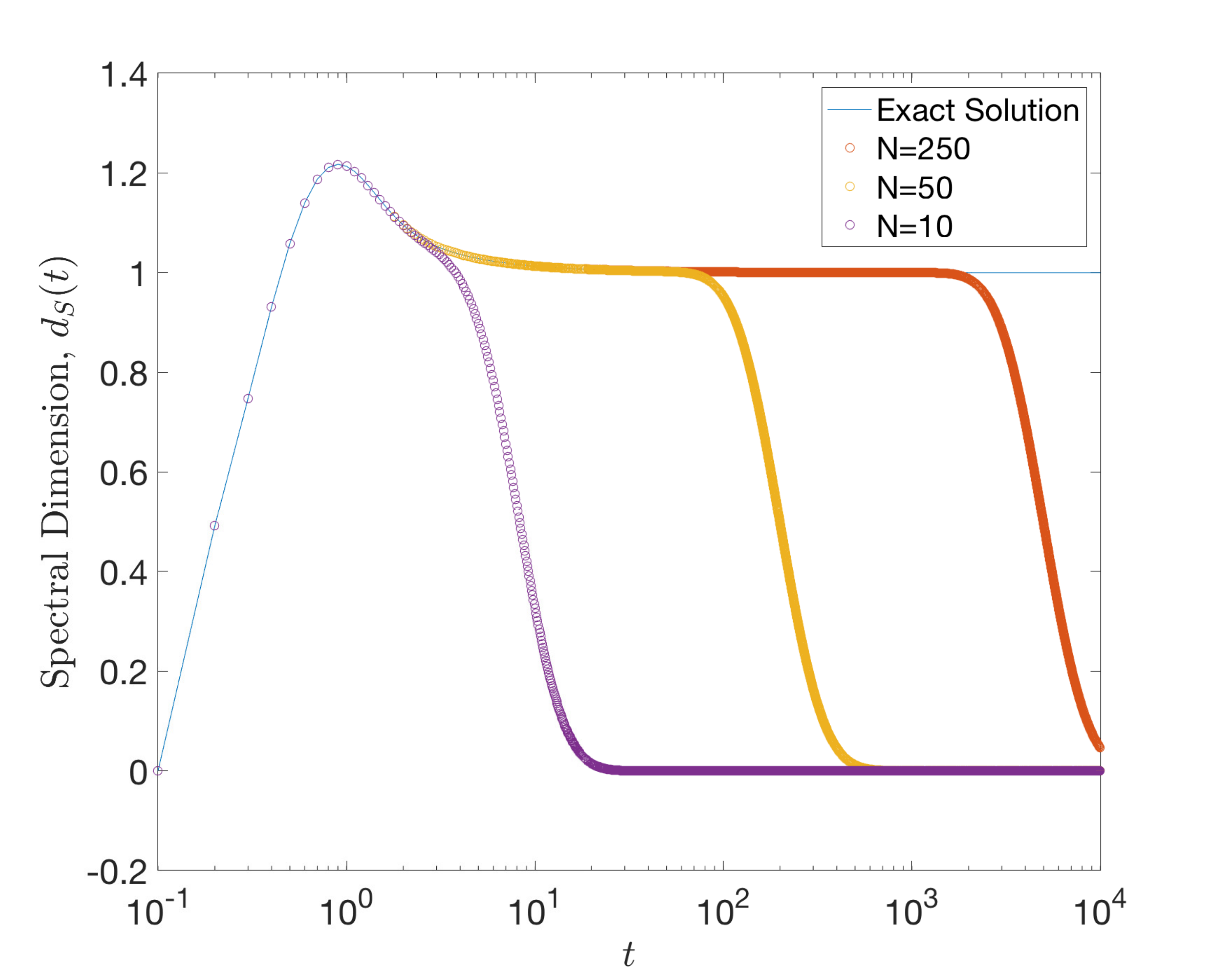}
        \caption{Spectral dimension for toroidal lattice graphs in $1$-dimension, i.e. square lattice graphs satisfying periodic boundary conditions. In $1$-dimension such graphs are simply cycles of length $N$.}
        \label{figure: SpectralDimension_Toroid}
    \end{figure}
    
    In the preceding section we provided strong numerical evidence that given a statistical model with configuration space $\Omega$ consisting of $3$-regular cubic graphs with independent short cycles, the classical phase $\Omega_0$ would consist overwhelmingly of prism graphs and M\"{o}bius ladders. This was to be expected: it can be proven that the classical phase consists of these graphs if we restrict to graphs without triangles while a strong heuristic argument and numerical evidence both point to this constraint effectively arising due to the low $\beta$ dynamics. We now consider the thermodynamic limit $N\rightarrow \infty$ of the classical configurations, given that $\omega_N\in \Omega_0(N)=\set{P_n,M_m}$, $N=2n$. We show that in the thermodynamic limit, the classical configurations are effectively one-dimensional and argue that the correct limiting geometry should in fact be $S^1$. It turns out that thus intuition can be rigorously formulated in terms of so-called \textit{Gromov-Hausdorff limits} as we prove in appendix \ref{appendix: GromovHausdorffDistance}. 
    
    \begin{figure}
        \centering
        \begin{subfigure}{0.7\textwidth}
			\centering
			\includegraphics[width=\textwidth]{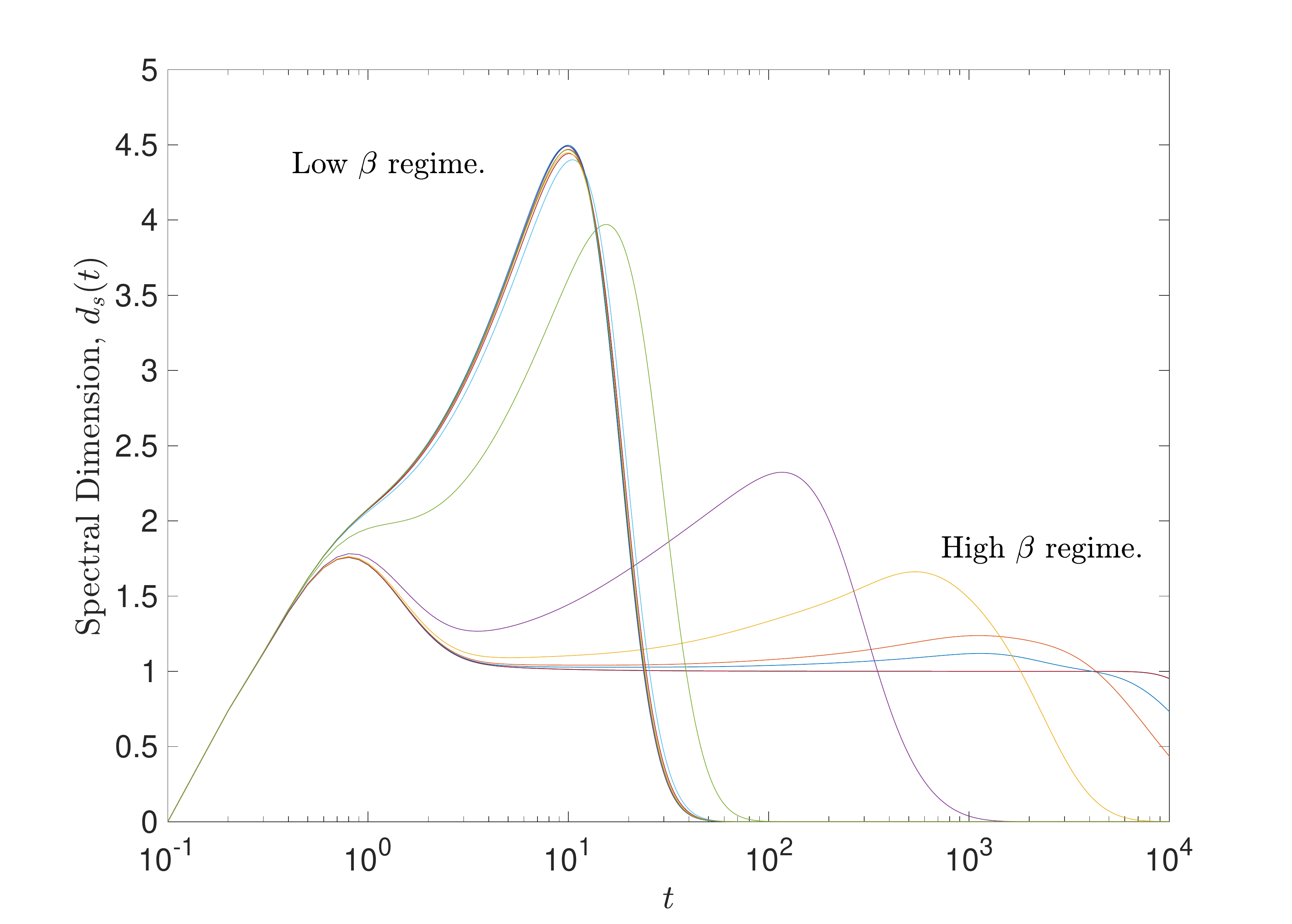}
			\subcaption{$d_s(t)$ for graphs with $N=1000$ at various values of $\beta$.}\label{subfigure: SpectralDimension_ClassicalConfiguration}
		\end{subfigure}
		\begin{subfigure}{0.7\textwidth}
			\centering
			\includegraphics[width=\textwidth]{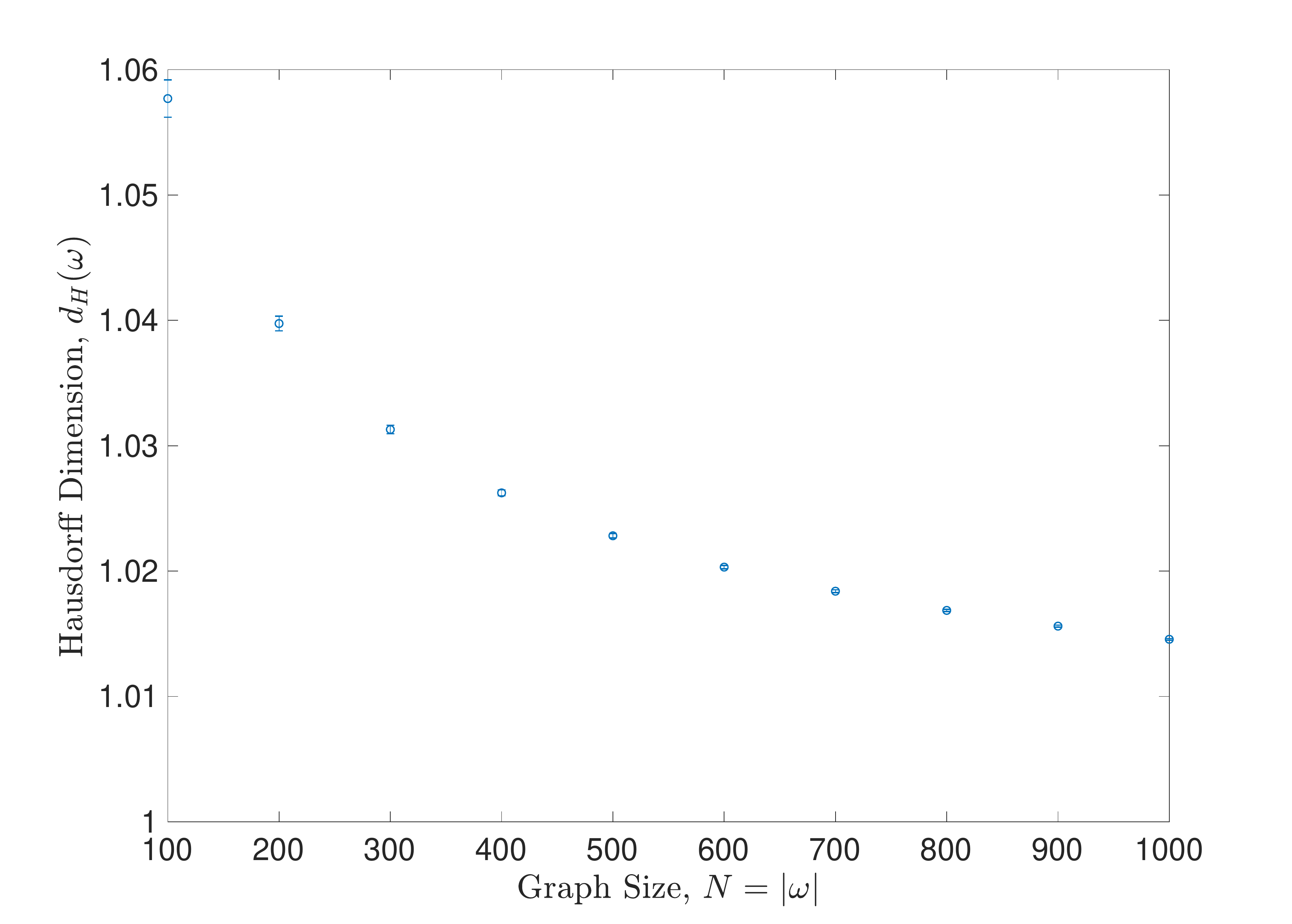}
			\subcaption{$d_H(\omega)$ for graphs $\omega$ of various sizes at $\beta=\exp(10)$.}\label{subfigure: HausdorffDimension}
		\end{subfigure}
        \caption{Spectral and Hausdorff dimension of $3$-regular graphs with independent short cycles. Both plots indicate that the dimension of observed classical configurations is approximately $1$.}
        \label{figure: Dimension}
    \end{figure}
    
    The first point to note is that dimensional evidence in the model already considered of $3$-regular graphs with independent short cycles suggests that the limiting geometry as $N\rightarrow \infty$ is one-dimensional. Dimensional data specifically refers to numerical data on the \textit{spectral} and \textit{Hausdorff} dimensions $d_s(\omega)$ and $d_H(\omega)$ of the graphs in question \cite{Carlip_Dimension,OritiCalcagniThurigen_SpectralDimension,Dunne_SpectralDim}; heuristically speaking, these quantities respectively define the dimension experienced by a particle diffusing in a space and the dimension governing the scaling of volume with distance in the space. Moreover both spectral and Hausdorff dimensions have become important quantities in quantum gravity: c.f. e.g. \cite{GurauRyan_MelonsBranchedPolymers,AmbjornWatabiki_ScalingQG,AmbjornEtAl_SD2dQG,ADJ_SummingGenera,JonssonWheater_SDBP,DurhuusJonssonWheater_SDCDT,OritiCalcagniThurigen_SpectralDimension, Carlip_Dimension,SteinhausThurigen_Emergence,EcksteinTrzesniewski_SD}. For our purposes, their present significance lies in the fact that we find $d_s(\omega)\sim d_H(\omega)\sim 1$ for classical configurations $\omega\in \Omega_0$ as $N\rightarrow \infty$. Figure \ref{figure: Dimension} shows the relevant plots. We shall briefly describe how the spectral and Hausdorff dimensions are defined here.
    
    The spectral dimension is defined for spaces equipped with a Laplacian $L$ and controls the scaling properties of the eigenvalues of the Laplacian $\lambda\in \sigma(L)$; note that $\sigma(L)$ is the spectrum of the Laplacian which is assumed to be a finite-dimensional operator. Indeed, defining the heat kernel trace as the function
    \begin{align}
        K(t)=\text{tr}\exp(-Lt) = \sum_{\lambda\in \sigma(L)}\exp(-\lambda t)
    \end{align}
    for $t$ in some open subset of $\mathbb{R}$, we may define the \textit{spectral dimension function}
    \begin{align}
        d_s(t)=-2\frac{\text{d} \log K(t)}{\text{d} \log t}.
    \end{align}
    The spectral dimension of the space on which the Laplacian is defined is then defined as $d_s(t)$ in some limit of $t$ or for some range of values of $t$. For Riemannian manifolds $\mathcal{M}$, for instance, we are interested in the small $t$ asymptotics where it can be shown that 
    \begin{align}
        K(t)\sim t^{-\frac{1}{2}D}
    \end{align}
    where $D=\dim(\mathcal{M})$. Thus if we define
    \begin{align}
        d_s(\mathcal{M})=\lim_{t\rightarrow 0}d_s(t)
    \end{align}
    we find $d_s(\mathcal{M})=D$. 
    
    For our purposes, we shall calculate the heat kernel trace $K(t)$ taking $L$ as the standard graph Laplacian: 
    \begin{align}
        L=D-A
    \end{align}
    where $D$ is an $N\times N$-dimensional diagonal matrix with the diagonal given by the graph degree sequence and $A$ is the adjacency matrix. $L$ is a finite symmetric matrix; its spectrum $\sigma(L)$ thus consists of the $N$-diagonal components of the diagonalisation of $L$. This ensures that $K(t)\rightarrow N$ as $t\rightarrow 0$ and $K(t)\rightarrow 0$ as $t\rightarrow \infty$; defining the spectral dimension as either the small or large $t$ limit of $d_s(t)$ thus gives $0$ trivially and for graphs such a prescription fails to specify an informative quantity of the space. Nonetheless it is possible to find a reasonable interpretation of the spectral dimension for discrete spaces by finding regions of the value $t$ for which the spectral dimension function $d_s(t)$ plateaus. The case of the torus is an instructive example: figure \ref{figure: SpectralDimension_Toroid} shows the characteristic behaviour of the spectral dimension function for our purposes viz. vanishing asymptotics, a peak due to discreteness effects and a plateau at the expected integer dimension. Similar behaviour is displayed by high $\beta$ configurations in figure \ref{subfigure: SpectralDimension_ClassicalConfiguration} indicating that $d_s(\omega)\approx1$ for classical configurations $\omega\in \Omega_0$.
    
    Alternatively, the Hausdorff dimension of a metric space equipped with some background measure effectively governs the volume growth of open balls in the space. The assumption is that for large $r$ the volume scales as
    \begin{align}
        \text{vol}(B_r(x))\sim r^{d_H}.
    \end{align}
    The Hausdorff dimension is then typically defined:
    \begin{align}
        d_H(\omega)=\lim_{r\rightarrow\infty}\frac{\log \text{vol}(B_r(x))}{\log r}
    \end{align}
    In a graph equipped with the shortest path metric, we take the background measure to be the counting measure and the volume of a ball of radius $r$ centred at a vertex $u$ is simply the number of points within a distance $r$ of $u$. For an unweighted graph this is simply all the points that can be reached from $u$ by a path of length at most $\lfloor r\rfloor$, the largest positive integer strictly less than $r$. If the graph is connected the diameter is finite and $B_r(x)=B_R(x)$ for all $r\geq R$ where $R$ is the diameter of the graph. Thus if we take the above definition $d_H(\omega)=0$ trivially for connected graphs and like the asymptotic definitions of the spectral dimension, the Hausdorff dimension so defined fails to be a useful measure of dimension for discrete spaces. As such we take the modified definition:
    \begin{align}\label{equation: HD}
        d_H(\omega)=\lim_{r\rightarrow \text{diam}(\omega)}\frac{\log |B_r(u)|}{\log r}
    \end{align}
    where $u\in V(\omega)$ and we assume that the diameter is large; the point is that if $|B_r(u)|=\alpha r^{d_H}$, then
    \begin{align}\label{equation: HDEstimate}
        \log |B_r(u)|=d_H\log r+\log(\alpha),
    \end{align}
    and $\log(\alpha)/\log(r)$ becomes negligible as $r\rightarrow R$ if the diameter $R$ is sufficiently large. In practice, we will estimate $d_H(\omega)$ according to equation \ref{equation: HDEstimate}, that is by taking the gradient of a plot of $\log |B_r(u)|$ against $\log r$, and then take the mean of this estimate over all vertices. This estimate is valid as long as $|B_r(u)|\sim r^{d_H}$ for most values of $r$ in the relevant range, i.e. as long as the plot of $\log |B_r(u)|$ against $\log r$ is approximately linear; in fact, for such graphs this approach extends the definition \ref{equation: HD} since it remains valid even if the diameter is small with respect to $\alpha$. Figure \ref{figure: HDEstimate} is a typical example of a plot of $\log|B_r(u)|$ against $\log r$, indicating that our estimation procedure for $d_H$ is indeed valid. Calculated values of this estimate for classical configurations of various sizes are given in \ref{subfigure: HausdorffDimension} and show that the Hausdorff dimension of these configurations is approximately $1$ for large $N$.
    
    \begin{figure}
        \centering
        \includegraphics[width=0.7\textwidth]{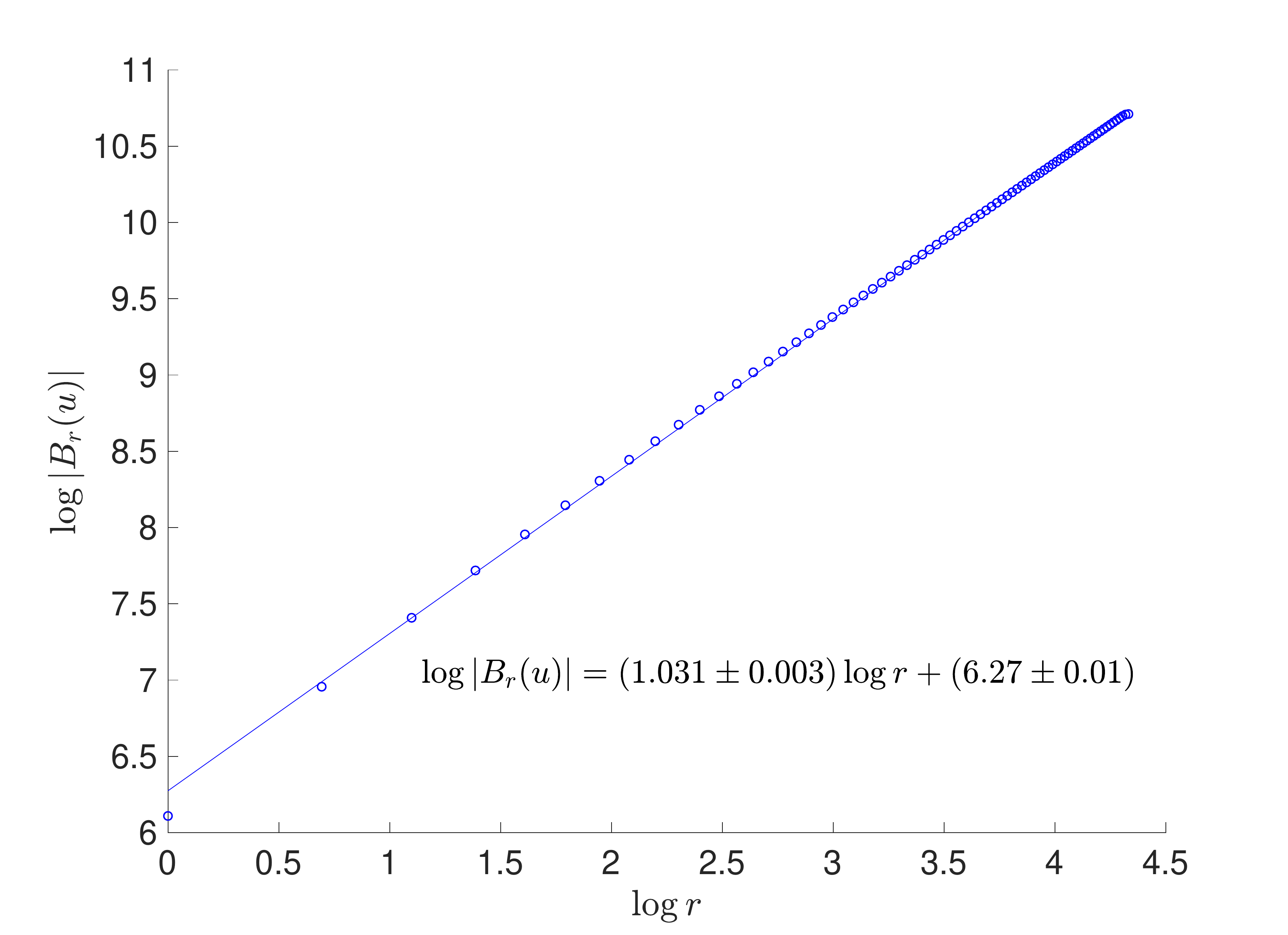}
        \caption{The volume growth of ball $B_r(u)$ centred at a vertex $u$ in a graph with $300$ vertices as the radius $r$ is varied. The $\log$-$\log$ relationship is approximately linear.}
        \label{figure: HDEstimate}
    \end{figure}
    
    As mentioned above we have additional---somewhat circumstantial---evidence for the one-dimensionality of the limiting geometry coming from the apparent instability of the second cellular homology group of our classical geometries in the thermodynamic limit. The key point is that the classical configurations are decidedly not choosing their topology at random. The situation is somewhat puzzling: Ollivier curvature is a local quantity and should not be able to distinguish between the M\"{o}bius ladder and the prism graph of the same size; as such there seems to be nothing in the action to distinguish the two possible classical configurations and from this perspective we perhaps expect M\"{o}bius ladders and prism graphs to appear in equal number. The numerical evidence, however, belies this expectation. We have the following anecdotal evidence:
    \begin{itemize}
        \item Every run of graphs of sizes $N=2n\in \set{100,200,300,400,500,600,700,800,900,1000}$ resulted in a prism graph. 
        \item Every run of graphs of sizes $N=2n\in \set{150,250,350,450}$ resulted in a M\"{o}bius ladder.
    \end{itemize}
    Precise numbers of runs in the above vary from a few to 15 runs for $N=100$; the data is not systematic but already makes the possibility of random choice between topologies highly implausible. We also see both prism graphs and M\"{o}bius ladders for fairly large graphs (about $400$ points) and the possibility of topology choice being a finite size effect seems slight. With these points in mind it is worth considering a more systematic study of the orientability of classical configurations obtained in simulations of small graphs: figure \ref{figure: Orientability} shows the average value of the quantity
    \begin{align}
        s(\omega)=\left\{\begin{array}{rl}
            1, & \omega\text{ orientable} \\
            -1, & \text{otherwise}
        \end{array}\right.\nonumber
    \end{align}
    where $\omega$ is an observed classical configuration over 100 runs of the code for graphs of size $N=2n\in \set{20,22,...,48,50}$. As can be seen the classical configuration appears to alternate between orientable and nonorientable results as two vertices---or one square---are added to the graph. In particular noting that $n=\square_\omega$ we see that the quantity 
    \begin{align}
        \mathfrak{o}(\omega)\coloneqq \square_\omega \mod 2
    \end{align}
    appears to contain crucial topological information about the tree-level configurations that actually appear in the simulations: $\omega$ is orientable if $\mathfrak{o}(\omega)=0$ and nonorientable otherwise. It would be curious to see if this can be explained. It is tempting to see some relation to the fact that the diagram
    \begin{equation}
        \begin{tikzcd}
        \Omega_0\arrow{r}{\mathfrak{o}}\arrow[swap]{d}{\pi} & \mathbb{Z}/2\mathbb{Z}\\
        \text{Vect}_1(S^1)\arrow[swap]{ur}{w_1} & 
        \end{tikzcd}
    \end{equation}
    commutes, where $\Omega_0$ is the set of classical configurations, $\text{Vect}_1(S^1)$ the set of line bundles over $S^1$ and $w_1$ the first Stieffel-Whitney class; it is well-known that, in particular, $w_1$ is a bijection and $\text{Vect}_1(S^1)$ consists of a cylinder and a M\"{o}bius band (the possible classical geometries observed in this model) so $\pi$ is to be interpreted as a mapping that sends prism graphs to cylinders and M\"{o}bius ladders to M\"{o}bius strips. 
    
    \begin{figure}
        \centering
        \includegraphics[width=0.7\textwidth]{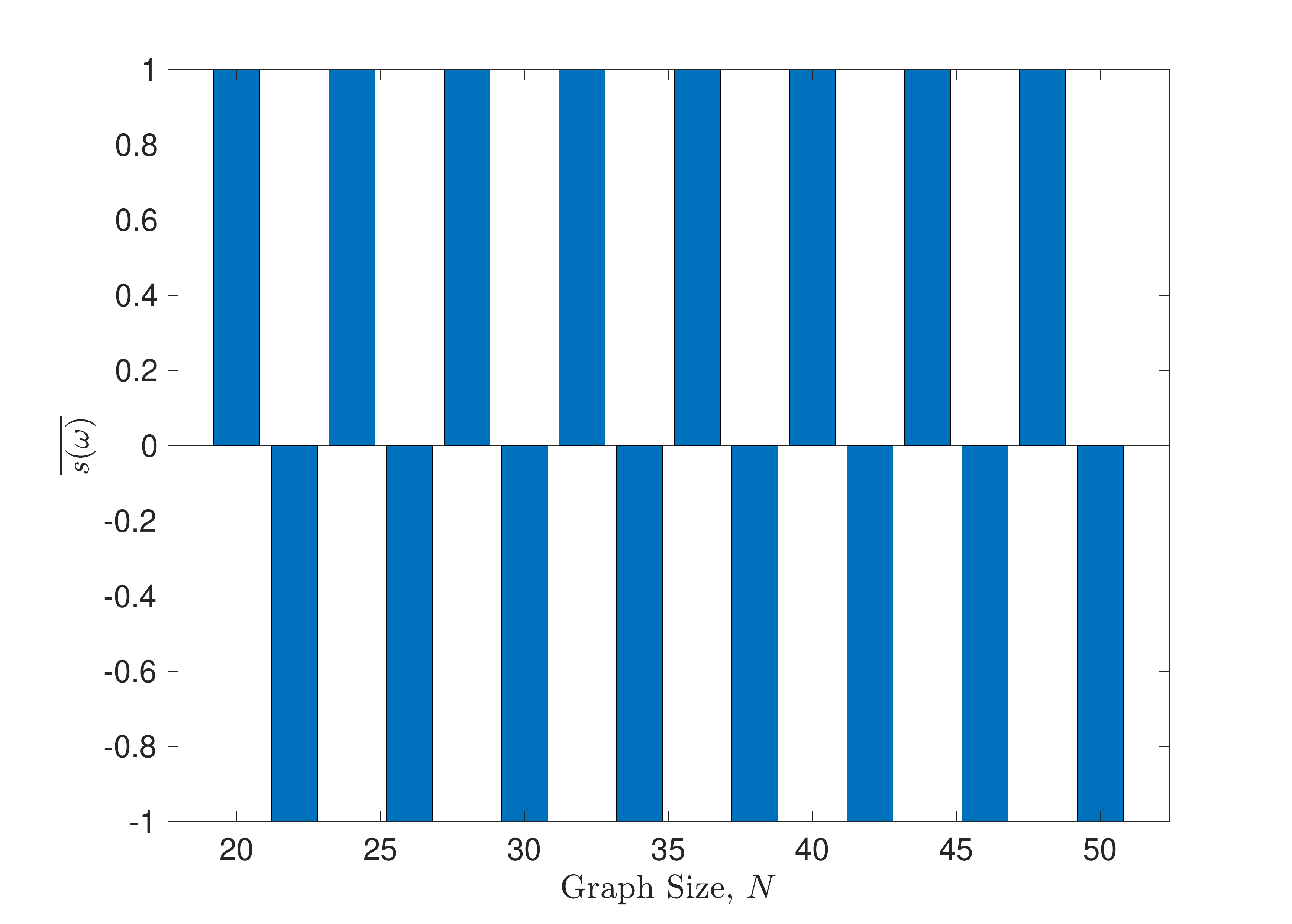}
        \caption{Mean value of $s(\omega)$ evaluated for observed classical configurations $\omega$, averaged over $100$ runs of the code for each shown graph size.}
        \label{figure: Orientability}
    \end{figure}
    
    Given, then, that our classical configurations are converging to a one-dimensional space, what is the space and in what sense are they converging to that space? The limiting geometry is presumably (hopefully) a connected manifold $\mathcal{M}$. Again, as is well-known this means we have (up to homeomorphism) four possibilities \cite{1manifold}:
    \begin{enumerate}
        \item if $\mathcal{M}$ is noncompact and without boundary then $\mathcal{M}=\mathbb{R}$.
        \item if $\mathcal{M}$ is noncompact and has a nonvoid boundary then $\mathcal{M}=\mathbb{R}_+$.
        \item if $\mathcal{M}$ is closed, i.e. compact and without boundary then $\mathcal{M}=S^1$.
        \item if $\mathcal{M}$ is compact with boundary then $\mathcal{M}=[0,1]$.
    \end{enumerate}
    Note that $1$-manifolds with boundary have a privileged set of points corresponding to the boundaries: $\set{0,1}=\partial [0,1]$ and $\set{0}=\partial \mathbb{R}_+$. If the limiting geometry were to be one of these manifolds with boundary, then, the classical configurations would also presumably have to be endowed with a privileged set of vertices that converge to boundary points. This cannot be the case insofar as we regard our classical configurations as abstract graphs and the limiting geometries are without boundary i.e. either $S^1$ or $\mathbb{R}$ respectively. 
    
    Naively, since the geometry in question is the limit of a discrete space it is rather more natural to assume that the limit is compact, but in fact it seems clear that the precise geometry obtained in the limit depends on the way that the limit is taken. Suppose we have a cylinder $[0,\ell]\times S^1_r$, where $\ell>0$ and $S^1_r$ is the circle of radius $r>0$. We may suppose that the cylinder is discretised in terms of a prism graph $P_n$ with the radius
    \begin{align}\label{equation: Radius}
        r=\frac{n\ell}{2\pi}.
    \end{align}
    Then we have effectively weighted each edge of the prism graph with an edge length $\ell$. \textit{A priori}, we have two natural choices: keep $\ell$ fixed and let $r$ diverge as $n\rightarrow \infty$, or alternatively keep $r$ fixed and let $\ell\rightarrow 0$ as $n$ increases. The former corresponds to the non-compact limit $[0,1]\times \mathbb{R}$ which is two-dimensional and ruled out by the arguments above. The latter, on the other hand, corresponds to a one-dimensional compact limit $S^1_r$; it is also heuristically more in line with the idea that $P_n$ is a lattice regularisation of the underlying geometry $[0,\ell]\times S^1_r$. In particular, insofar as $\ell$ is a lattice cutoff we wish to take $\ell\rightarrow 0$ as $n\rightarrow \infty$. Roughly speaking, the limiting geometry is simply defined as the limit of cylinders $[0,\ell]\times S^1_r$ as the width $\ell\rightarrow 0$, which is of course simply the circle $S^1_r$. Similar considerations hold for M\"{o}bius strips.
    
    These heuristic considerations on the limiting geometry have a rigorous formulation in terms of \textit{Gromov-Hausdorff limits}. The fundamental idea of `metric sociology' as Gromov termed it, is that there is a notion of distance between compact metric spaces, that gives us a precise notion of what it means for one compact metric space to converge to another. This is a particularly attractive framework for studying the problem of emergent geometry in a Euclidean context because every compact Riemannian manifold can be obtained as the Gromov-Hausdorff limit of some sequence of graphs \cite{BuragoBuragoIvanov_MetricGeometry}. We introduce the Gromov-Hausdorff distance in appendix \ref{appendix: GromovHausdorffDistance} and show that the classical configurations and geometries discussed here each converge to $S^1_r$ as $N\rightarrow \infty$.
    \section{Phase Structure}
    In the preceding section we investigated the thermodynamic limit of the classical configurations of our model and found that the statistical models under investigation gave rise to the limiting geometry $S^1$. In physical terms, we in fact took a joint thermodynamic-continuum limit
    \begin{align}
        N\rightarrow \infty && \ell\rightarrow 0
    \end{align}
    keeping the quantity $r=N \ell/2\pi$ constant, with the edge length $\ell$ a UV cutoff. The factor of $2\pi$ is simply a choice of normalisation which ensures that the constant quantity is the radius of the limiting circle and can be dropped without loss of generality. More generally, we are interested in taking a joint thermodynamic and continuum limit $N\rightarrow \infty$, $\ell(N)\rightarrow 0$, while preserving some fixed length scale $\ell_0$; noting that $N$ is essentially a volume measure, the correct generalisation of the expression \ref{equation: Radius} to arbitrary dimensions $D$ is:
    \begin{align}\label{equation: ThermoContinuum}
        \ell_0=\ell(N)N^{\frac{1}{D}}.
    \end{align}
    We expect $N$ to be large so the physical problem is the justification of the expression \ref{equation: ThermoContinuum}.
    
    Our acquaintance with ordinary statistical mechanics suggests that such an expression holds as $\beta\rightarrow \beta_c$ where $\beta_c$ is some finite number at which the system in question undergoes a continuous phase transition. In models of quantum geometry one expects quantum fluctuations about classical spacetime geometries in the limit $\beta\rightarrow \beta_c$ and the extent to which quantum geometries preserve the smooth structures characteristic of manifolds is not clear. Of course to rigorously obtain (smooth) manifold limits as observed classical configurations, we must assume that the expression holds exactly as $\beta\rightarrow\infty$. 
    
    Thus our aim is to investigate the phase structure of the statistical models defined above. In particular we are looking to find a second-order phase transition. Following \cite{Binder_PhaseTransition} we note that strong evidence for both the existence and the continuous nature of a transition may be derived by an analysis of finite-size effects. A recent example of this kind of study in a similar context is \cite{Glaser_FinSizeScale} which provides strong evidence that two-dimensional causal set theory experiences a first-order phase transition as an analytic continuation parameter $\beta$ is varied. In this section we conduct a similar analysis of finite size effects and provide evidence that our models evince a second-order phase transition.
    \subsection{\texorpdfstring{Existence of a Phase Transition and Estimating the Critical $\beta$}{Existence of a Phase Transition and Estimating the Critical Beta}}
    We shall study the phase structure of the theory via an examination of the quantities
    \begin{align}
        \braket{\mathcal{A}}_\beta=\frac{\partial (\beta \mathcal{F})}{\partial \beta } && C = \beta^2\braket{(\mathcal{A}-\braket{\mathcal{A}}_\beta)^2}_\beta =-\beta^2\frac{\partial^2(\beta \mathcal{F})}{\partial \beta^2}
    \end{align}
    where $C$ is called the \textit{specific heat} of the system. Note that $\beta \mathcal{F}=-\log \mathcal{Z}$. In a phase transition we expect to find some nonanalyticity in the free energy $\beta \mathcal{F}$ in the thermodynamic limit and in the Ehrenfest classification the order of the derivative in which a discontinuity is observed gives the order of the phase transition. That is to say in a first-order transition we expect to observe a discontinuity in $\beta \mathcal{A}$ and a delta-function peak in $C$, while in a second-order transition we expect $\beta \mathcal{A}$ to be continuous and the specific heat to contain the discontinuity. Since such symptoms of nonanalyticity only arise in the thermodynamic limit \cite{Goldenfeld_SM} the difficulty lies in the assessment of the phase transition in finite systems.
    
    \begin{figure}
        \centering
        \begin{subfigure}{0.7\textwidth}
        \centering
        \includegraphics[width=\textwidth]{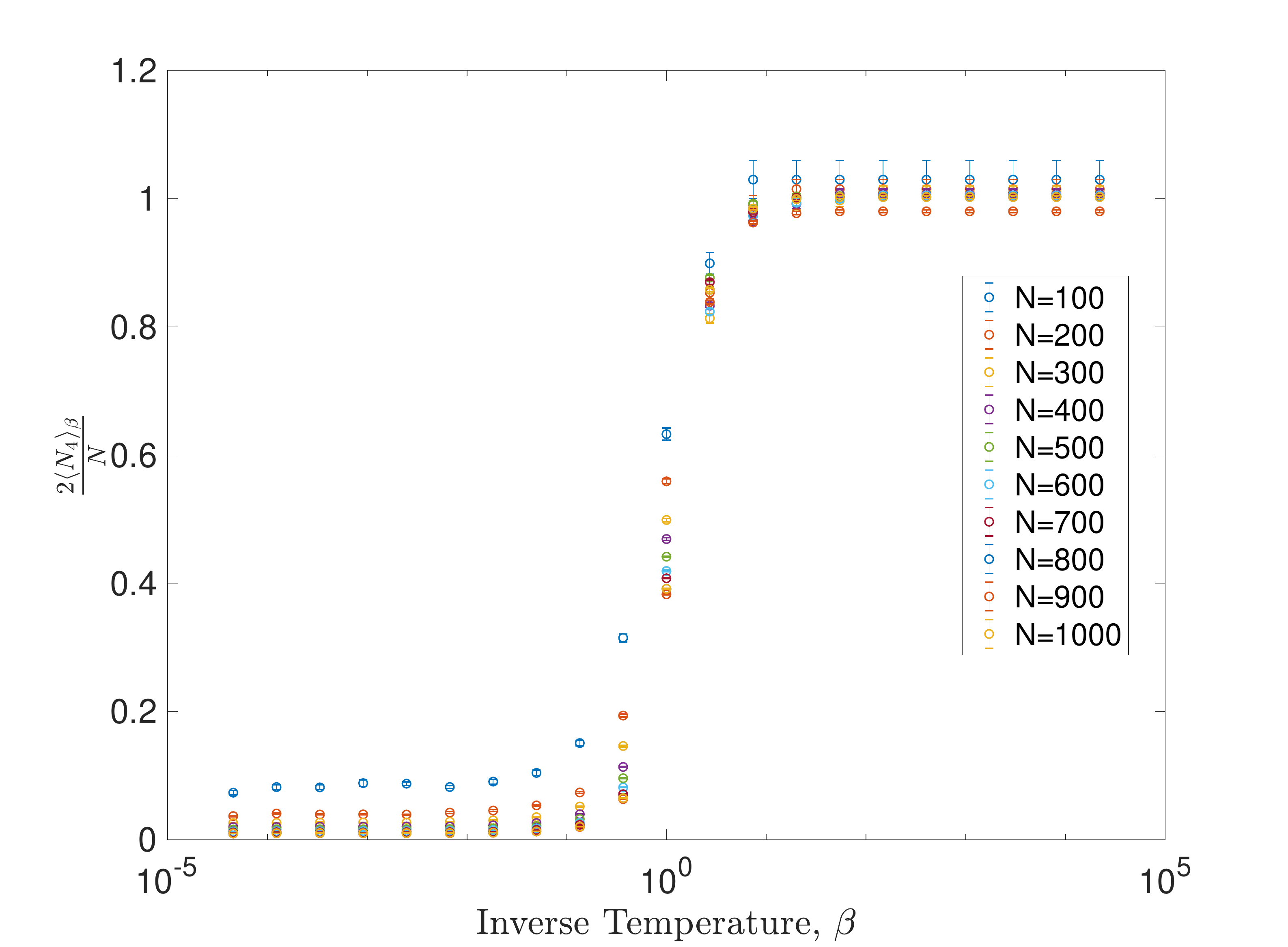}
        \subcaption{Expected number of squares; the normalisation is set to $1$ for the prism graph/M\"{o}bius ladder on $N$ vertices.}
        \end{subfigure}
        \begin{subfigure}{0.7\textwidth}
        \centering
        \includegraphics[width=\textwidth]{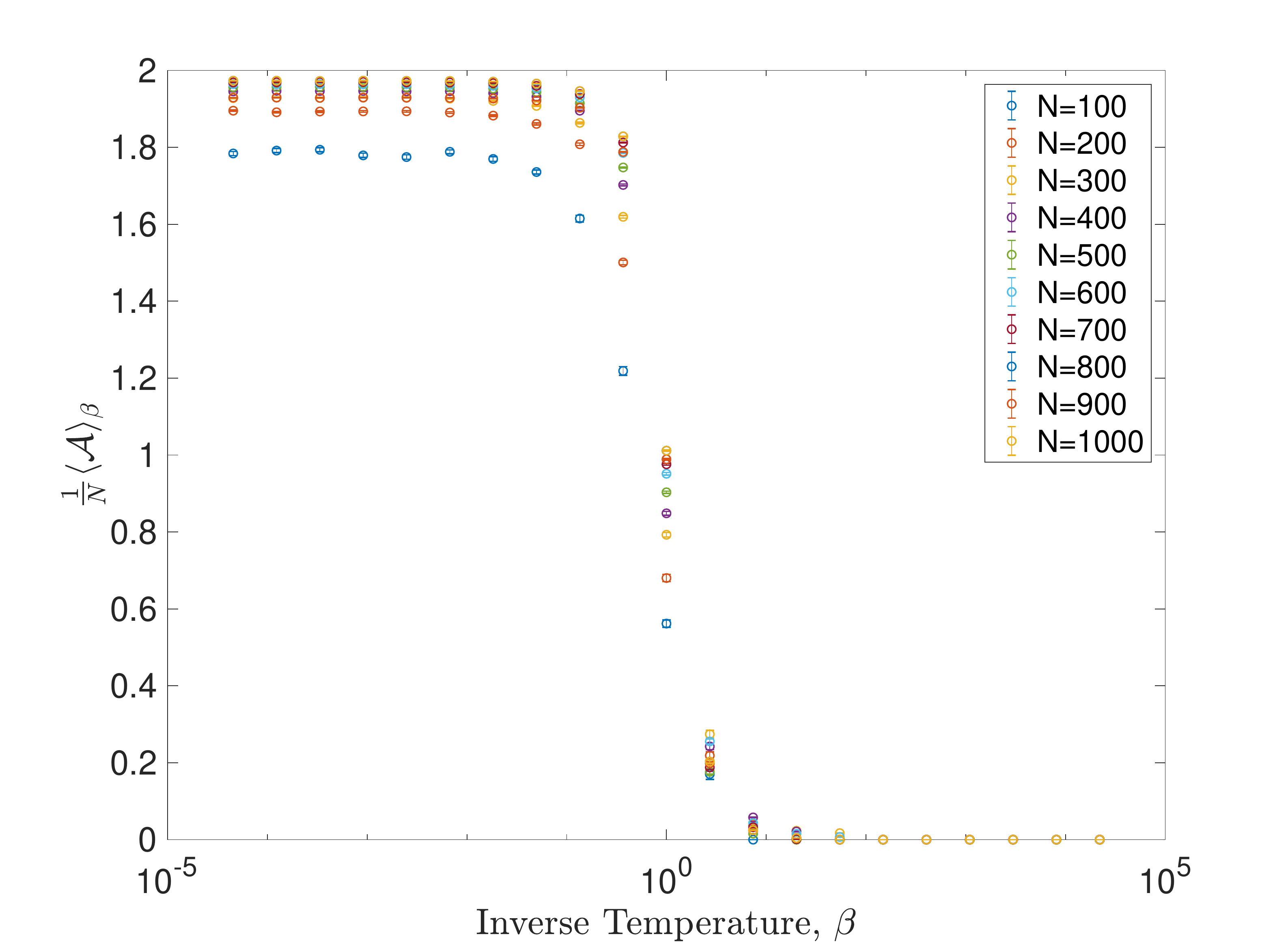}
        \subcaption{Intensive expected action.}
        \end{subfigure}
        \caption{Variation of two observables for a variety of graph sizes in an exponential range in the parameter $\beta$. Each point of the graph is obtained from $|E(\omega)|=3N/2$ sweeps where each sweep consists of $|E(\omega)|=3N/2$ attempted edge switches. We clearly see two regimes in both observables, and we identify a critical transitional region in the range $0.05\approx\exp(-3)\leq \beta\leq \exp(3)\approx 20$. Later we shall study this region more carefully.}
        \label{figure: OrderParameters}
    \end{figure}
    
    It has already been argued that two distinct phases exist: the random phase at low $\beta$ and the classical phase as $\beta\rightarrow \infty$. Indeed figure \ref{figure: OrderParameters} shows the behaviour of some observables including the action for an exponential range of values of the parameter $\beta$ and for a large range of graph sizes. The appearance of the two expected regimes as $\beta$ is varied is quite clear. From the same figure it appears that the phase transition occurs roughly in the region $0.05\leq \beta\leq 20$, and in fact we find that if we allow for longer relaxation times the upper value can be somewhat reduced. Figure \ref{figure: Action} shows a plot of the quantity $\beta\braket{\mathcal{A}}_\beta$ in the critical region for a smaller range of graph sizes. While we should not draw too many conclusions on the basis of such plots, there does not appear to be a discontinuity anywhere in the plot and as such we already have an indication that the phase transition is not first-order.
    
    \begin{figure}
        \centering
        \includegraphics[width=0.7\textwidth]{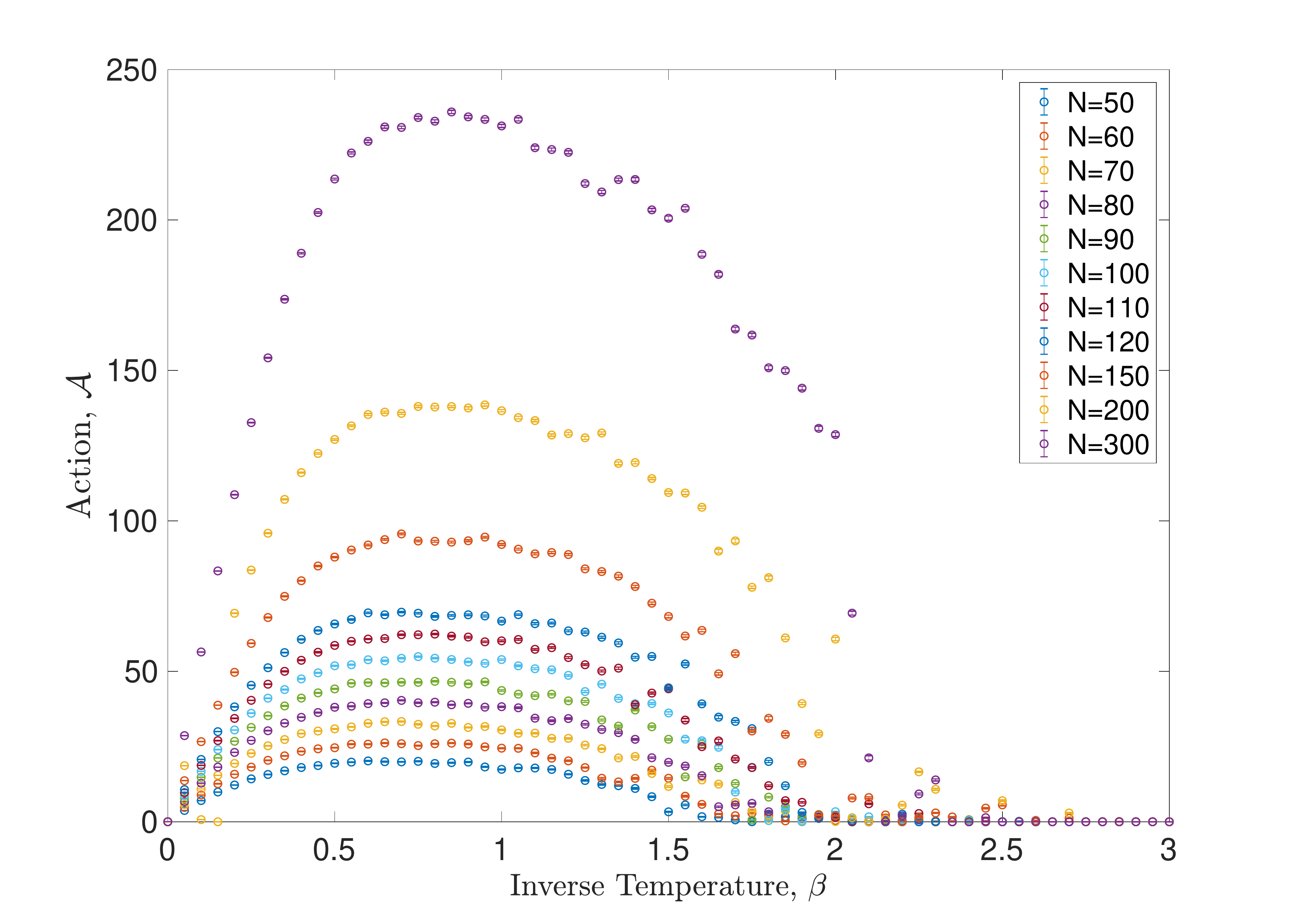}
        \caption{$\beta\braket{\mathcal{A}}_\beta$ in the critical region.}
        \label{figure: Action}
    \end{figure}
    
    A key step in the analysis of the phase transition is an estimation of the critical $\beta_c$.  {In general this procedure is a little subtle since the random graph models under consideration are expected to be infinite-dimensional statistical models. Indeed infinite dimensionality is typical in both random graph models \cite{Bollobas_RandomGraph, Janson_RandomGraph} and the statistical mechanics of networks \cite{AlbertBarabasi_StatMechCompNet, ParkNewman_StatMech} and is also a feature of discrete quantum gravity models \cite{KellyEtAl, Glaser_FinSizeScale} essentially as an expression of some expected nonlocality in the system. In practice this leads an $N$-dependent function $\beta_c$ with the precise $N$-dependence to be estimated via finite size scaling arguments as in \cite{Glaser_FinSizeScale}. If $\beta_c(N)$ appears to converge to some value $\beta_c$ as $N$ is increased we see that $N$-dependence of $\beta_c(N)$ detected is simply a finite-size effect and not a consequence of the infinite dimensionality of random graph models. Conversely if convergence is only observed after factoring out some power of $N$ then we see that infinite dimensionality of the model is at play in addition to finite size effects.}
    
    The most natural way to estimate $\beta_c(N)$ is to find the value of $\beta$ which maximises $C(N)$ where $C(N)$ denotes a numerical estimate of the specific heat $C$ for graphs of order $N$. In a first-order transition we expect $C$ to behave as a delta-function singularity at the transition point since, by definition, the action $\beta\braket{\mathcal{A}}_\beta$ has a discontinuity at the transition point. On the other hand, in a second-order transition we expect $C$ to diverge at $\beta_c$ as a natural consequence of the appearance of critical fluctuations and the divergence of the correlation length \cite{Goldenfeld_SM}. In both cases, in finite systems amenable to numerical analysis, the divergences are smeared out into finite peaks which become sharper as $N$ increases. Figure \ref{figure: SpecificHeat} shows the specific heat at various values graph sizes; we indeed observe sharp peaks forming as $N$ increases.
    
    \begin{figure}
        \centering
        \includegraphics[width=0.7\textwidth]{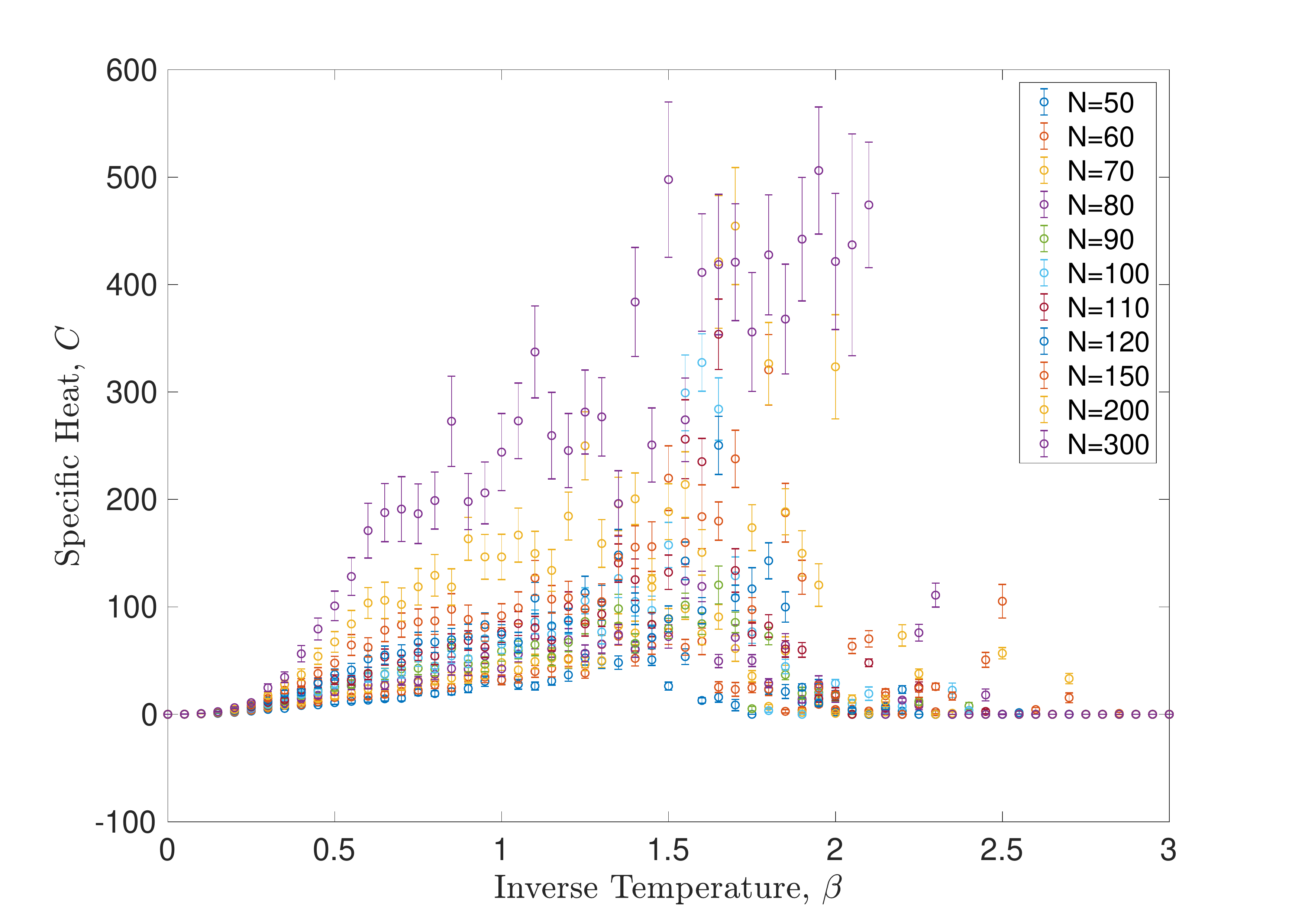}
        \caption{Specific heat in the critical region.}
        \label{figure: SpecificHeat}
    \end{figure}
    
    Figure \ref{figure: PseudoCriticalPoint} shows the value of $\beta$ which maximises the specific heat as a function of the graph size. This effectively shows the scaling of the  observed critical point values with $N$. Fitting against a plot of $y=m N^{-1}+c$ suggests that the critical value $\beta_c(\infty)=1.84\pm 0.03$; more important than the precise value stemming from this quantitative estimate is the qualitative adequacy of the $N^{-1}$ fit, in line with earlier expectations, showing that the system is finite-dimensional and thus has a well-defined critical inverse temperature.
    
    \begin{figure}
        \centering
        \includegraphics[width=0.7\textwidth]{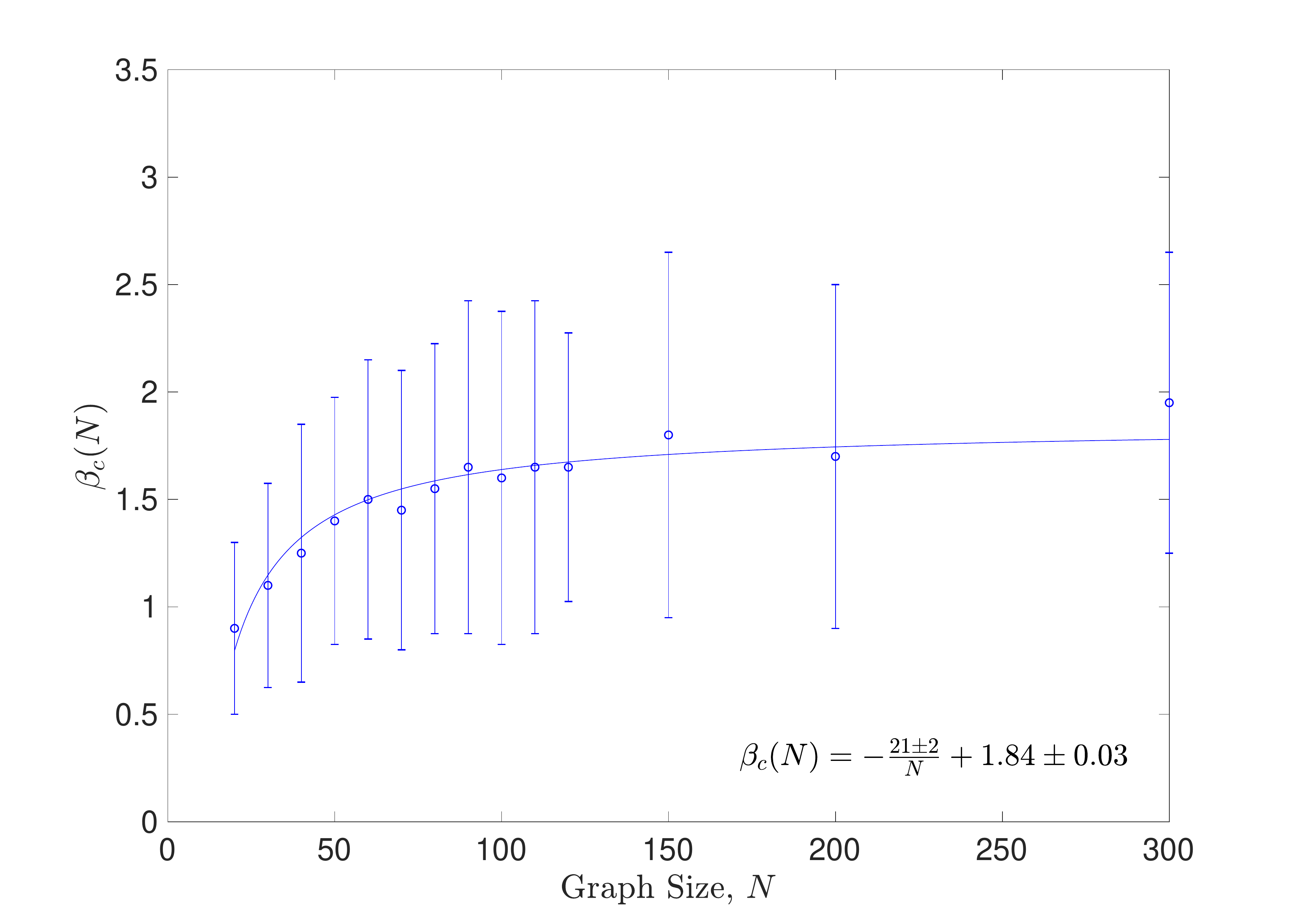}
        \caption{Observed values of the pseudocritical value $\beta_c(N)$. A fit of functions of the form $y=mN^{-1}+c$ to the observed data gives $\beta_c(N)=-(21\pm 2)N^{-1}+(1.84\pm 0.03)$. This suggests a constant critical value of $\beta_c(\infty)=1.84\pm 0.03.$}
        \label{figure: PseudoCriticalPoint}
    \end{figure}
    
    \subsection{The Nature of the Phase Transition}
    
    As already discussed above, a phase transition is first-order if there is a discontinuity in the quantity $\beta\braket{\mathcal{A}}_\beta$ and second-order if $\beta\braket{\mathcal{A}}_\beta$ is continuous while the specific heat $C$ is discontinuous. Such discontinuities are not easy to spot in finite systems while, as already argued, divergences in $C$ which do manifest in easily detectable traits of plots of finite systems cannot be used to distinguish between the order of the transition. Thus we need to consider other diagnostics if we are to identify the nature of the phase transition in question. 
    
    \begin{figure}
        \centering
        \includegraphics[width=0.7\textwidth]{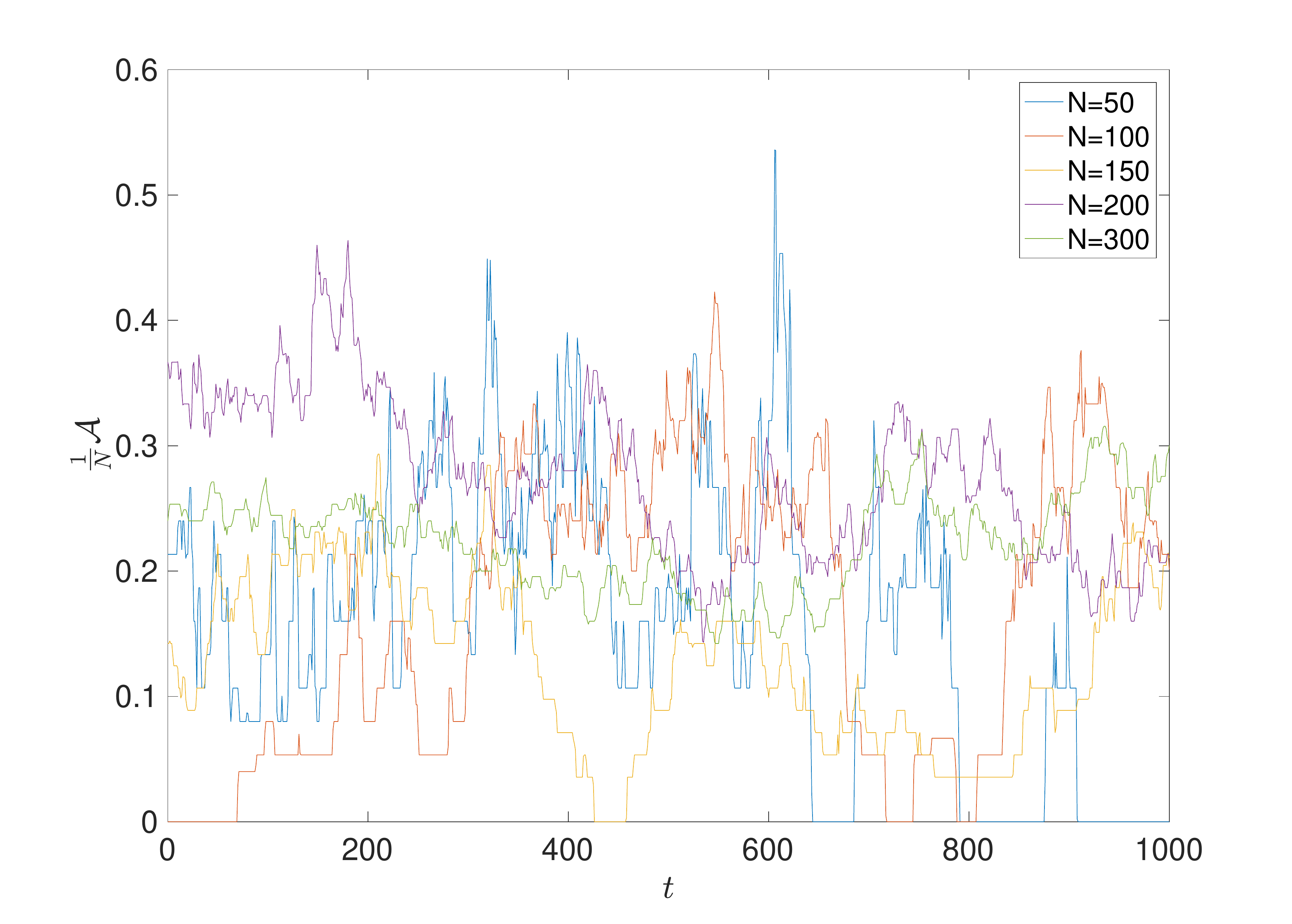}
        \caption{Time series data for the action at the critical temperature. $t$ is measured in sweeps, i.e. attempted Monte Carlo updates per edge.}
        \label{figure: PhaseCoexistence}
    \end{figure}
    
    One characteristic feature of first-order transitions which is not apparent in second-order transitions is phase coexistence at the critical value $\beta_c$: the action at $\beta_c$ can take on any of at least two distinct values separated by the discontinuity that characterises the transition. Regarded as a random variable the action will thus be distributed as a superposition of the distributions in each of the distinct pure phases. Assuming two distinct pure phases with distributions $P_0$ and $P_1$ respectively, the distribution of $\mathcal{A}$ at the critical temperature is then given
    \begin{align}\label{equation: ActionDistribution}
        \text{dist}_{\mathcal{A}}(E)=\alpha P_0(E)+(1-\alpha)P_1(E)
    \end{align}
    for some $\alpha\in (0,1)$, where $E$ is some event. Assuming that we have enough observations we expect both $P_0$ and $P_1$ to be normally distributed around (or even entirely concentrated at) some point and the action distribution \ref{equation: ActionDistribution} takes the form of two superimposed Gaussians, becoming more pronounced as the size of $N$ increases. Assuming that the discontinuity is larger than the effective support of the two Gaussians combined, which will occur if we consider sufficiently large $n$, we thus expect to observe two distinct peaks in the frequency distribution of the action at the critical value. This can be observed directly by looking at a frequency histogram of the observations of $\mathcal{A}$, checked for in time-series data where we expect to observe sharp transitions between the two Gaussian centres, or observed in the following modified \textit{Binder cumulant}:
    \begin{align}
        B = \frac{1}{3}\left(1-\frac{\braket{\mathcal{A}^4}_\beta}{\braket{\mathcal{A}^2}^2_\beta}\right).
    \end{align}
    Up to an overall factor, this is one minus the kurtosis of the distribution, and up a constant shift is the coefficient introduced in \cite{Binder_FiniteSizeEffects}. It can be shown that in a first-order transition, an observable taking the double Gaussian distribution \ref{equation: ActionDistribution} will have a non-zero minimum at the transition point and will take on the value $0$ elsewhere \cite{Binder_FiniteSizeEffects}. In a second-order transition, on the other hand, the action is distributed according to a single Gaussian with corresponding consequences for the observed frequency histogram and time-series data. In particular we expect $B$ to vanish identically everywhere in a second-order transition. 
    
    \begin{figure}
        \centering
        \begin{subfigure}{0.7\textwidth}
			\centering
			\includegraphics[width=\textwidth]{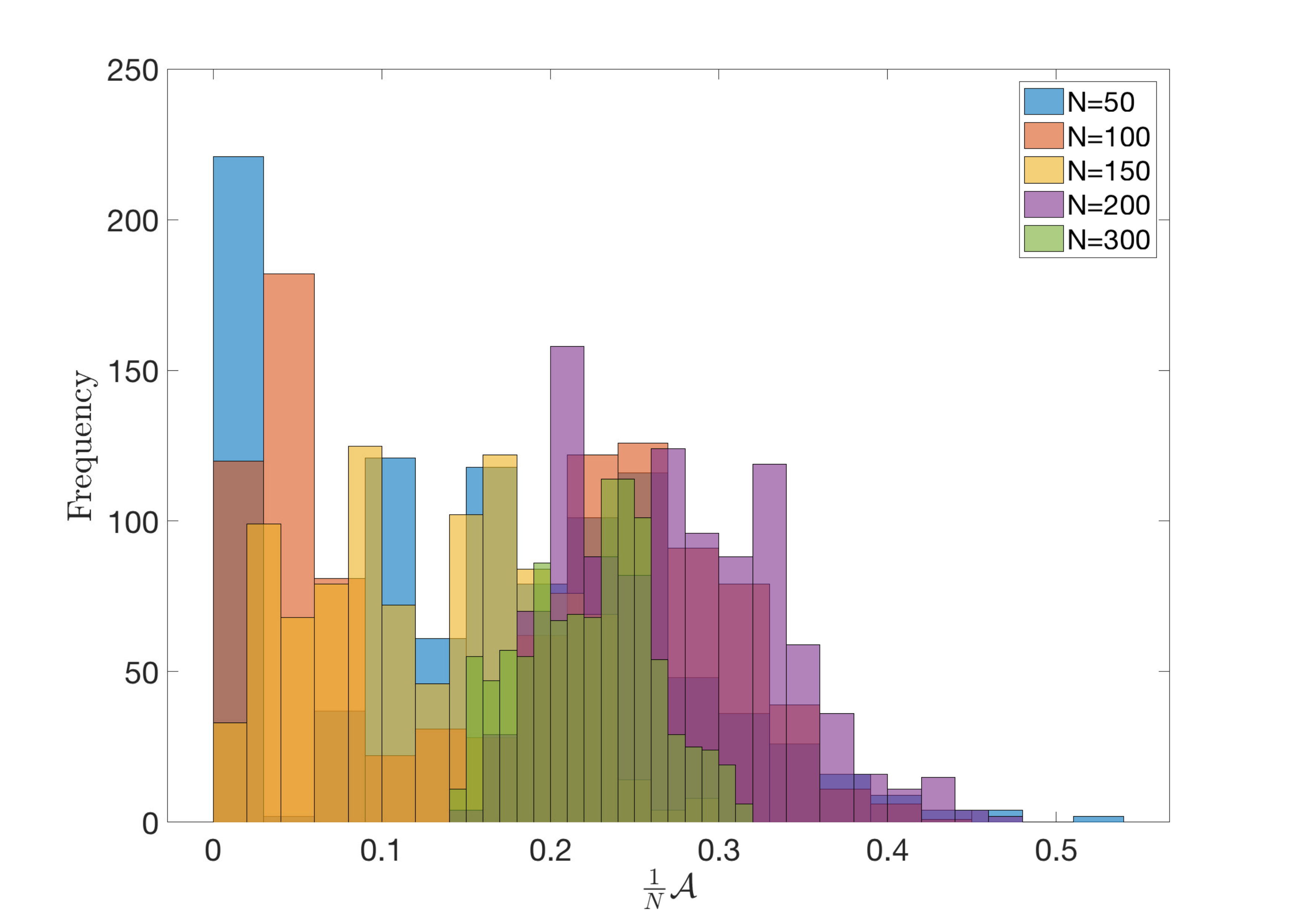}
			\subcaption{Histograms}\label{subfigure: Frequency_Histogram}
		\end{subfigure}
		\begin{subfigure}{0.7\textwidth}
			\centering
			\includegraphics[width=\textwidth]{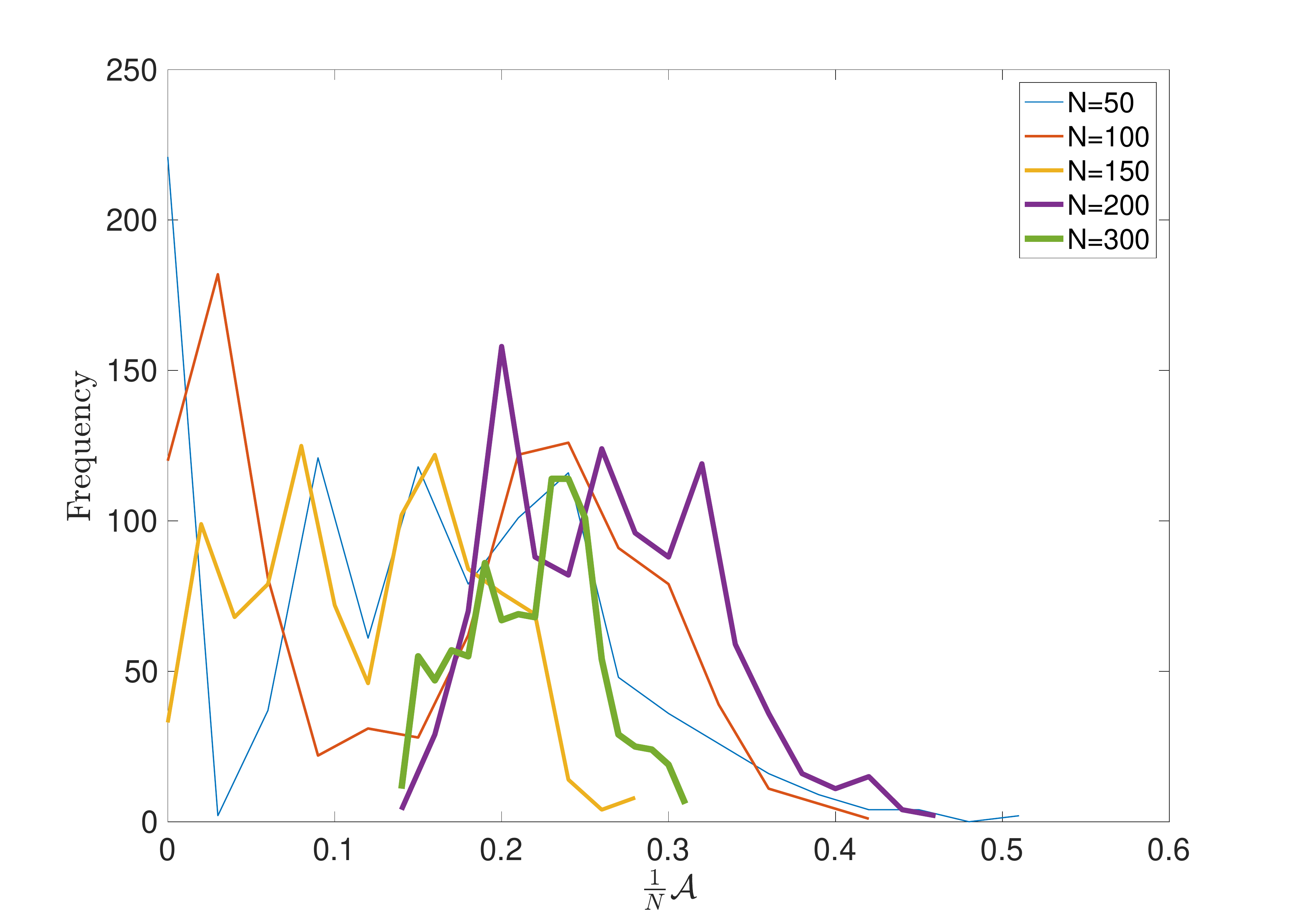}
			\subcaption{Frequency peaks.}\label{subfigure: Frequency_Peak}
		\end{subfigure}
        \caption{Frequency histograms of $\frac{1}{N}\mathcal{A}$. Figure \ref{subfigure: Frequency_Histogram} shows the frequency histograms for a variety of graph sizes. The tendency is somewhat more transparent in figure \ref{subfigure: Frequency_Peak} where only the peaks---of a smaller number of bins---are displayed. There appear to be two peaks that merge as $N$ increases.}\label{figure: FrequencyGaussians}
    \end{figure}
    
    In finite systems, of course, phase coexistence is observed for both first and second-order transitions. However in the former case it is a fundamental feature of the transition while it is simply a finite-size effect in the latter; as such phase coexistence and the associated phenomena become more pronounced as we increase $N$ in a first-order transition and less pronounced in a second-order transition. More concretely, we expect transitions between phases to become less sharp in time-series data, distinct Gaussians to merge and $B$ to approach $0$ as $N$ increases in a second-order transition. Figures \ref{figure: PhaseCoexistence}, \ref{figure: FrequencyGaussians} and \ref{figure: BinderCoefficient} all corroborate these expectations indicating that we indeed have a second-order transition.
    
    \begin{figure}
        \centering
        \includegraphics[width=0.7\textwidth]{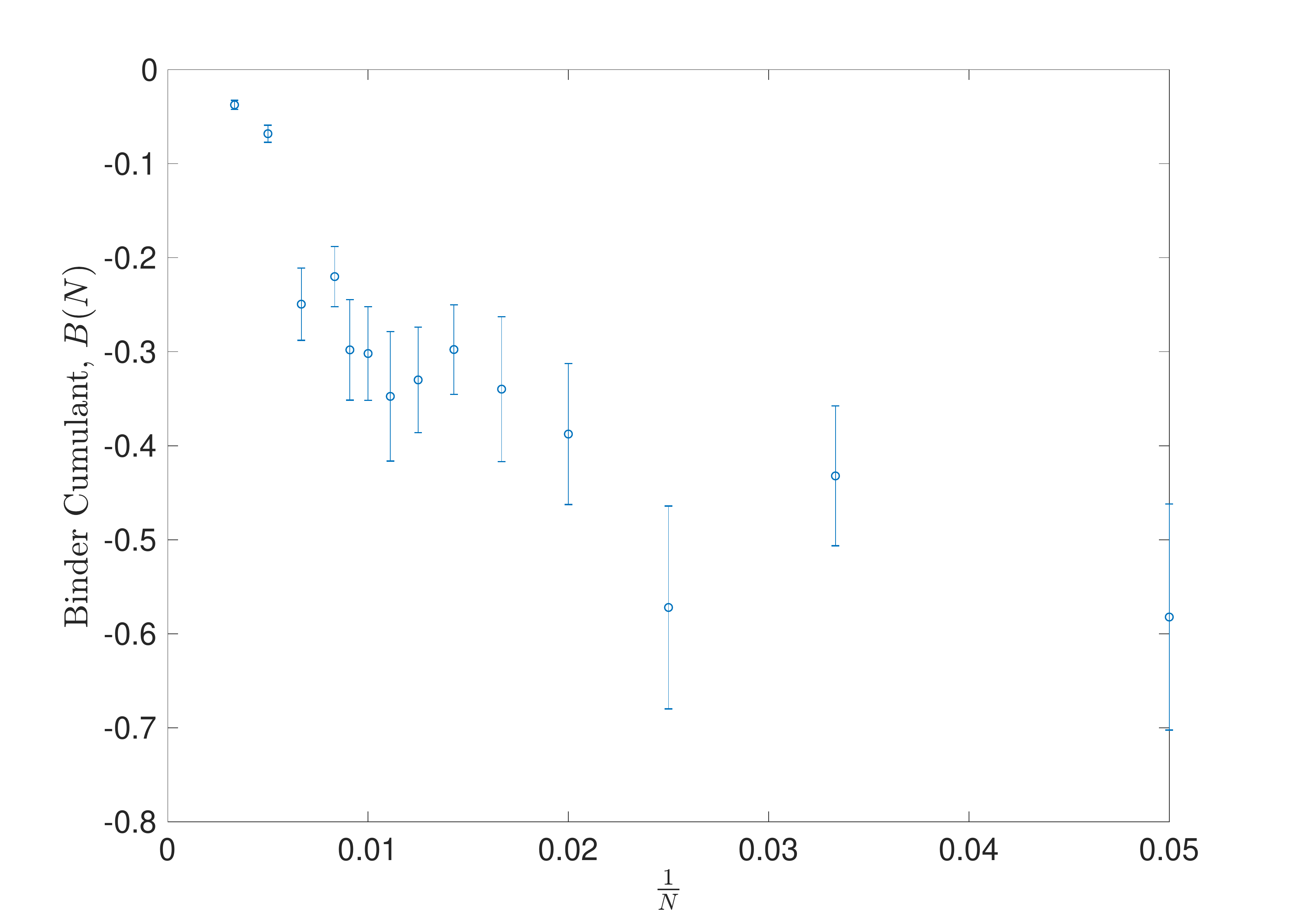}
        \caption{Binder coefficient for a variety of graph sizes. We see a clear approach to $0$ as $N^{-1}\rightarrow 0$.}
        \label{figure: BinderCoefficient}
    \end{figure}
    
     {Finally let us consider the divergence of the normalised specific heat. This is characterised in terms of a critical exponent $\lambda$ defined by:
    \begin{align}
        \frac{C}{N}\sim N^\lambda.
    \end{align}
    This may be estimated by looking at the collapse of the normalised specific heat $C/N$ in the critical regime. This is shown in figure \ref{figure: Scaling} where we have also shown the normalised dimensionless action; we obtain an estimate of $\lambda=0.15$. Note that the maximum value of the specific heat is indeed increasing as may be seen from figure \ref{figure: MaxSpHeat}.}
    
    \begin{figure}
        \centering
        \begin{subfigure}{0.7\textwidth}
			\centering
			\includegraphics[width=\textwidth]{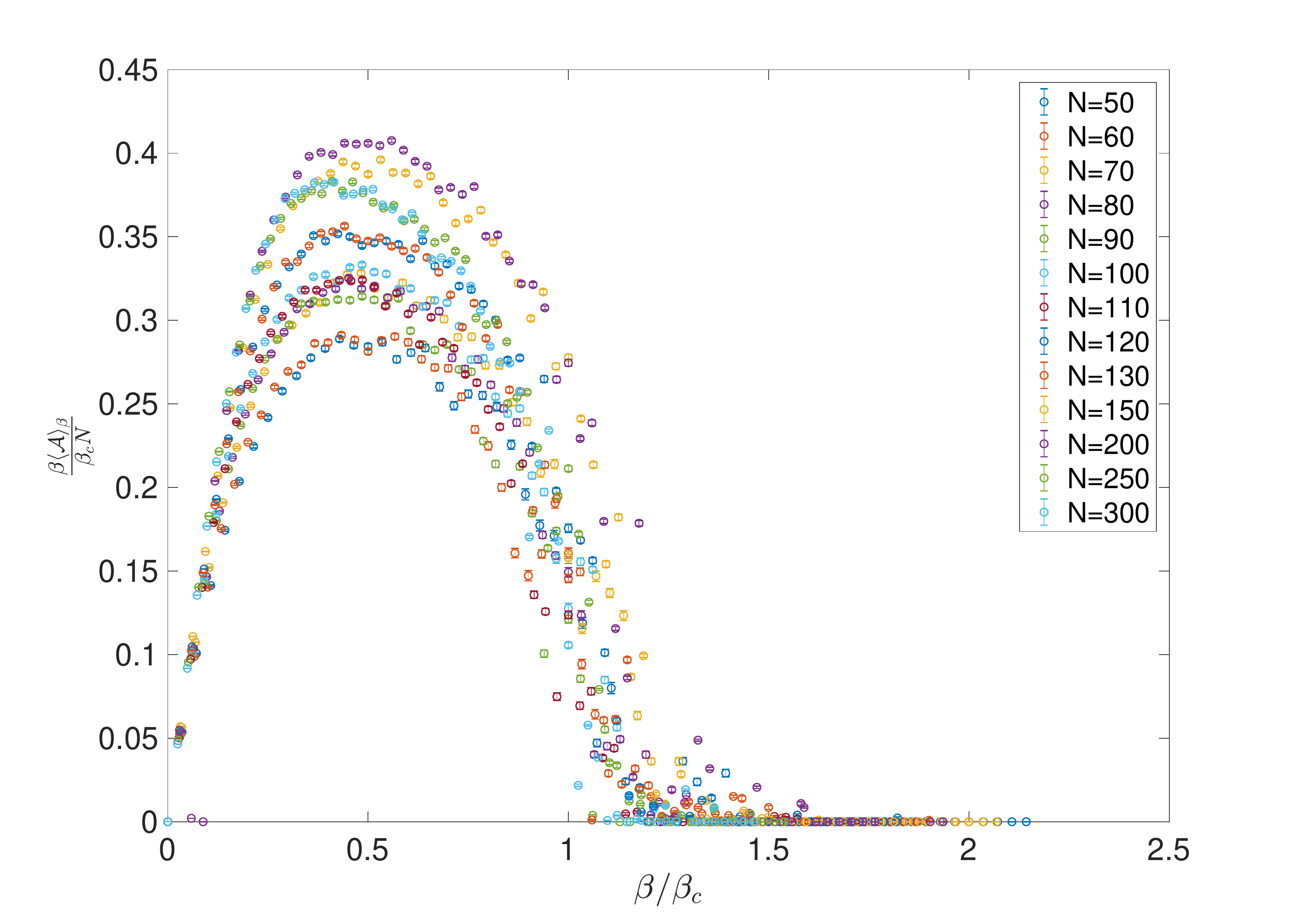}
			\subcaption{Action}\label{subfigure: Scaling_Action}
		\end{subfigure}
		\begin{subfigure}{0.7\textwidth}
			\centering
			\includegraphics[width=\textwidth]{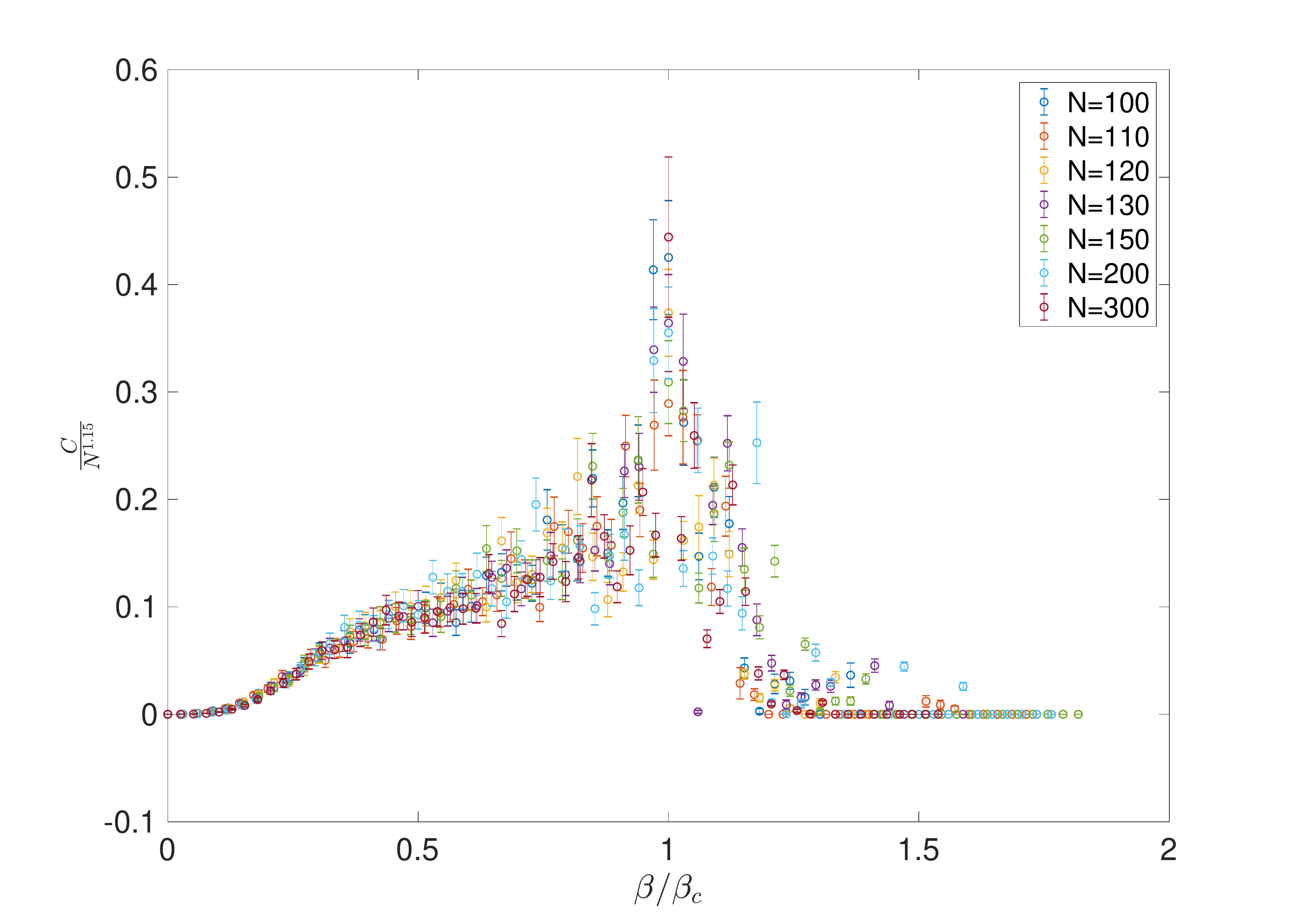}
			\subcaption{Specific heat.}\label{subfigure: Scaling_SpecificHeat}
		\end{subfigure}
        \caption{Scaling of $(\beta/\beta_c)\braket{\mathcal{A}}_\beta/N$ and $C/N$. We see relatively good collapse in the latter if we choose $\lambda=1.15$ compared to figure \ref{figure: SpecificHeat}.}
        \label{figure: Scaling}
    \end{figure}
    \begin{figure}
        \centering
        \includegraphics[width=0.7\textwidth]{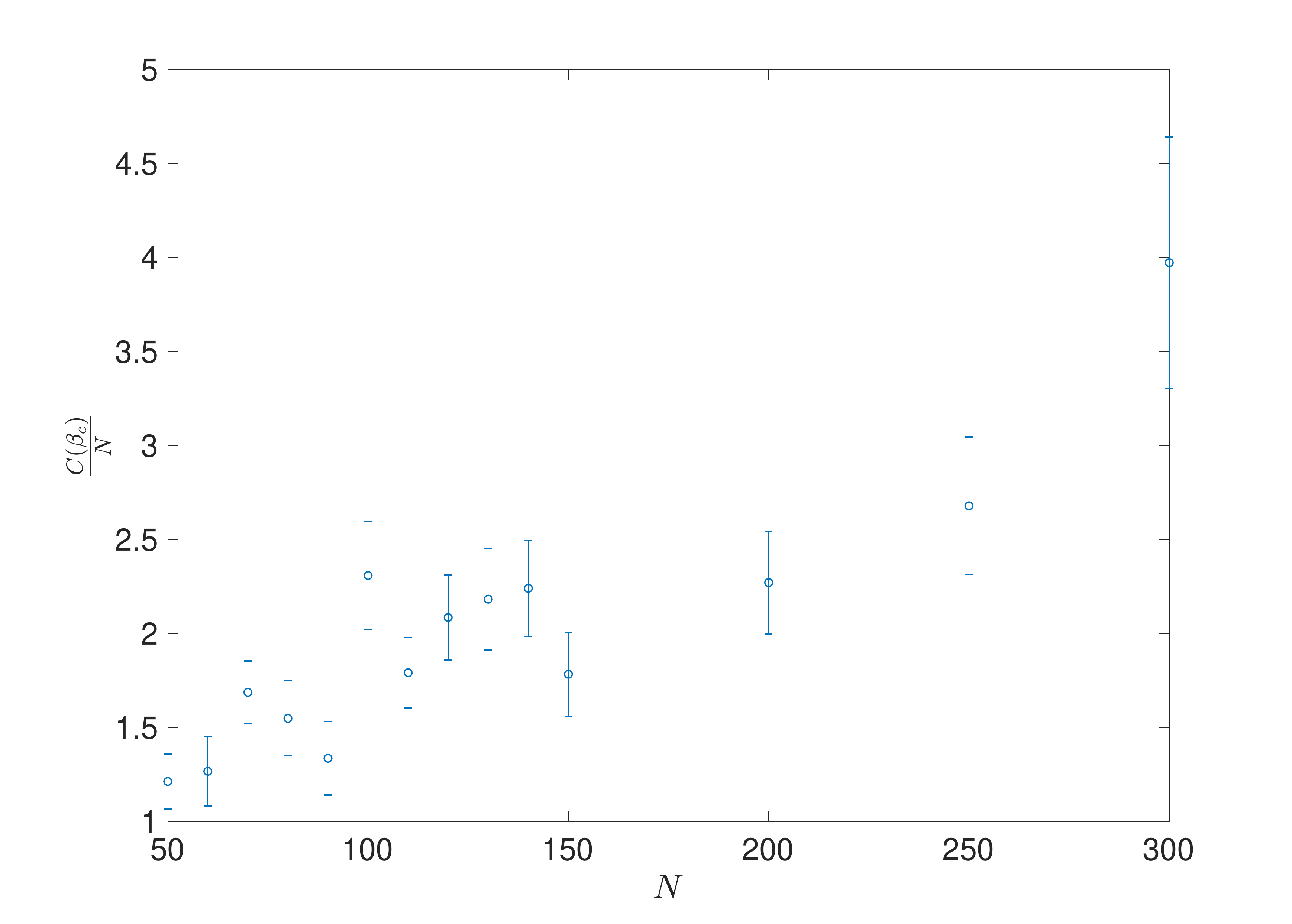}
        \caption{Variation of the maximum value of the specific heat with graph size.}
        \label{figure: MaxSpHeat}
    \end{figure}
    \section{Conclusion}
    In conclusion we have presented a toy model of emergent geometry in $1$-dimension. From a practical point of view, the key features of the model are the (dynamic or kinematic) suppression of triangles in conjunction and a classification of Ricci-flat graphs enabling us to guarantee that the scaling limit of possible classical configurations is a smooth geometry. The other essential ingredient is the fact that the system undergoes a continuous phase transition. Since Ricci-flat graphs are likely to exhibit significantly different properties from random regular graphs, and the evidence from \cite{KellyEtAl} suggests that the continuous phase transition was likely to persist as we increase the degree, difficulties in extending this model are thus concentrated on the classification of Ricci-flat graphs; some nontrivial information is contained in \cite{Kelly_Exact} but it seems unlikely that we will be able to obtain similarly rigorous results for the case of $4$-regular graphs, which in our formalism corresponds to surfaces. Nonetheless there are prospects for similar conclusions to be drawn in the case of surfaces ($4$-regular graphs) in the form of spectral and Hausdorff dimension results. This will be the topic of a future work. 
    \section*{Acknowledgements}
    C.K. acknowledges studentship funding from EPSRC under the grant number EP/L015110/1. F.B. acknowledges funding from EPSRC (UK) and the Max Planck Society for the Advancement of Science (Germany).
    \appendix
    \section{Ollivier Curvature}\label{appendix: OllivierCurvature}
    In this appendix we introduce the Ollivier curvature beginning with some basic ideas in optimal transport theory. We rely heavily on \cite{Villani_OptimalTransport,Ollivier_RCMCMS,JostLiu_RicciCurv, KellyEtAl}.
    \subsection{General Features}
    The Ollivier curvature is closely related to ideas of optimal transport theory in metric measure geometry. Let $X$ be a Polish (separable completely metrisable) space and let $\mathcal{P}(X)$ denote the family of all probability measures on the Borel $\sigma$-algebra of $X$. Given probability measures $\mu,\:\nu\in \mathcal{P}(X)$, a \textit{transport plan} from $\mu$ to $\nu$ is a probability measure $\xi$ on the Borel $\sigma$-algebra of $X^2$ satisfying the following \textit{marginal constraints}:
    \begin{subequations}
    \begin{align}
        \int_X \rm{d} \xi(x,y)f(x)&=\int \rm{d} \mu(x)f(x)\\
        \int_X \rm{d} \xi(x,y)f(y)&=\int \rm{d} \nu(y)f(y)
    \end{align}
    \end{subequations}
    for all measurable mappings $f:X\rightarrow \mathbb{R}$. Let $\Pi(\mu,\nu)$ denote the set of all transport plans from $\mu$ to $\nu$. Roughly speaking a transport plan $\xi\in \Pi(\mu,\nu)$ defines a method of transforming a distribution $\mu$ of matter into a distribution $\nu$ of matter.
    
    Given a (measurable) \textit{cost} function $c:X^2\rightarrow \mathbb{R}$, the \textit{transport cost} of a transport plan $\xi\in \Pi(\mu,\nu)$ is defined
    \begin{align}
        W_c(\xi)=\int_X\rm{d} \xi(x,y)c(x,y).
    \end{align}
    The \textit{optimal transport cost} is then
    \begin{align}
        W_c(\mu,\nu)=\inf_{\xi\in \Pi(\mu,\nu)}W_c(\xi).
    \end{align}
    Given a metric $\rho$ on $X$ (compatible with the topology) the \textit{Wasserstein $p$-distance} is defined
    \begin{align}
        W_p(\mu,\nu)=\sqrt[p]{W_{\rho^p}(\mu,\nu)}
    \end{align}
    i.e. as the $p$-th root of the optimal transport cost for cost function given by the $p$th power of the metric. It can be shown that the Wasserstein $p$-distances define infinite metrics on the space $\mathcal{P}(X)$. The metrics are finite if we restrict to the space of probability measures with finite $n$th moments for $n\leq p$.
    
    The Ollivier curvature is an extension of the Wasserstein distance particular to \textit{metric measure spaces}, i.e. metric spaces $(X,\rho)$ equipped with a family of probability measures $\set{\mu_x}_{x\in X}\subseteq \mathcal{P}(X)$. We may refer to the family $\set{\mu_x}_{x\in X}$ as a \textit{random walk} on $X$. We give two particularly important examples:
    \begin{itemize}
        \item Every (connected locally finite) graph $G$ is naturally a metric measure space in the following way: the metric structure is given by the standard geodesic distance between vertices; the random walk on $X$ is given by picking
        \begin{align}
            \mu_x^p(y)=\left\{\begin{array}{rl}
                p, & y=x \\
                \frac{1}{d_x}(1-p), & y\in N_G(x)\\
                0, & \rm{otherwise}
            \end{array}\right.
        \end{align}
        for each $x\in X$ for some $p\in [0,1]$. We call this random walk the \textit{lazy random walk} of idleness $p$. The lazy random walk of idleness $0$ is referred to as the \textit{uniform random walk.}
        \item A Riemannian manifold $(\mathcal{M},g)$ is a metric measure space where again the metric is given by the standard geodesic metric and a random walk defined by the assignment:
        \begin{align}
            \rm{d}\mu_x^\varepsilon(y)=\left\{\begin{array}{rl}
                \frac{1}{\rm{vol}(B_{\varepsilon}(x))}, & y\in B_\varepsilon(x) \\
                0, & \rm{otherwise}
            \end{array}\right.
        \end{align}
        for each $x\in \mathcal{M}$ for some (small) choice of $\varepsilon>0$.
    \end{itemize}
    In such contexts, the Ollivier curvature is defined
    \begin{align}
        \kappa_X(x,y)=1-\frac{W_1(\mu_x,\mu_y)}{\rho(x,y)}
    \end{align}
    for all distinct $x,\:y\in X$. Rearranging we see that
    \begin{align}
        W_1(\mu_x,\mu_y)=(1-\kappa_X(x,y))\rho(x,y),
    \end{align}
    i.e. lower bounds on the curvature imply control over the dilatation of the natural imbeddings $x\mapsto \mu_x$. In particular we see that $\kappa_X(x,y)\geq 0$ iff $W_1(x,y)\leq \rho(x,y)$; for Riemannian manifolds this means that a space is positively curved iff the average distance between two open balls centred at $x$ and $y$ respectively is less than the distance between their centres. This is closely related to the Ricci curvature; indeed for $x$ and $y$ sufficiently close, if $u$ is a vector field generating a geodesic from $x$ to $y$ we have
    \begin{align}
        \kappa_X(x,y)=\frac{\varepsilon^2}{2(D+2)}\rm{Ric}(u,u)+\mathcal{O}(\varepsilon^3)+\mathcal{O}(\varepsilon^2\rho(x,y)),
    \end{align}
    where $D=\rm{dim}(\mathcal{M})$. In this sense the Ollivier curvature is a generalisation of the manifold Ricci curvature to much rougher contexts than is typical. 
    \subsection{Discrete Properties}
    The discrete context has several important consequences for the Ollivier curvature. In particular suppose we are working in a (connected, locally finite, simple, unweighted) graph $G$ equipped with the uniform random walk. Henceforth, the Ollivier curvature is also regarded as a mapping $\kappa_G:E(G)\rightarrow \mathbb{R}$ on edges; clearly $\rho(u,v)=1$ for any edge $uv\in E(G)$ so
    \begin{align}
        \kappa_G(uv)=1-W_1(\mu_u,\mu_v),
    \end{align}
    
    The fundamental effect of discreteness is to give a linear character to the entire problem: by the construction of finite-dimensional free vector spaces, probability measures $\mu$ on graphs with finite support may be regarded as real vectors $\boldsymbol{\mu}\in [0,1]^{n}$ for any $n\geq |\rm{supp}(\mu)|$ such that $\boldsymbol{\mu}_x=\mu(x)$ for all $x\in \rm{supp}(\mu)$ and $0$ otherwise; in particular, assuming that $G$ is equipped with the uniform random walk, the measure $\mu_u$ is given by a $d_u$-dimensional real vector $\boldsymbol{\mu_u}\in (0,1]^{d_u}$. Transport plans $\boldsymbol{\xi}\in \Pi(\boldsymbol{\mu_u},\boldsymbol{\mu_v})$ are then $[0,1]$-valued matrices satisfying the following \textit{discrete} marginal constraints:
    \begin{align}
        \boldsymbol{\xi}\boldsymbol{1}_{d_v}=\boldsymbol{\mu_u} && \boldsymbol{\xi}^T\boldsymbol{1}_{d_u}=\boldsymbol{\mu_v}
    \end{align}
    where $\boldsymbol{1}_n$ is the $n$-dimensional column with $1$ for each entry. Letting $\mathbb{D}(u,v)$ denote the $d_u\times d_v$-dimensional matrix with entry $(x,y)\in N_G(u)\times N_G(v)$ given by the distance $\rho(x,y)$, we have that the transport cost
    \begin{align}
        W_\rho(\xi)=\boldsymbol{\xi}\cdot \mathbb{D}(u,v)
    \end{align}
    where $\cdot$ denotes the Frobenius (elementwise) inner product. The Wasserstein distance $W(\mu_u,\mu_v)$ is then defined via a linear programme which gives a general computational framework for exact evaluations of the Ollivier curvature.
    
    The second point to note about the discrete context is the \textit{Kantorovitch duality theorem}. In particular
    \begin{subequations}
    \begin{align}
        W_1(\mu_u,\mu_v)&=\sup_{f\in \mathbb{L}(X,\mathbb{R})}\sum_{x\in X}(\mu_u(x)-\mu_v(x))f(x)\\
        &=\sup_{f\in \mathbb{L}(X,\mathbb{R})}\left(\sum_{x\in N_G(u)}\frac{f(x)}{d_u}-\sum_{y\in N_G(v)}\frac{f(y)}{d_v}\right) 
    \end{align}
    \end{subequations}
    where $\mathbb{L}(X,\mathbb{R})$ is the set of all $1$-Lipschitz maps $f:X\rightarrow\mathbb{R}$. (Recall that a $1$-Lipschitz map between two metric spaces $(X,\rho_X)$ and $(Y,\rho_Y)$ is a mapping $f:X\rightarrow Y$ such that $\rho_Y(f(x),f(y))\leq \rho_X(x,y)$ for all $x,\:y\in X$.) In the discrete context, Kantorovitch duality is simply an expression of the strong duality theorem in linear optimisation theory. Note that Kantorovitch duality---in a somewhat more involved form---generalises to continuous spaces. 
    
    A third key feature of the discrete setting---which does not generalise well beyond the discrete context---is the existence of \textit{core neighbourhoods}. A core neighbourhood of an edge $uv\in E(G)$ is a subgraph $H\subseteq G$ such that $N_G(u)\cup N_G(v)\subseteq V(H)$ and such that $\rho_H(x,y)=\rho_G(x,y)$ for all $(x,y)\in N_G(u)\times N_G(v)$. Then clearly
    \begin{align}
        \kappa_G(uv)=\kappa_H(uv)
    \end{align}
    and for the purposes of calculating the Ollivier curvature one may as well calculated the curvature in the core neighbourhood. The main utility of this notion comes when we realise that the induced subgraph of $G$ defined by 
    \begin{align}
        C(uv)=N_G(u)\cup N_G(v)\cup \pentagon(uv)
    \end{align}
    is a core neighbourhood, where
    \begin{align}
        \pentagon(uv)=\set{w\in V(G):\rho(u,w)=2\:\rm{and}\:\rho(v,w)=2}.
    \end{align}
    Roughly speaking, $C(uv)$ is the set of all vertices that lie on a triangle, square or pentagon supported by $uv$, or that neighbour $u$ or $v$ without lying on a short cycle supported by $u$ or $v$. This is a core neighbourhood because for any $(x,y)\in N_G(u)\times N_G(v)$ we have a $3$-path $xuvy$.
    
    The final key feature of the discrete setting is the \textit{discrete} nature of the Ollivier curvature. In particular, the Wasserstein distance may be found by optimising of integer-valued $1$-Lipschitz maps:
    \begin{align}
        W_1(\mu_u,\mu_v)=\sup_{f\in \mathbb{L}(X,\mathbb{Z})}\sum_{x\in X}(\mu_u(x)-\mu_v(x))f(x)
    \end{align}
    This can be shown directly or by using standard ideas of linear optimisation theory.
    
    \section{Gromov-Hausdorff Distance}\label{appendix: GromovHausdorffDistance}
    The Gromov-Hausdorff distance is a metric on the space of isometry classes of compact metric spaces. It is a generalisation of the \textit{Hausdorff distance} between subsets of a metric space. Significantly it defines a notion of convergence between metric spaces and gives a rigorous framework for thinking about emergent geometry in the thermodynamic limit of classes of graph. Much of the material in this appendix is covered in \cite{BuragoBuragoIvanov_MetricGeometry, Gromov_MetricStructures}; for elementary metric topology also see, for instance, \cite{KolmogorovFomin_RealAnalysis}.
    \subsection{Metric Topology}
    First let us recall some standard ideas from metric topology: a metric on a set $X$ is a positive definite, symmetric, subadditive function $\rho:X^2\rightarrow [-\infty,\infty]$, i.e. $\rho$ is a metric on $X$ iff it satisfies the following properties:
    \begin{itemize}
        \item \textit{Positivity:} $\rho(x,y)\geq 0$ for all $x,\:y\in X$.
        \item \textit{Definiteness:} $\rho(x,y)=0$ iff $x=y$ for all $x,\:y\in X$. Note that a mapping is \textit{semidefinite} iff $\rho(x,x)=0$ for all $x\in X$.
        \item \textit{Symmetry}: $\rho(x,y)=\rho(y,x)$ for all $x,\:y\in X$.
        \item \textit{Subadditivity:} $\rho(x,y)\leq \rho(x,z)+\rho(z,y)$ for all $x,\:y,\:z\in X$.
    \end{itemize}
    A metric may be thought of as defining the \textit{distance} between any two points of a space. If instead of definiteness, $\rho$ is just semidefinite then $\rho$ is said to be a \textit{pseudometric} on $X$. A (pseudo)metric is said to be \textit{finite} iff $\rho(x,y)< \infty$ for all $x,\:y\in X$. A \textit{(pseudo)metric space} is a pair $(X,\rho)$ where $X$ is a set and $\rho$ a (pseudo)metric on $X$.
    
    Clearly every (pseudo)metric on $X$ restricts to a (pseudo)metric on subsets of $X$. Given a pseudometric $\rho$ on $X$, there is an equivalence relation $\cong$ on $X$ given by $x\cong y$ iff $\rho(x,y)=0$. $\rho$ is naturally interpreted as a metric on the quotient $X/\cong$. The resulting metric space is said to be \textit{induced} by the pseudometric $\rho$ and is denoted $X/\rho$.
    
    Every pseudometric $\rho$ on a space $X$ gives rise to a topology in the following way: for each $x\in X$ and each $r>0$, define the \textit{open ball of radius $r$ and centred at $x$} as the set
    \begin{align}
        B_r(x)=\set{y\in X:\rho(x,y)<r}.
    \end{align}
    The open balls form a base for a topology in $X$ called the \textit{metric topology} of $X$. Two metrics $\rho_1$ and $\rho_2$ on $X$ give rise to the same topology iff for each $x\in X$ there are constants $\alpha,\:\beta>0$ such that
    \begin{align}
        \alpha \rho_1(x,y)\leq \rho_2(x,y)\leq \beta\rho_2(x,y)
    \end{align}
    for each $y\in X$. Given two metric spaces $(X,\rho_X)$ and $(Y,\rho_Y)$, an \textit{isometry} is a distance preserving mapping between $X$ and $Y$, i.e. a bijection $f:X\rightarrow Y$ such that $\rho_Y(f(x),f(y))=\rho(x,y)$ for all $x,\:y\in X$. Two metric spaces are \textit{isometric} iff they are related by an isometry. Clearly isometric spaces are homeomorphic though the converse does not necessarily hold. 
    
    It turns out that every metric space is first-countable, which, in particular, means that all questions of convergence in metric spaces can be settled by considering the convergence of \textit{sequences}. Recall that a sequence $\set{x_n}_{n\in \mathbb{N}}\subseteq X$ is said to \textit{converge} to a point $x\in X$ iff for each $\varepsilon>0$ there is an $N\in \mathbb{N}$ such that $\rho(x,x_m)<\varepsilon$ for all $m\geq N$. We then call $x$ a \textit{limit} of the sequence $\set{x_n}_{n\in \mathbb{N}}$. A sequence is said to be \textit{convergent} iff it has a limit and \textit{divergent} otherwise. Every metric space is Hausdorff which means that every convergent sequence has a unique limit, denoted:
    \begin{align}
        \lim_{n\rightarrow\infty}x_n.
    \end{align}
    Roughly speaking, a sequence $\set{x_n}_{n\in \mathbb{N}}$ has a limit $x$ iff most of the sequence is arbitrarily close to $x$, i.e. there are at most a finite number of points in the sequence more than some given distance away from $x$. 
    
    A sequence $\set{x_n}_{n\in \mathbb{N}}$ is said to be \textit{Cauchy} iff for every $\varepsilon>0$ there exists an integer $N\in \mathbb{N}$ such that $\rho(x_m,x_n)<\varepsilon$ for every $m,\:n\geq N$. Heuristically, a sequence is Cauchy iff most of the points of the sequence lie arbitrarily close to one another. It is clear that every convergent sequence is Cauchy; the converse need not hold. For instance consider the sequence $\set{10^{-n}}_{n\in \mathbb{N}}$ in the space $(0,1]$; this is Cauchy but does not converge in the space $(0,1]$ since its limit (in $\mathbb{R}$) is $0$. From this example it appears that a Cauchy sequence is one which \textit{should} converge, but fails to do so because the space fails to contain all the relevant points. This intuition is captured in the notion of \textit{completeness:} a metric space is said to be \textit{(Cauchy) complete} iff every Cauchy sequence of the space is convergent. A subset of a metric space is said to be complete iff it is complete as a metric subspace. One example of a complete metric space is the real line $\mathbb{R}$ equipped with its standard metric
    \begin{align}
        \rho(x,y)=|x-y|.
    \end{align}
    
    Every metric $(X,\rho)$ space can be isometrically imbedded in a unique least complete metric space known as the \textit{completion} of $(X,\rho)$. To construct the completion we define a pseudometric $\tilde{\rho}$ on the space of Cauchy sequences of $X$ via
    \begin{align}
        \tilde{\rho}(\set{x_n}_{n\in \mathbb{N}},\set{y_n}_{n\in \mathbb{N}})=\lim_{n\rightarrow \infty}\rho(x_n,y_n)
    \end{align}
    where the right-hand-side exists since the real numbers are complete. $(X,\rho)$ isometrically imbeds in the metric space induced by $\tilde{\rho}$ which is also the least complete metric space to contain $(X,\rho)$ in this manner.
    
    The \textit{closure} $\overline{A}$ of a set $A\subseteq X$ is the set of limits of all sequences $\set{x_n}_{n\in \mathbb{N}}\subseteq A$ while a set is \textit{closed} iff it is equal to its closure. It is clear that every complete subset of a metric space is closed since every convergent sequence is Cauchy; as a partial converse, every closed subset of a \textit{complete} metric space is complete as a metric subspace. 
    
    For any set $A\subseteq X$ and any $\varepsilon>0$, the \textit{$\varepsilon$-thickening} of $A$ is defined as the set
    \begin{align}
        A_\varepsilon=\bigcup_{x\in X}B_\varepsilon(x).
    \end{align}
    A set $A\subseteq X$ is said to be a \textit{$\varepsilon$-net} in $X$ iff $X=A_\varepsilon$. $\varepsilon$-thickenings and $\varepsilon$-nets will feature prominently in the subsequent.
    
    A set $A\subseteq X$ is said to be \textit{bounded} iff it is contained in some open ball of the space $X$. $X$ is \textit{totally bounded} iff for each $\varepsilon>0$, $X$ admits a finite $\varepsilon$-net. Every totally bounded set is bounded. $X$ is said to be \textit{compact} iff it is complete and totally bounded, while a subset $A\subseteq X$ is said to be compact iff it is compact as a metric subspace. Clearly then every compact subset is closed and bounded. The converse does not hold for arbitrary metric spaces, though it does for Euclidean spaces by the Heine-Borel theorem. Note that every compact metric space admits a finite $\varepsilon$-net for each $\varepsilon>0$.
    
    \subsection{Defining the Gromov-Hausdorff Distance}
    Let $(X,\rho)$ be a (pseudo)metric space, i.e. let $\rho$ be a metric on $X$. For any $x\in X$ and any $A\subseteq X$ we define
    \begin{align}
        \rho(x,A)=\inf_{y\in A} \rho(x,y)
    \end{align}
    This is essentially the smallest distance between $x$ and a point of $A$. Clearly if $x\in A$ then $\rho(X,A)=0$. The \textit{Hausdorff distance} between two subsets $A,\:B\subseteq X$ is then defined
    \begin{align}
        \rho_H(A,B)=\max\set{\sup_{x\in A}\rho(x,B),\sup_{y\in B}\rho_X(y,A)}.
    \end{align}
    The Hausdorff distance has an alternative characterisation in terms of $\varepsilon$-thickenings. In particular, it is simple to show that:
    \begin{align}
        \rho_H(A,B)=\inf\set{\varepsilon>0:A\subseteq B_\varepsilon\:\rm{and}\:B\subseteq A_\varepsilon}.
    \end{align}
    To see this, it is sufficient to note that if $A\subseteq B_\varepsilon$ for some $\varepsilon>0$, then $\rho(x,B)\leq \varepsilon$ for each $x\in A$ and similarly $B\subseteq A_\varepsilon$ implies $\rho(y,A)\leq \varepsilon$ for each $y\in B$. The suprema/infima then force equality.
    
    The Hausdorff distance defines a pseudometric on the space $\mathfrak{p}(X)$ of all subsets of $X$: positivity, semidefiniteness and symmetry are all trivial. For subadditivity note that
    \begin{align}
        A_{\varepsilon+\epsilon}\subseteq (A_\varepsilon)_\epsilon
    \end{align}
    by the subadditivity of $\rho$. Thus we have an induced metric space $\mathfrak{p}(X)/\rho_H$. It turns out that we may identify $\mathfrak{p}(X)/\rho_H$ with the set $\mathfrak{C}(X)$ of closed subsets of $X$. To see this note that:
    \begin{enumerate}
        \item $\rho_H(A,\overline{A})=0$ for all $A\subseteq X$.
        \item $\rho_H(A,B)\neq 0$ for distinct $A,\:B\in \mathfrak{C}(X)$.
    \end{enumerate}
    That is to say, the closed sets of $X$ can be chosen as representatives of the equivalence classes in $\mathfrak{p}(X)/\rho_H$. For the first of these results note that $\rho(x,\overline{A})=0$ for all $x\in A$ since $A\subseteq \overline{A}$. Similarly, each $y\in \overline{A}$ is the limit of some sequence of elements in $A$ and thus $\rho(y,A)=0$, which ensures $\rho_H(A,\overline{A})=0$. For the second statement suppose that $\rho_H(A,B)=0$, i.e. $A\subseteq B_\varepsilon$ and $B\subseteq A_\varepsilon$ for all $\varepsilon>0$. Then $A\subseteq \overline{B}=B\subseteq \overline{A}=A$ and $A=B$ as required.
    
    We have the following properties that we state without proof:
    \begin{enumerate}
        \item $\rho_H(A,B)<\infty$ if $A$ and $B$ are bounded.
        \item $(\mathfrak{C}(X),\rho_H)$ is a complete (compact) metric space if $(X,\rho)$ is complete (compact).
    \end{enumerate}
    In particular the first of these properties ensures that $\rho_H$ restricts nicely to a finite metric on the compact subsets of $X$. We are now ready to define the Gromov-Hausdorff distance:
    \begin{definition}
    Let $X$ and $Y$ be compact metric spaces. Then we define the \textit{Gromov-Hausdorff distance} between $X$ and $Y$ as
    \begin{align}
        \rho_{GH}(X,Y)=\inf \rho_H^Z(\iota_1(X),\iota_2(Y))
    \end{align}
    where the infimum is taken over all triples $(Z,\iota_1,\iota_2)$ where $Z$ is a metric space and $\iota_1$ and $\iota_2$ are isometric imbeddings of $X$ and $Y$ into $Z$ respectively, and $\rho_H^Z$ is the Hausdorff metric in $Z$.
    \end{definition}
    Recall that the infimum preserves subadditivity and note that  symmetry and positivity  of the Gromov-Hausdorff distance are trivial. Also if $X$ and $Y$ are isometric, any isometry $f:X\rightarrow Y$ defines an isometric imbedding of $X$ into $Y$ such that $\rho_H^Y(f(X),Y)=0$ and $\rho_{GH}(X,Y)=0$, i.e. the Gromov-Hausdorff metric is immediately a finite pseudometric on the \textit{Gromov-Hausdorff space} of isometry classes of compact metric spaces. In fact it can be shown that if the Gromov-Hausdorff distance between two spaces vanishes then the spaces are isometric, and the Gromov-Hausdorff distance is a finite metric---not just pseudometric---on the Gromov-Hausdorff space.
    \subsection{Gromov-Hausdorff Limits and Emergent Geometry}
    The purpose of this section is threefold: one we introduce various ways to calculate/estimate the Gromov-Hausdorff distance between two spaces; two, we discuss various ways of approaching Gromov-Hausdorff convergence. These results are used in the next section to show the convergence of classical configurations to $S^1_r$. Finally we prove two results on the convergence of finite spaces that suggests that Gromov-Hausdorff convergence is a very promising formalism for investigating questions about emergent geometry in general. Note that the material in this section is widely known, but many of the results are simply quoted and not given explicitly in standard references \cite{BuragoBuragoIvanov_MetricGeometry, Gromov_MetricStructures}. As such we have chosen to present some proofs rather explicitly.
    
    Since the Gromov-Hausdorff distance between two metric spaces $X$ and $Y$ is defined by minimising over all possible imbeddings of $X$ and $Y$ into an ambient metric space $Z$, the most obvious strategy to obtain an estimate of an upper-bound of the Gromov-Hausdorff distance is of course to construct an explicit isometric imbedding into some given metric space $Z$. In fact we immediately have the following lemma:
    \begin{lemma}\label{lemma: IsometricImbedding}
    Let $\set{X_k}_{k\in \mathbb{N}}$ and $Y$ be compact metric spaces. Given a sequence of metric spaces $\set{Z_k}_{k\in \mathbb{N}}$ and isometric imbeddings $\set{\iota_k^X:X_K\rightarrow Z_K}_{k\in \mathbb{N}}$, $\set{\iota_k^Y:Y\rightarrow Z_K}_{k\in \mathbb{N}}$, then $X_k\rightarrow Y$ in the sense of Gromov-Hausdorff if for each $\varepsilon>0$ there is an $N\in \mathbb{N}$ such that $\rho_H^{Z_n}(\iota_n^X(X),\iota_n^Y(Y))<\varepsilon$ for all $n>N$.
    \end{lemma}
    We use this lemma to show the convergence of cylinders/M\"{o}bius strips to the circle in the next section.
    
    Though this approach to Gromov-Hausdorff convergence is immediate from the definition, it is not always practical because actually constructing the required isometric imbeddings can be a little  difficult. Some reformulations of the Gromov-Hausdorff distance give better methods for estimating the Gromov-Hausdorff distance. We introduce these methods here.
    
    The first reformulation of the Gromov-Hausdorff distance is in terms of pseudometrics: 
    \begin{lemma}\label{lemma: PseudoMetricReduction}
    Let $(X,\rho_X)$ and $(Y,\rho_Y)$ be compact metric spaces.
    \begin{align}
        \rho_{GH}(X,Y)=\inf \rho_H^{X\sqcup Y}(X,Y)
    \end{align}
    where the infimum is taken over pseudometrics on $\rho:X\sqcup Y\rightarrow \mathbb{R}$ such that $\rho|X=\rho_X$ and $\rho|Y=\rho_Y$ and $X\sqcup Y$ denotes the disjoint union of $X$ and $Y$.
    \end{lemma}
    \begin{proof}
        Let $(Z,\rho_Z)$ be a metric space and suppose that $f:X\rightarrow Z$ and $g:Y\rightarrow Z$ are isometric imbeddings. We may define a pseudometric $\rho$ on $X\sqcup Y$ by $\rho|X\times X=\rho_X$, $\rho|Y\times Y=\rho_Y$ and $\rho(x,y)=\rho_Z(f(x),g(y))$ for all $(x,y)\in X\times Y$; $\rho$ obviously inherits positivity, semi-definiteness, symmetry and subadditivity from $\rho_X$, $\rho_Y$ and $\rho_Z$ but need not be definite since it is possible that $f(x)=g(y)$ for some pair $(x,y)\in X\times Y$. Thus $\rho$ is a pseudometric on $X\sqcup Y$. Thus $ \inf \rho_H^{X\sqcup Y}(X,Y)\leq \rho_{GH}(X,Y) $. On the other hand any pseudometric $\rho$ on $X\sqcup Y$ gives rise to a pair of isometric imbeddings $f=\pi\circ \iota_X$ and $g=\pi\circ \iota_Y$ of $X$ and $Y$ respectively into the induced metric space $X\sqcup Y/\rho$ where $\pi:X\sqcup Y\rightarrow X\sqcup Y/\rho$ is the quotient map and $\iota_X$ and $\iota_Y$ are the natural imbeddings of $X$ and $Y$ into $X\sqcup Y$ respectively.
    \end{proof}
    We now discuss a reformulation in terms of the \textit{distortion} of \textit{correspondences} between metric spaces. We begin by introducing the latter notion:
    \begin{definition}
    Let $X$ and $Y$ be sets. A \textit{correspondence} between $X$ and $Y$ is a relation $R\subseteq X\times Y$ such that the domain and codomain satisfy $X=\text{dom}(R)$ and $Y=\text{cod}(R)$ respectively, i.e. for each $x\in X$ there is a $y\in Y$ such that $(x,y)\in R$ and vice versa.
    \end{definition}
    \begin{proposition}\label{proposition: SurjectionCorrespondence}
    A relation $R\subseteq X\times Y$ is a correspondence between $X$ and $Y$ iff there is a set $Z$ and a pair of surjective maps $f:Z\rightarrow X$ and $g:Z\rightarrow Y$ such that $R=\set{(f(z),g(z)):z\in Z}$.
    \end{proposition}
    \begin{proof}
    For necessity suppose that $R$ is a correspondence, let $Z=R$ and let $f$ and $g$ be the projections $(x,y)\mapsto x$ and $(x,y)\mapsto y$ respectively. Then $f$ and $g$ are surjections since $R$ is a correspondence as required. Sufficiency is obvious by the surjectivity of the mappings $f$ and $g$.
    \end{proof}
    \begin{definition}
    Let $(X,\rho_X)$ and $(Y,\rho_Y)$ be metric spaces. The \textit{distortion} of a correspondence $R\subseteq X\times Y$ is defined
    \begin{align}
        \text{dis}(R)\coloneqq \sup\set{|\rho_X(x_1,x_2)-\rho_Y(y_1,y_2)|:(x_1,y_1),\:(x_2,y_2)\in R}.
    \end{align}
    \end{definition}
    \begin{corollary}
    Let $(X,\rho_X)$ and $(Y,\rho_Y)$ be metric spaces. 
    \begin{enumerate}
        \item Let $Z$ a set and let $R$ be a correspondence between $X$ and $Y$ such that $R=\set{(f(z),g(z)):z\in Z}$ for some surjections $f:Z\rightarrow Z$, $g:Z\rightarrow Y$ as per proposition \ref{proposition: SurjectionCorrespondence}. Then we have
        \begin{align}
        \text{dis}(R)= \sup_{z_1,\:z_2\in Z}|\rho_X(f(z_1),f(z_2))-\rho_Y(g(z_1),g(z_2))|.
        \end{align}
        \item $\text{dis}(R)=0$ iff there is an isometry $f:X\rightarrow Y$ such that $R=\set{(x,f(x)):x\in X}$.
    \end{enumerate}
    \end{corollary}
    \begin{proof}
    (i) is trivial. For (ii) note that $\text{dis}(R)$ is the supremum of a positive quantity so it vanishes iff $|\rho_X(x_1,x_2)-\rho_Y(y_1,y_2)|=0$ for all $(x_1,y_1),\:(x_2,y_2)\in R$. Thus for any pair of points $(x_1,x_2)\in X$ we must have $\rho_Y(y_1,y_2)=\rho_X(x_1,x_2)$ for any $y_k\in Y_{x_k}\coloneqq \set{y\in \text{cod}(R):(x_k,y)\in R}$, $k\in \set{1,2}$. We show that $R$ is the graph of a unique bijection $f:x\mapsto y$, i.e. $f(x)=y$ iff $(x,y)\in R$: in particular $\rho_X(x,x)=0$ trivially for all $x\in X$ so $\rho_Y(y_1,y_2)=0$ for all $y_1,\:y_2\in Y_x$ by the preceding claim. But then $y_1=y_2$ by definiteness and $f$ is a well-defined bijection such that $\rho_Y(f(x_1),f(x_2))=\rho_X(x_1,x_2)$ for all $x_1,\:x_2\in X$ and $f$ is an isometry as required. The converse is trivial.
    \end{proof}
    We now show that the Gromov-Hausdorff distance is (up to a constant) the infimum of the distortion of correspondences between the relevant metric spaces:
    \begin{theorem}\label{theorem: CorrespDist}
    Let $(X,\rho_X)$ and $(Y,\rho_Y)$ be compact metric spaces. Then
    \begin{align}
        \rho_{GH}(X,Y)=\frac{1}{2}\inf_{R}\text{dis}(R)
    \end{align}
    where the infimum is taken over all correspondences between $X$ and $Y$.
    \end{theorem}
    \begin{proof}
    We show (i) that if $\rho_{GH}(X,Y)<r$ then there is a correspondence between $X$ and $Y$ such that $\text{dis}(R)<2r$ and (ii) that $2\rho_{GH}(X,Y)\leq \text{dis}(R)$ for any correspondence $R$ between $X$ and $Y$.
    \begin{enumerate}
        \item Given $\rho_{GH}(X,Y)<r$ we may assume without loss of generality that $X$ and $Y$ are subspaces of $(Z,\rho_Z)$ with $\rho_Z^H(X,Y)<r$; then defining $R=\set{(x,y)\in X\times Y:\rho_Z(x,y)<r}$ gives $R$ a correspondence. Then $\text{dis}(R)<2r$ by the triangle inequality:
        \begin{align}
            \text{dis}(R)\leq |\rho_Z(x_1,y_1)-\rho_Z(x_2,y_2)|\leq \rho_Z(x_1,y_1)+\rho_Z(x_2,y_2)<2r\nonumber
        \end{align}
        for any $x_1,\:x_2\in X$ and $y_1,\:y_2\in Y$.
        \item Choose some correspondence $R$ and define $r\coloneqq \frac{1}{2}\text{dis}(R)$. By lemma \ref{lemma: PseudoMetricReduction} it is sufficient to provide a pseudometric $\rho$ on $X\sqcup Y$ such that $\rho_H^{X\sqcup Y}(X,Y)\leq r$ since then $\rho_{GH}(X,Y)\leq \rho_H^{X\sqcup Y}(X,Y)\leq \frac{1}{2}\text{dis}(R)$. Let us define $\rho$ via $\rho|X\times X=\rho_X$, $\rho|Y\times Y=\rho_Y$ and
        \begin{align}
            \rho(x,y)=\inf \set{r+\rho_X(x,x_1)+\rho_Y(y_1,y):(x_1,y_1)\in R}.\nonumber
        \end{align}
        Symmetry and positive semi-definiteness are trivial---note that $\rho(x,y)\geq r$ but $x\neq y$ for $x\in X$ and $y\in Y$ since $X\sqcup Y$ is the \textit{disjoint} union of $X$ and $Y$. For subadditivity we note that we need to check that $\rho(x,y)\leq \rho(x,z)+\rho(z,y)$ for $x\in X$, $y\in Y$ and $z\in X\cup Y$; subadditivity is a trivial consequence of the subadditivity of $\rho_X$ and $\rho_Y$ in the cases where $x,\:y,:z\in X$ or $x,\:y,\:z\in Y$ respectively. Suppose that $z\in X$. Then
        \begin{align}
            \rho(x,y)\leq r+\rho_X(x,x_1)+\rho_Y(y_1,y)\leq \rho_X(x,z)+(r+\rho_X(z,x_1)+\rho_Y(y_1,y)\nonumber
        \end{align}
        for all $(x_1,y_1)\in R$, where we have used the subadditivity of $\rho_X$ to write $\rho_X(x,x_1)\leq \rho_X(x,z)+\rho_X(z,x_1)$. Taking the infimum over all pairs $(x_1,y_1)\in R$ on both sides gives the desired result. A similar argument holds if $z\in Y$. Thus $\rho$ is a pseudometric on $X\sqcup Y$. 
        
        It remains to show that the Hausdorff distance associated to $\rho$ is bounded above by $r$: to see this note that every point of $X$ lies within a distance $r$ of some point of $Y$ and vice versa. In particular, given some $x\in X$ we have some $y\in Y$ such that $(x,y)\in R$ since $R$ is a correspondence so $r\leq \rho(x,y)\leq r+\rho_X(x,x)+\rho_Y(y,y)=r$ and $X\subseteq Y_r$, the $r$-thickening of $Y$. A similar argument shows $Y\subseteq X_r$ so $\rho_H(X,Y)\leq r$ as required.
    \end{enumerate}
    \end{proof}
    We consider one final method to estimate the Gromov-Hausdorff distance. This is in terms of \textit{nearly} isometric imbeddings:
    \begin{definition}
        Let $(X,\rho_X)$ and $(Y,\rho_Y)$ be metric spaces. 
        \begin{enumerate}
            \item The \textit{distortion} of a mapping $f:X\rightarrow Y$ is defined
            \begin{align}
                \text{dis}(f)\coloneqq \sup\set{|\rho_Y(f(x),f(y))-\rho_X(x,y))|:(x,y)\in X^2}.
            \end{align}
            \item A mapping $f:X\rightarrow Y$ is said to be a \textit{$\varepsilon$-isometry} iff $f(X)$ is a $\varepsilon$-net in $Y$ and $\text{dis}(f)<\varepsilon$.
        \end{enumerate}
    \end{definition}
    \begin{lemma}\label{lemma: NearIsometry}
        Let $X$ and $Y$ be compact metric spaces. If $\rho_{GH}(X,Y)<\varepsilon$ then there is a $2\varepsilon$-isometry $f:X\rightarrow Y$. Similarly if there is a $\varepsilon$-isometry $f:X\rightarrow Y$ then $\rho_{GH}(X,Y)<\frac{3}{2}\varepsilon$.
    \end{lemma}
    \begin{proof}
        Suppose that $\rho_{GH}(X,Y)<\varepsilon$; then by theorem \ref{theorem: CorrespDist} there is a correspondence $R\subseteq X\times Y$ such that $\text{dis}(R)<2\varepsilon$. We may construct a mapping $f:X\rightarrow Y$ as a choice function $f:X\rightarrow \bigsqcup_{x\in X}\set{y\in Y:(x,y)\in R}$. Clearly $\text{dis}(f)\leq \text{dis}(R)<2\varepsilon$. To see that $f(X)$ is a $2\varepsilon$-net in $Y$, for each $y\in Y$ choose an $x\in X$ such that $(x,y)\in R$. Then
        \begin{align}
            \rho_Y(f(x),y)&=\rho_Y(f(x),y)- \rho_X(x,x)+\rho_X(x,x)\nonumber\\
            &\leq |\rho_Y(f(x),y)- \rho_X(x,x)|+\rho_X(x,x)\nonumber\\
            &\leq \rho_X(x,x)+\text{dis}(R)\nonumber\\
            &<2\varepsilon\nonumber
        \end{align}
        as required. Now suppose that $f:X\rightarrow Y$ is a $\varepsilon$-isometry and construct $R\subseteq X\times Y$ by
        \begin{align}
            R=\set{(x,y)\in X\times Y:\rho_Y(f(x),y)<\varepsilon}.\nonumber
        \end{align}
        Clearly $(x,f(x))\in R$ and $\text{dom}(R)=X$; also since $f(X)$ is a $\varepsilon$-net in $Y$, each $y\in Y$ lies within $\varepsilon$ of $f(x)$ for some $x\in X$ and $\text{cod}(R)=Y$, i.e. $R$ is a correspondence as required. Now for any pairs $(x_1,y_1)$ and $(x_2,y_2)\in R$ we have
        \begin{align}
            |\rho_Y(y_1,y_2)-\rho_X(x_1,x_2)|\leq \rho_Y(y_1,f(x_1))+\rho_Y(f(x_2),y_2)+|\rho_Y(f(x_1),f(x_2))-\rho_X(x_1,x_2)|\leq 2\varepsilon +\text{dis}(f)<3\varepsilon
        \end{align}
        where we have used subadditivity in the first step and the definition of $R$ and $\text{dis}(f)$ in the second, while the final step follows since $f$ is an $\varepsilon$-isometry and $\text{dis}(f)<\varepsilon$. Taking the supremum of both sides over all pairs $(x_1,y_2),\:(x_2,y_2)\in R$ gives $\text{dis}(R)<3\varepsilon$ so the required result follows immediately from theorem \ref{theorem: CorrespDist}
    \end{proof}
    This gives an essentially complete set of methods for estimating the Gromov-Hausdorff distance. We now turn back to more explicit convergence properties.
    
    We shall show one basic lemma that demonstrates that Gromov-Hausdorff convergence reduces to convergence of finite subsets:
    \begin{definition}
        Two compact metric spaces $X$ and $Y$ are said to be \textit{$(\varepsilon,\delta)$-approximations} of one another iff we have finite $\varepsilon$-nets $A\subseteq X$ and $B\subseteq Y$ of cardinality $N$ such that $|\rho_X(x_k,x_{\ell})-\rho_Y(y_k,y_{\ell})|<\delta$ for all $k,\:\ell\in \set{0,...,N-1}$.
    \end{definition}
    \begin{lemma}
        Let $X$ and $Y$ be compact metric spaces. If $X$ is an $(\varepsilon,\delta)$-approximation of $Y$ then $\rho_{GH}(X,Y)<2\varepsilon+\frac{1}{2}\delta$. Similarly if $\rho_{GH}(X,Y)<\varepsilon$ then $X$ is a $5\varepsilon$-approximation of $Y$.
    \end{lemma}
    \begin{proof}
        Suppose that $X$ is an $(\varepsilon,\delta)$-approximation of $Y$; then we have $\varepsilon$-nets $A\subseteq X$ and $B\subseteq Y$ of cardinality $N$ such that such that $|\rho_X(x_k,x_{\ell})-\rho_Y(y_k,y_{\ell})|<\delta$ for all $k,\:\ell\in \set{0,...,N-1}$. by definition. By the latter property, the correspondence $R=\set{(x_k,y_k)\in A\times B:k<N}$ has distortion $\text{dis}(R)<\delta$, i.e. $\rho_{GH}(A,B)<\frac{1}{2}\delta$. Then by subadditivity
        \begin{align}
            \rho_{GH}(X,Y)\leq \rho_{GH}(X,A)+\rho_{GH}(A,B)+\rho_{GH}(B,Y)<2\varepsilon+\frac{1}{2}\delta.\nonumber
        \end{align}
        Now suppose that $\rho_{GH}(X,Y)<\varepsilon$ and let $A=\set{x_n}_{n<N}$ be a finite $\varepsilon$-net in $X$. By lemma \ref{lemma: NearIsometry} we have a $2\varepsilon$-isometry $f:X\rightarrow Y$; define $B=f(A)=\set{f(x_k)}_{k<N}$. Now
        \begin{align}
            |\rho_Y(f(x_k),f(x_{\ell}))-\rho_X(x_k,x_{\ell})|\leq \text{dis}(f)<2\varepsilon \nonumber
        \end{align}
        since $f$ is a $2\varepsilon$-isometry so it is sufficient to show that $B$ is a $5\varepsilon$-net in $Y$. In particular, since $f(X)$ is a $2\varepsilon$-net in $Y$, for each $y\in Y$, there is an $x\in X$ such that $\rho_Y(y,f(x))<2\varepsilon$, while there is an $x_k\in A$ such that $\rho_X(x,x_k)<\varepsilon$ since $A$ is a $\varepsilon$-net in $X$. Thus for some $k\in \set{0,...,N-1}$ we have
        \begin{align}
            \rho_Y(y,f(x_k))\leq \rho_Y(y,f(x))+\rho_y(f(x),f(x_k))\leq \rho_Y(y,f(x))+\text{dis}(f)+\rho_X(x,x_k)<5\varepsilon\nonumber.
        \end{align}
        Since $f(x_k)\in B$ this proves the statement.
    \end{proof}
    We may now prove two central results from the perspective of emergent geometry. First note that if $A$ is a $\varepsilon$-net in $X$ then $\rho_H(A,X)<\varepsilon$.
    \begin{theorem}
        \leavevmode
        \begin{enumerate}
            \item Every compact metric space is the Gromov-Hausdorff limit of sequence of finite metric spaces.
            \item Every compact length space is the Gromov-Hausdoff limit of a finite graph.
        \end{enumerate}
    \end{theorem}
    \begin{proof}
        Let $X$ be a compact metric space.
        \begin{enumerate}
            \item For each $n\in \mathbb{N}^+$ take a sequence $\varepsilon_n\rightarrow 0$ and choose a finite $\varepsilon_n$-net $A_n\subseteq X$; note that it is possible to choose a finite $\varepsilon_n$-net for each $n$ because $X$ is compact. Then $\rho_{GH}(A_n,X)\leq \rho_{H}(A_n,X)<\varepsilon_n$ for each $n$ and we have the desired result i.e. $X$ is the Gromov-Hausdorff limit of the spaces $A_n$.
            \item See the proof of proposition 7.5.5 in \cite{BuragoBuragoIvanov_MetricGeometry}.
        \end{enumerate}
    \end{proof}
    \subsection{Gromov-Hausdorff Convergence of Classical Configurations}
    The aim of this section is to apply the general convergence results in the preceding section in order to show that $S^1_r$ is the Gromov-Hausdorff limit of classical configurations of some statistical model. Let us warm-up by showing that---appropriately scaled---cylinders and M\"{o}bius strips converge to the circle:
    \begin{definition}
    For each $n\in \mathbb{N}^+$, let $\text{Cyl}_n=S^1_{r}\times[-\ell(n)/2,\ell(n)/2]$ where $S^1_r$ is the circle of radius $r>0$ where $r$ is some fixed radius and $\ell(n)$ is defined as:
    \begin{align}\label{equation: Length}
        \ell(n)=\frac{2\pi r}{n}.
    \end{align}
    Also let $\text{Mob}_n$ denote the M\"{o}bius strip obtained by gluing the strip $[0,\ell(n)n]\times [0,\ell(n)]$ along the boundary lines $(0,t)\cong (\ell(n)n,\ell(n)-t)$, $t\in [0,\ell(n)]$. The space $\text{Cyl}_n$ may be regarded as Riemannian manifolds (with boundary) and thus as metric spaces when equipped with a metric defined via the line element
    \begin{align}
        \text{d} s^2=r^2\text{d}\theta^2+\text{d} z^2,
    \end{align}
    where $(\theta,z)\in (-\pi,\pi]\times [-\ell(n)/2,\ell(n)/2]$. Since M\"{o}bius strips and cylinders are locally identical, \emph{mutatis mutandis} the same holds for the spaces $\text{Mob}_n$.
    \end{definition}
    \begin{proposition}\label{proposition: ConvergenceCyl}
        Let $\set{\mathcal{M}_n}_{n\in \mathbb{N}^+}$ denote a sequence of metric spaces such that $\mathcal{M}_n\in \set{\text{Cyl}_n,\text{Mob}_n}$ for each $n\in \mathbb{N}^+$.
        Then
        \begin{align}
            \lim_{n\rightarrow \infty}\mathcal{M}_n=S^1_r
        \end{align}
        where convergence is in the sense of Gromov-Hausdorff.
    \end{proposition}
    \begin{proof}
    We use lemma \ref{lemma: IsometricImbedding}; in particular, if we imbed $\iota:S^1_r\hookrightarrow\mathcal{M}_n$ as the central circle then we have
    \begin{align}\label{equation: GHDistCyl}
        \rho_{GH}(S^1_r,\mathcal{M}_n)\leq \rho_H(\iota(S^1_r),\mathcal{M}_n)=\frac{1}{2}\ell(n)=\frac{\pi r}{n}
    \end{align}
    since $\iota(S^1_r)\subseteq \mathcal{M}_n$ is an $R$-net in $\mathcal{M}_n$ for all $R>\ell(n)/2$ but the $(\ell(n)/2)$-thickening of $\iota(S^1_r)$ does not cover (the boundary of) $\mathcal{M}_n$. Thus if we choose
    \begin{align}
        N= \left\lceil\frac{\pi r}{\varepsilon}\right\rceil\nonumber
    \end{align}
    for any $\varepsilon>0$, this ensures that
    \begin{align}
        \frac{1}{2}\ell(n)=\frac{\pi r}{n}<\varepsilon\nonumber
    \end{align}
    for all $n>N$ as required.
    \end{proof}
    This of course formalises the naive intuition that as the width of a cylinder/M\"{o}bius strip vanishes we end up with a circle. 
    
    We have essentially the same result for prism graphs and M\"{o}bius ladders:
    \begin{lemma}\label{lemma: ConvergenceCubicGraphs}
    Let $\set{\omega_n}_{n>3}$ be a sequence of graphs such that $\omega_n\in \set{P_n^{\ell(n)},M_n^{\ell(n)}}$ for each value of $n$, where $\ell(n)$ is given as in equation \ref{equation: Length} and $P_n^{\ell(n)}$ and $M_n^{\ell(n)}$ are the graphs $P_n$ and $M_n$ respectively, with each edge weighted $\ell(n)$. The Gromov-Hausdorff limit of this sequence is $S^1_r$.
    \end{lemma}
    \begin{proof}
    We have
    \begin{align}
        \rho_{GH}(\omega_n,S_r^1)\leq \rho_{GH}(\omega_n,\mathcal{M}_n)+\rho_{GH}(\mathcal{M}_n,S_r^1)=\frac{\pi r}{n}+\rho_{GH}(\omega_n,\mathcal{M}_n)\nonumber
    \end{align}
    where $\mathcal{M}_n$ is $\text{Cyl}_n$ if $\omega_n=P_n^{\ell(n)}$ and $\mathcal{M}_n=\text{Mob}_n$ if $\omega_n=M_n^{\ell(n)}$; to obtain the right-hand side we have applied subadditivity and used equation \ref{equation: GHDistCyl} calculated in the course of the proof of proposition \ref{proposition: ConvergenceCyl}. It thus remains to bound $\rho_{GH}(\omega_n,\mathcal{M}_n)$; consider the natural imbedding $\iota:\omega_n\hookrightarrow \mathcal{M}_n$  of $\omega_n$ into the boundary. For any two points $u,\:v\in \omega_n$ that lie in the same circle $S^1_r\times \set{-\frac{1}{2}\ell(n)}$ or $S^1_r\times \set{\frac{1}{2}\ell(n)}$, the distance in both graphs is simply given by the length of the shorter circle segment connecting the two points and $\rho_{\omega_n}(u,v)=\rho_{\mathcal{M}_n}(\iota(u),\iota(v))$. However if  $u$ and $v$ lie in distinct circles, i.e. $u\in S^1_r\times \set{-\frac{1}{2}\ell(n)}$ and $v\in S^1_r\times \set{\frac{1}{2}\ell(n)}$ we have a graph distance $\rho_{\omega_n}(u,v)=\ell(n)+r\theta$ where $\theta$ is the smaller angle between $u$ and $v$, while
    \begin{align}
        \rho_{\mathcal{M}_n}(\iota(u),\iota(v))=\sqrt{\ell(n)^2+r^2\theta^2}\nonumber.
    \end{align}
    Thus
    \begin{align}
        |\rho_{\mathcal{M}_n}(\iota(u),\iota(v))-\rho_{\omega_n}(u,v)|=(\ell(n)+r\theta)-\sqrt{\ell(n)^2+r^2\theta^2}\nonumber.
    \end{align}
    Maximising over pairs $u,\:v\in \omega_n$ gives
    \begin{align}
        \sup |\rho_{\mathcal{M}_n}(\iota(u),\iota(v))-\rho_{\omega_n}(u,v)|\leq \pi r\left(1-\sqrt{1+\frac{\ell(n)}{r\pi}} \right) +\ell(n)=\pi r\left(1-\sqrt{1+\frac{2}{n}} \right) +\frac{2\pi r}{n}\nonumber.
    \end{align}
    For $n>2$ we may use the binomial expansion to obtain
    \begin{align}
        \sup |\rho_{\mathcal{M}_n}(\iota(u),\iota(v))-\rho_{\omega_n}(u,v)|=\frac{\pi r}{n}+\sum_{k=2}^\infty \frac{1\cdot \left(1-2\right)\cdots \left(1-2(k-1)\right)}{k!}n^{-k}\nonumber.
    \end{align}
    Thus
    \begin{align}
        \rho_{GH}(\omega_n,S_r^1)\leq\frac{2\pi r}{n}+\sum_{k=2}^\infty \frac{1\cdot \left(1-2\right)\cdots \left(1-2(k-1)\right)}{k!}n^{-k}\nonumber
    \end{align}
    for all $n>2$. The right-hand side converges to $0$ as $n\rightarrow \infty$ which shows that $\rho_{GH}(\omega_n,S_r^1)\rightarrow 0$ as $n\rightarrow \infty$ which proves the statement. 
    \end{proof}
    We thus have the essential convergence result viz. any sequence of prism graphs/M\"{o}bius ladders converges to $S^1$. It remains to show that these graphs constitute the classical (action minimising) configurations. This is immediately achieved if we exclude triangles:
    \begin{lemma}\label{lemma: ClassicalConfigurations}
    Let $(\Omega^4(2N),\mathcal{A})$ be the statistical model where $\Omega^4(2N)$ is the class of $3$-regular graphs on $2N$ vertices of girth at least $4$ and $\mathcal{A}$ is the discrete Einstein-Hilbert action \ref{equation: Action}. Then the classical phase $\Omega_0^4(2N)=\Omega_0\cap \Omega^4(2N)=\set{P_N,M_N}$.
    \end{lemma}
    \begin{proof}
        Note that an edge has strictly positive curvature iff it supports a triangle by expression \ref{equation: OllivCurvCubic} (whence it has curvature $1/3$ or $2/3$). Thus if the girth $g(\omega)>3$ there are no triangles in $\omega$ and the maximum total curvature is thus $0$; this is achieved iff $\omega$ is Ricci flat for $N>3$. Then the result follows by the classification of Cushing et al.
    \end{proof}
    \begin{remark}
    Note that triangles must be effectively excluded in some manner due to the examples displayed in figure \ref{figure: TotalCurvatureTriangle}. Rather than insisting that the graphs have girth greater than $3$, a bipartite model would also suffice and is perhaps a little more natural. In both cases we have suppressed triangles kinematically for certain, whereas ideally this would arise dynamically with high probability. In such a scenario, however, any precise results are likely to require significantly more involved methods for analogous results to be shown.
    \end{remark}
    The two lemmas in conjunction immediately give the following result:
    \begin{theorem}\label{theorem: Convergence}
        Let $\Omega^4(2n;\ell(n))$ denote the set of all $3$-regular graphs on $2n$ vertices with girth at least $4$ and with uniform edge weight $\ell(n)$. Also let $\set{\omega_n}$ denote a sequence of classical configurations of the statistical models $\set{(\Omega^4(2n;\ell(n)),\mathcal{A})}$ respectively. Then we have the Gromov-Hausdorff limit $\omega_n\rightarrow S^1_r$ as $n\rightarrow \infty$. $\qed$
    \end{theorem}
    This theorem essentially shows that for large $N$ and large $\beta$ the Gibbs distribution for our model is concentrated on the circle $S^1_r$.
    \subsection{General Covariance and Gromov-Hausdorff Limits}
    We finish with some comments on the generally covariant nature of Gromov-Hausdorff limits. The main point has already been discussed in the introduction viz. Gromov-Hausdorff convergence characterises the limit invariantly essentially because one minimises over all possible metric backgrounds. However there are two technical points to bear in mind insofar as gravitational gauge transformations---typically and somewhat loosely referred to as \textit{diffeomorphisms} in the physics literature---are generally interpreted as (local) Riemannian isometries connected to the identity. Firstly, it is not immediately obvious mathematically speaking that diffeomorphisms so characterised are in fact isometries with respect to the topological metric (geodesic distance); physically speaking the statement is rather obvious insofar as the geodesic distance between two point particles is essentially the classical action of a point particle along its trajectory and thus should be invariant under gauge transformations. Nonetheless the mathematical result is both nontrivial: the best argument has two steps. First one notes that local Riemannian isometries are also local metric (topological) isometries in the sense that they restrict to metric isometries on suitably chosen open balls; this follows by the naturality of the exponential map. The next step is to show that a bijective local metric isometry is a global isometry for arbitrary length spaces as long \textit{as long as the local isometry admit a continuous inverse}. This is obviously the case for elements of the diffeomorphism group. If the local isometry does not have a continuous inverse then one can construct counterexamples showing that a local bijective isometry is not necessarily a global isometry. The second point to note is that the Gromov-Hausdorff limit does not simply factor out gauge transformations, it also factors out global symmetries, i.e. the group of connected components of the space of local Riemannian isometries. It is not clear to what extent these global symmetries apply to graphs in the limiting sequence.
    
    We shall briefly substantiate the claims relating to the first point. We shall use the following terminology:
    \begin{definition}
        \leavevmode
        \begin{enumerate}
            \item Let $(\mathcal{M},g)$ and $(\mathcal{N},h)$ be Riemannian manifolds. A smooth mapping $f:\mathcal{M}\rightarrow \mathcal{N}$ is said to be a \textit{Riemannian isometry} iff for each $p\in \mathcal{M}$ and all $u,\:v\in T_p\mathcal{M}$ we have $g_p(u,v)=h_{f(p)}(f_*u, f_*v)$. Let $\text{Isom}_g(\mathcal{M})$ denote the group of (Lie) group of local isometries equipped naturally with the compact-open topology. A Riemannian isometry $f:\mathcal{M}\rightarrow \mathcal{M}$ is a \textit{gauge transformation} iff it is in the connected component of the identity of $\text{Isom}_g(\mathcal{M})$. 
            \item Let $(X,\rho_X)$ and $(Y,\rho_Y)$ be metric spaces. A \textit{local isometry} is a continuous mapping $f:X\rightarrow Y$ such that for each $p\in X$ there is an $r>0$ such that $f|B_r(p):B_r(p)\rightarrow f(B_r(p))=B_r(f(p))$ is an isometry.
            \item A \textit{global isometry} or \textit{metric isometry} is simply an isometry of metric spaces.
        \end{enumerate}
    \end{definition}
    We recall the definition of the exponential map in a Riemannian manifold:
    \begin{definition}
        Let $(\mathcal{M},g)$ be a Riemannian manifold and fix a point $p\in \mathcal{M}$. We define the \textit{exponential map} $\exp_p:U\rightarrow \mathcal{M}$ for each $u$ in some subset $U\subseteq T_p\mathcal{M}$ by letting $\exp_p(u)=\gamma_u(1)$, where $\gamma_u$ is the unique geodesic such that $\gamma_u(0)=p$ and $\dot{\gamma}_u=u$.
    \end{definition}
    $\exp_p$ restricts to a smooth mapping on $B_r(0)$ into $B_r(p)$ for $r>0$ sufficiently small. Naturality of the exponential map then takes the following form for any local isometry $f:\mathcal{M}\rightarrow \mathcal{M}$: the diagram
    \begin{equation}
        \begin{tikzcd}
            U\subseteq T_p\mathcal{M}\arrow{r}{f_*} \arrow[swap]{d}{\exp_p} & f_*U\subseteq T_{f(p)}\mathcal{M}\arrow{d}{\exp_{f(p)}}\\
            \mathcal{M}\arrow[swap]{r}{f} & \mathcal{M}
        \end{tikzcd}
    \end{equation}
    commutes for each $p\in \mathcal{M}$. Since $f$ is a local isometry, the mapping $f_*|B_r(0_p):B_r(0_p)\rightarrow f_*(B_r(0))=B_r(0_{f(p)})$ is an isometry, so $f|B_r(p)=\exp_{f(p)}\circ f_*\circ \exp_p^{-1}|B_r(p)$. Noting that local isometries also preserve geodesics is enough to ensure that $f|B_r(p):B_r(p)\rightarrow B_r(f(p))$ is an isometry. The upshot is the following:
    \begin{proposition}
    Let $f:\mathcal{M}\rightarrow \mathcal{M}$ be a Riemannian isometry of the Riemannian manifold $(\mathcal{M},g)$. Then $f$ is a local isometry.
    \end{proposition}
    We now turn to local isometries in length spaces. Recall that a length space is a metric space in which the distance is equal to the infimum over lengths of a class of admissible curves. Both graphs and Riemannian manifolds are length spaces. More precisely, for any metric space $(X,\rho_X)$, one considers a family $\mathscr{C}$ of \textit{admissible curves} $\gamma:[a,b]\rightarrow X$, where $\gamma$ is at least piecewise continuous. The \textit{length} of any admissible curve $\gamma\in \mathscr{C}$ is then defined by:
    \begin{align}
        L(\gamma)=\sup \sum_{k=0}^{N-1}\rho_X(\gamma(t_k),\gamma(t_{k+1}))
    \end{align}
    where the supremum is taken over finite partitions $a=t_0< t_1< \cdots < t_N=b$ of the interval $[a,b]$. Together with the class $\mathscr{C}$ of admissible curves the function $L$ defines a \textit{length structure} on $X$; associated to any length structure is an induced metric $\rho_L:X\times X\rightarrow \mathbb{R}$ defined by
    \begin{align}
        \rho_L(x,y)=\inf_{\gamma\in \mathscr{C}(x,y)} L(\gamma) && \mathscr{C}(x,y)=\set{\gamma\in \mathscr{C}:\gamma(a)=x \text{ and }\gamma(b)=y}.
    \end{align}
    We call $(X,\rho_X)$ a \textit{length space} for some class of admissible curves $\mathscr{C}$ when $\rho_X=\rho_L$. 
    
    We show that between length spaces, every local isometry preserves the length of curves:
    \begin{proposition}
        Let $(X,\rho_X)$ and $(Y,\rho_Y)$ be length spaces and let $f:X\rightarrow Y$ be a local isometry such that each admissible curve in $X$ is mapped to an admissible curve in $Y$.\footnote{Assuming that $X$ is a topological space, $Y$ is a length space and that $f$ is a local homeomorphism there is a unique length structure on $X$ making $f$ into a local isomorphism mapping admissible curves to admissible curves. Thus our assumptions present no real restriction on possible length spaces.} Then $L_X(\gamma)=L_Y(f\circ \gamma)$ for any admissible curve $\gamma:[a,b]\rightarrow X$.
    \end{proposition}
    \begin{proof}
        Fix some $\varepsilon>0$ and choose a partition $\set{t_k}_{k\leq N}$ of the interval $[a,b]$ such that $f|(t_k,t_{k+1})$ is an isometry onto its image for each $k\in \set{0,...,N-1}$ and such that
        \begin{align}
            \sum_k \rho_X(\gamma(t_k),\gamma(t_{k+1}))&\leq L_X(\gamma)< \sum_k \rho_X(\gamma(t_k),\gamma(t_{k+1}))+\varepsilon \nonumber\\
            \sum_k \rho_Y(f(\gamma(t_k)),f(\gamma(t_{k+1})))&\leq L_Y(f\circ \gamma)< \sum_k \rho_Y(f(\gamma(t_k)),f(\gamma(t_{k+1})))+\varepsilon\nonumber.
        \end{align}
        This is possible because $f$ is a local isometry and the length of an admissible curve is given by the supremum over partitions of the interval. Then
        \begin{align}
            L_Y(f\circ \gamma)-\varepsilon&<\sum_k \rho_Y(f(\gamma(t_k)),f(\gamma(t_{k+1})))\nonumber\\
            &=\sum_k \rho_X(\gamma(t_k),\gamma(t_{k+1}))\nonumber\\
            &\leq L_X(\gamma)\nonumber\\
            &< \sum_k \rho_X(\gamma(t_k),\gamma(t_{k+1}))+\varepsilon\nonumber\\
            &=\sum_k \rho_Y(f(\gamma(t_k)),f(\gamma(t_{k+1})))+\varepsilon\nonumber\\
            &\leq L_Y(f\circ \gamma)+\varepsilon\nonumber.
        \end{align}
        Since this holds for all $\varepsilon>0$ we have $L_X(\gamma)=L_Y(f\circ \gamma)$ as required.\nonumber
    \end{proof}
    In general a local isometry is not a global isometry.
    \begin{example}
        The mapping $f:\mathbb{R}\rightarrow S^1$ given by
        \begin{align}
            f:t \mapsto (\cos (t),\sin (t))\nonumber
        \end{align}
        is a local but not global isometry.
    \end{example}
    \begin{proof}
        $f$ obviously restricts to an isometry on any interval $(a,b)$ such that $b-a<\pi$, since $\rho_{\mathbb{R}}(x,y)=|x-y|<\pi$ for any $x,\:y\in (a,b)$ while $\rho_{S^1}(x,y)$ is equal to $r\theta$, where $r=1$ is the radius of $S^1$ and $\theta$ the smaller angle between the points $f(x)$ and $f(y)$. Since $|x-y|<\pi$, $\theta=|x-y|$. $f$ is obviously not a global isometry since $\rho_{S^1}(f(0),f(2\pi))=0$.
    \end{proof}
    We can show the following, however:
    \begin{proposition}
        Let $(X,\rho_X)$ and $(Y,\rho_Y)$ be length spaces.
        \begin{enumerate}
            \item Let $f:X\rightarrow Y$ be a local isometry preserving admissible curves. The $f$ is \textit{nonexpanding}, i.e. $\rho_X(x,y)\leq \rho_Y(f(x),f(y))$ for all $x,\:y\in X$.
            \item A bijective local isometry preserving admissible curves $f:X\rightarrow Y$ is a global isometry iff $f^{-1}:Y\rightarrow X$ is a local isometry preserving admissible curves. 
        \end{enumerate}
    \end{proposition}
    \begin{proof}
        \leavevmode
        \begin{enumerate}
            \item Since every admissible curve in $X$ is mapped to an admissible curve in $Y$ by $f$ and since $\rho_X$ and $\rho_Y$ are given by infima over admissible curves of the lengths of the curves, the fact that $f$ preserves the length of curves immediately guarantees this statement.
            \item If $f$ is a global isometry, $\rho_X(x_1,x_2)=\rho_Y(y_1,y_2)$ where $y_1=f(x_1)$ and $y_2=f(x_2)$ for all $x_1,\:x_2\in X$ so $\rho_Y(y_1,y_2)=\rho_X(f^{-1}(y_1),f^{-1}(y_2))$ for all $y_1,\:y_2\in Y$. Then $f^{-1}$ is a global isometry and thus a local isometry. Conversely, suppose that both $f$ and $f^{-1}$ are local isometries preserving admissible curves. Then since $f$ and $f^{-1}$ are both nonexpanding we have $\rho_X(x,y)\leq \rho_Y(f(x),f(y))\leq \rho_X(f^{-1}(f(x)),f^{-1}(f(y)))=\rho_X(x,y)$ and $f$ is a global isometry as required.
        \end{enumerate}
    \end{proof}
    We end with an essential condition for $f^{-1}$ to be a local isometry:
    \begin{proposition}
        Let $(X,\rho_X)$ and $(Y,\rho_Y)$ be length spaces and let $f:X\rightarrow Y$ be a bijective local isometry. $f^{-1}:Y\rightarrow X$ is a local isometry iff it is continuous.
    \end{proposition}
    \begin{proof}
        Local isometries are continuous by definition so necessity is trivial. For sufficiency note that since $f^{-1}$ is continuous, $f(U)$ is open for any open set $U$ of $X$; in particular if $f$ is an isometry when restricted to some neighbourhood $U$ of $p\in X$, $f^{-1}$ restricts to an isometry on the open neighbourhood $f(U)$ of $f(p)$ making $f^{-1}$ a local isometry.
    \end{proof}
    \begin{corollary}
        Every gauge transformation $f:\mathcal{M}\rightarrow \mathcal{M}$ is a global isometry. 
    \end{corollary}
    \printbibliography
\end{document}